\documentclass[11pt,a4paper]{article}
\usepackage{amsmath,amsfonts,amssymb,mathtools,amsthm,newlfont,graphicx,amscd,bbm,enumerate,galois,mathrsfs,xypic,geometry,bbm,xypic}
\usepackage{simplewick}
\usepackage{booktabs}
\usepackage{multirow}
\usepackage{bigstrut}
\usepackage{tabularx,extarrows}
\usepackage{tikz}
\usetikzlibrary{cd}

\usepackage{hyperref}

\usepackage{color}

\newcommand*{\llim}{\varprojlim}                
\newcommand*{\fps}[1]{[\![#1]\!]}        
\newcommand*{\fls}[1]{(\!(#1)\!)}        

\newcommand{\lmd}{\lambda}

\newcommand{\Lmd}{\Lambda}

\newcommand{\p}{\partial}

\newcommand{\afa}{\alpha}

\newcommand{\veps}{\varepsilon}
\newcommand{\fai}{\varphi}

\newcommand{\rme}{{\mathrm e}}

\newcommand*{\mfkg}{\mathfrak{g}}

\newcommand{\mcalA}{\mathcal{A}}

\newcommand{\mcalF}{\mathcal{F}}

\newcommand{\mcalI}{\mathcal{I}}

\newcommand{\mcalL}{\mathcal{L}}
\newcommand{\mcalM}{\mathcal{M}}
\newcommand{\mcalP}{\mathcal{P}}

\newcommand{\mcalU}{\mathcal{U}}

\newcommand{\bbC}{\mathbb{C}}

\newcommand{\bbR}{\mathbb{R}}
\newcommand{\bbZ}{\mathbb{Z}}
\newcommand{\bft}{\mathbf{t}}

\newcommand{\bsv}{\boldsymbol{v}}

\newcommand{\bsg}{\boldsymbol{g}}
\newcommand{\bsu}{\boldsymbol{u}}
\newcommand{\bsz}{\boldsymbol{z}}
\newcommand{\bsI}{\boldsymbol{I}}
\newcommand{\bseta}{\boldsymbol{\eta}}

\newcommand*{\pp}[1]
  {\frac{\partial   }
        {\partial #1}
  }

\newcommand*{\pfrac}[2]
  {\frac{\partial #1}
        {\partial #2}
  }

  \newcommand*{\Bigset}[2]
  {
   \left\{ #1 \middle| #2 \right\}
  }

\newcommand{\beq}{\begin{equation}}
\newcommand{\eeq}{\end{equation}}

\DeclareMathOperator{\res}{Res}
\DeclareMathOperator{\diag}{diag}

\DeclareMathOperator{\Der}{Der}
\DeclareMathOperator{\grad}{grad}
\DeclareMathOperator*{\Res}{Res}       

\DeclareMathOperator{\td}{d\!}

\newtheorem{thm}{Theorem}[section]

\newtheorem{rmk}[thm]{Remark}

\newtheorem{lem}[thm]{Lemma}

\newtheorem{prop}[thm]{Proposition}
\newtheorem{defn}[thm]{Definition}
\newtheorem{ex}[thm]{Example}

\numberwithin{equation}{section}

\allowdisplaybreaks[4]

\begin{document}

\title{{Bihamiltonian structure of the $(n,1)$-type rational reductions of the 2D-Toda hierarchy}
	\footnotetext{{\footnotesize$^*$Corresponding author }
		}}
	\author{{
		Haonan Qu$^{1}$,~ Qiulan Zhao$^{*2}$} \\[1em]
		\footnotesize $^{1}$ School of Mathematics, Sun Yat-Sen University,
		 Guangzhou 510275, China; \\
		\footnotesize $^{2}$ College of Mathematics and Systems Science, \\
\footnotesize
        Shandong University of Science and Technology, Qingdao 266590, China\\[10pt]
\footnotesize
Email: quhn@mail.sysu.edu.cn,\, qlzhao@sdust.edu.cn
	}
    \date{}
	\maketitle
	\begin{abstract}
  We derive a local bihamiltonian structure
  for the rational reduction of the 2D-Toda hierarchy (RR2T) of $(n,1)$-type by direct computations,
  and construct an $(n+1)$-dimensional semisimple generalized Frobenius manifold with non-flat unity
  whose Principal Hierarchy contains its dispersionless flows.
	\end{abstract}


\tableofcontents


\section{Introduction}
The 2D-Toda hierarchy \cite{2dToda, Takasaki 2Dtoda, Ueno-Takasaki}
is a fundamental integrable hierarchy.
A special class of its reductions,
called the \textit{rational reduction of the 2D-Toda hierarchy} (abbreviated as RR2T)
and introduced by Brini and his collaborators \cite{rational-reduction},
plays an important role in the study of integrable systems and mathematical physics,
and is closely related to equivariant orbifold cohomology and equivariant mirror symmetry.
The Lax operator of the $(n,n')$-type RR2T,
where $(n,n')\in \mathbb{Z}_+^2$, takes the form
\begin{equation}\label{Lnn'}
  L_{(n,n')} =
  \left(
    1+W_1\Lmd^{-1}+\cdots + W_{n'}\Lmd^{-n'}
  \right)^{-1}
  \left(
    \Lmd^n + U_1\Lmd^{n-1} + \cdots + U_n
  \right),
\end{equation}
where $\Lmd:=\rme^{\veps\p_x}$ is the shift operator, and $W_i, U_j$ are unknown functions.
Many well-known integrable hierarchies,
such as the bi-graded Toda hierarchy \cite{extend bi Toda}
and the $q$-deformed Gelfand--Dickey hierarchy \cite{q-deformed GD},
can be obtained as further reductions of RR2T.
In the special case of $(n,n')=(1,1)$,
it reduces to the Ablowitz--Ladik hierarchy \cite{Ablowitz},
also known as the relativistic Toda hierarchy \cite{Kupershmidt, rtoda}.
The Ablowitz--Ladik hierarchy admits a local tri-Hamiltonian structure
at the full-dispersive level \cite{AL-triham},
and its connections with the Gromov--Witten invariants of local $\mathbb{CP}^1$
and with the theory of generalized Frobenius manifolds with non-flat unity have been further investigated in
\cite{local CP1, Brini 2012, GFM, LiuQuZhang, extend AL, LiuResolvedConifold}.

The main object of this paper is the bihamiltonian structure of the $(n,1)$-type RR2T.
It is known that the 2D-Toda hierarchy admits a compatible tri-Hamiltonian structure
constructed via the $R$-matrix formalism \cite{trh-toda},
and the second Hamiltonian structure can be consistently reduced to every $(n,n')$-type RR2T \cite{rational-reduction}.
However, the reduction procedure for the first and third Hamiltonian structures
does not appear to extend in a straightforward way to the general $(n,n')$-case.
At the dispersionless level,
the dispersionless limit of RR2T admits a bihamiltonian structure for several subcases,
and even admit a tri-Hamiltonian structure in the special cases $n=n'$ \cite{rational-reduction}.
While these results suggest the existence of rich Hamiltonian structures for RR2T,
the corresponding theory at the full-dispersive level is less well understood.
In this paper, we show that the $(n,1)$-type RR2T admits a local bihamiltonian structure at the full-dispersive level.
Unlike previous approaches based on reductions of the 2D-Toda Hamiltonian structures,
we construct the bihamiltonian structure of the $(n,1)$-type RR2T directly from its defining equations.
More precisely, we first derive Hamiltonian formulations of the flows by explicit computation,
and then use the super-variable formalism \cite{Liu lecture notes, Jacobi structure}
to verify the Jacobi identities and compatibility of the associated Poisson brackets.

Another motivation for the present work comes from the theory of Frobenius manifolds.
Introduced by Dubrovin \cite{Dubrovin1996} in the 1990s,
Frobenius manifolds provide a geometric framework for two-dimensional topological field theory \cite{Witten-1, Witten-2}
and have deep connections with Gromov--Witten theory, integrable systems, and related areas.
Associated to every Frobenius manifold is an integrable hierarchy of hydrodynamic type,
called the Principal Hierarchy, which carries a natural bihamiltonian structure \cite{Normal forms}.
In recent years,
it has become necessary to consider a class of generalized Frobenius manifolds with non-flat unity
\cite{Brini 2012, double hurwitz}.
This class of structures has been further studied in connection with integrable systems,
and the Dubrovin--Zhang theory of topological deformations of Principal Hierarchies
has been extended to this setting \cite{Liu2024, GFM, LiuQuZhang}.
This development has led to applications to integrable hierarchies
such as the $q$-deformed KdV hierarchy and the Ablowitz--Ladik hierarchy \cite{Liu2024, extend AL}.
These developments naturally raise the question of finding further examples
relating generalized Frobenius manifolds and integrable systems.
In this paper,
we construct an $(n+1)$-dimensional semisimple generalized Frobenius manifold $\mcalM_{(n,1)}$ with non-flat unity
via the Landau--Ginzburg superpotential formalism,
and show that the dispersionless limit of the $(n,1)$-type RR2T belongs to its Principal Hierarchy.

The paper is organized as follows.
In Sect.\,2, we introduce the Lax formulation of the $(n,1)$-type RR2T
together with two Hamiltonian operators $\mcalP_1$ and $\mcalP_2$,
and derive Hamiltonian formulations for its flows.
In Sect.\,3, we prove that the pair $(\mcalP_1,\mcalP_2)$
forms a bihamiltonian structure by means of the super-variable formalism.
In Sect.\,4, we construct a generalized Frobenius manifold $\mcalM_{(n,1)}$
with non-flat unity and show that the dispersionless limit of the $(n,1)$-type RR2T
belongs to the Principal Hierarchy of $\mcalM_{(n,1)}$.
Finally, in Sect.\,5, we summarize the main results and discuss several directions for future research.

\section{Basic properties of the $(n,1)$-type RR2T}

In this section, we briefly review the completed differential polynomial ring and the Lie algebra of pseudo-difference operators,
and then present the Lax and Hamiltonian formulations of the $(n,1)$-type RR2T.

\subsection{The Lie algebra of pseudo-difference operators}

Although the hierarchy under consideration is a discrete integrable system involving difference operators,
we begin by recalling the notion of differential operators.

Let $\mcalM$ be an $m$-dimensional manifold,
and let $\mcalU \subseteq \mcalM$ be an open subset equipped with local coordinates $\bsv := (v^\afa)_{1 \leq \afa \leq m}$.
Let us introduce the polynomial ring
\[
  R := C^\infty(\mcalU)[v^{\afa,s} \,|\, 1\leq\afa\leq m,\, s\in\bbZ_{+}]
\]
generated by the family of independent formal variables $\{v^{\afa,s}\}$, which are referred to as \textit{jet variables}.
We also set $v^{\afa,0}:= v^\afa$. The ring $R$ admits the \textit{differential grading} given by
\[
  \deg_\p f = 0,\quad \deg_\p v^{\afa,s}=s
\]
for all $f\in C^\infty(\mcalU)$, $1\leq\afa\leq m$ and $s\geq 0$.
Let $I_n$ be the ideal of $R$ generated by monomials whose differential degree is at least $n$,
then the filtration of ideals $I_1\supset I_2\supset I_3\supset\cdots$
gives rise to the \textit{completed differential polynomial ring}
\begin{equation}\label{completion}
  \mcalA:=\llim_{n}R/I_n.
\end{equation}
An element $f\in\mcalA$ is usually written as a formal series
\[
  f = \sum_{s\geq 0} \veps^sf_s,\quad \text{where $f_s\in R$ such that $\deg_\p f_s = s$}.
\]
Note that $\veps$ is just a formal parameter used to track the differential degree, and can be omitted.
Introduce the differential operator
\begin{equation}\label{operator px}
  \p_x:=\sum_{\afa=1}^{m}\sum_{s\geq 0}v^{\afa,s+1}\pp{v^{\afa,s}} \in \Der(\mcalA),
\end{equation}
and then it is clear that $v^{\afa,s}=\p_x^s v^\afa$.
Let $\mcalF:=\mcalA/\p_x\mcalA$ be the space of \textit{local functionals}, and
  \begin{align*}
    \int\td x\colon\, \mcalA &\to\, \mcalF, \quad
    f \mapsto\, \int f\td x
  \end{align*}
 be the quotient map.
Note that each operator $\mcalP\in\mcalA\fps{\veps\p_x}$ acts on $\mcalA$ naturally,
and can be regarded as a linear endomorphism of $\mcalA$. For each $\mcalP\in\mcalA\fps{\veps\p_x}$ of the form
\[
  \mcalP = \sum_{s\geq 0}\mcalP_s (\veps\p_x)^s,\quad \mcalP_s\in\mcalA,
\]
introduce the adjoint operator
$
  \mcalP^\dag := \sum_{s\geq 0} (-\veps\p_x)^s\circ \mcalP_s,
$
then it is well known that the following ``integration by parts formula'' holds true
\begin{equation}\label{by parts-260329}
  \int f(\mcalP g)\td x = \int (\mcalP^\dag f)g\td x,\quad \forall\, f,g\in\mcalA.
\end{equation}
The following maps are called \textit{variational derivatives}:
\[
  \frac{\delta}{\delta v^\afa}\colon \mcalF\to\mcalA,\quad
  \int f\td x\mapsto \sum_{s\geq 0}(-\p_x)^s\pfrac{f}{v^{\afa,s}},\quad 1\leq\afa\leq m.
\]
They are well-defined, and are uniquely characterized by the relation
\begin{equation}\label{delta relation-260329}
  \delta F = \int \delta v^\afa\cdot\frac{\delta F}{\delta v^\afa}\td x,\quad \forall\, F\in\mcalF.
\end{equation}

To study discrete integrable hierarchies, we introduce the \textit{shift operator}
\begin{equation}
  \Lmd := \rme^{\veps\p_x} = \sum_{k\geq 0}\frac{(\veps\p_x)^k}{k!} \in \mcalA\fps{\veps\p_x},
\end{equation}
which is also referred to as a difference operator. For each $s\in\bbC$ and $f\in\mcalA$, we set
\begin{equation}
  f^{[s]} := \Lmd^s(f) = \rme^{\veps s\p_x}f \in \mcalA.
\end{equation}
It is clear that $\Lmd^\dag = \Lmd^{-1}$,
and $(fg)^{[s]}=f^{[s]}g^{[s]}$ for all $f,g\in\mcalA$ and $s\in\bbC$. In particular,
$\Lmd$ is a ring automorphism of $\mcalA$.

Introduce the spaces of formal Laurent series with respect to $\Lmd^{\mp}$ as follows
\begin{align}
  \mfkg^- &:=\, \mcalA\fls{\Lmd^{-1}}
=
  \Bigset{\sum_{k\leq k_0}a_k\Lmd^k \,}{\, k_0\in\bbZ,\, a_k\in\mcalA},
\\
  \mfkg^+ &:=\, \mcalA\fls{\Lmd}
=
  \Bigset{\sum_{k\geq k_0}a_k\Lmd^k \,}{\, k_0\in\bbZ,\, a_k\in\mcalA}.
\end{align}
It is known that each $\mfkg^{\mp}$ is an associative algebra whose product is defined by
\[
  \left(
    \sum_k a_k\Lmd^k
  \right)\cdot
  \left(
    \sum_\ell b_\ell\Lmd^\ell
  \right)
=
  \sum_{k,\ell}a_k b_\ell^{[k]}\Lmd^{k+\ell}.
\]
Moreover, each $\mfkg^{\mp}$ is a Lie algebra with respect to the commutator $[X,Y]:=XY-YX$.
Elements in $\mfkg^-$ or $\mfkg^+$ are called \textit{pseudo-difference operators} of the first type or of the second type, respectively.
For each $X=\sum_{k\in\bbZ}a_k\Lmd^k\in\mfkg^-$ or $\mfkg^+$, we denote
\begin{equation}
  X_+:=\sum_{k\geq 0} a_k\Lmd^k,\quad
  X_-:=\sum_{k<0} a_k\Lmd^k,\quad
  \Res X := a_0.
\end{equation}
It is well-known that for each $X,Y\in\mfkg^-$ or $\mfkg^+$,
\begin{equation}\label{res commutator-260329}
  \int \Res [X,Y] \td x = 0.
\end{equation}

\subsection{Lax formulation}

Now let us introduce the Lax formulation of the $(n,1)$-type RR2T.
Fix an integer $n \geq 1$, and let $\bsv:=(v^1, \cdots, v^{n+1})$ be a family of local coordinates on an $(n+1)$-dimensional manifold.
We relabel these local coordinates as
\begin{equation}
  U_\afa := v^\afa, \quad W := v^{n+1}, \quad 1 \leq \afa \leq n.
\end{equation}
For convenience, we introduce the following additional notation:
\begin{equation}\label{short notation}
  U_0:= 1,\quad U_k := 0 \,\,\text{if $k>n$ or $k<0$}.
\end{equation}
In particular, we emphasize that $U_{n+1}=0$, which is NOT equal to $W$.

We introduce the difference operators $A, B \in \mfkg^-\cap\mfkg^+$ as follows
\begin{align}
 A:=\Lmd^n + U_1\Lmd^{n-1}+ U_2\Lmd^{n-2}+\cdots+U_n,\qquad
 B:= 1- W\Lmd^{-1},
\end{align}
then there exist unique operators $A^{-1}\in\mfkg^+$ and $B^{-1}\in\mfkg^-$ such that
$AA^{-1}=A^{-1}A=1$ and $BB^{-1}=B^{-1}B=1$.
Then we introduce the operators $L, \tilde L\in\mfkg^-$ and $M, \tilde M\in\mfkg^+$ as follows
\begin{equation} \label{Lax L}
  L:= B^{-1}A,\qquad \tilde L:=AB^{-1},
\end{equation}
\begin{equation}
  M:= A^{-1}B,\qquad \tilde M:=BA^{-1}.
\end{equation}
Note that
the Lax operator $L$ in \eqref{Lax L} is precisely the special case of \eqref{Lnn'} with $n'=1$ and $W_1\mapsto -W$.
The operators $L$ and $\tilde L$ have the form
\[
  L = \Lmd^n + \sum_{s<n}a_s^1\Lmd^s,\quad
  \tilde L = \Lmd^n + \sum_{s<n}\tilde a^1_s\Lmd^s
\]
for certain $a_s^1, \tilde a^1_s\in\mcalA$. It can be verified that there exist unique operators
$L^{\frac 1n}, \tilde L^{\frac 1n}\in\mfkg^-$ such that
$(L^{\frac 1n})^n = L$ and $(\tilde L^{\frac 1n})^n = \tilde L$.

\begin{defn}
The rational reduction of the 2D-Toda hierarchy (RR2T) of $(n,1)$-type consists of the system of
evolutionary PDEs given by
\begin{align}
  \frac\veps{c_{i,k}}\pfrac A{t^{i,k}} &=\, \left(\tilde L^{k+\frac{i-1}{n}}\right)_+A - A\left(L^{k+\frac{i-1}{n}}\right)_+, \label{n1-1}\\
  \frac\veps{c_{i,k}}\pfrac B{t^{i,k}} &=\, \left(\tilde L^{k+\frac{i-1}{n}}\right)_+B - B\left(L^{k+\frac{i-1}{n}}\right)_+, \label{n1-2}\\
  \frac\veps{c_{0,-k-1}}\pfrac A{t^{0,-k-1}} &=\, \left(\tilde M^{k+1}\right)_-A - A\left(M^{k+1}\right)_-,\label{n1-3}\\
  \frac\veps{c_{0,-k-1}}\pfrac B{t^{0,-k-1}} &=\, \left(\tilde M^{k+1}\right)_-B - B\left(M^{k+1}\right)_-, \label{n1-4}
\end{align}
for all $i\in\{2,3,\dots, n+1\}$ and $k\in\bbZ_{\geq 0}$,
where the constants $c_{i,k}$ and $c_{0,-k-1}$ are given by
\begin{equation}\label{const coef cik}
  c_{i,k}:=
  \begin{cases}
    \frac{\sqrt{n}}{i-1}
    \frac{\Gamma\left(1+\frac{i-1}{n}\right)}{\Gamma\left(k+1+\frac{i-1}{n}\right)},& 2\leq i\leq n \\
    \frac{1}{(k+1)!},& i=n+1
  \end{cases},
  \quad
  c_{0,-k-1} := (-1)^k k!.
\end{equation}
\end{defn}

In the above systems, $\{U_\afa\}_{1\leq\afa\leq n}$ and $W$ are regarded as unknown functions of spatial variable $x$ and time variables
\[\bft := \Bigset{t^{i,k}, t^{0,-k-1}}{2\leq i\leq n+1,\, k\geq 0}.\]
These equations are compatible, and can be rewritten as the following Lax equations
\begin{align}
  \veps\pfrac L{t^{i,k}} &=\, c_{i,k}\left[\left(L^{k+\frac{i-1}{n}}\right)_+, L\right], \label{positive-Lax}\\
  \veps\pfrac L{t^{0,-k-1}} &=\, c_{0,-k-1}\left[\left(M^{k+1}\right)_-, L\right] \label{negative-Lax}
\end{align}
for all $i\in\{2,\dots,n+1\}$ and $k\in\bbZ_{\geq 0}$.
The flows $\pp{t^{i,k}}$ and $\pp{t^{0,-k-1}}$ are called \textit{positive flows} and
\textit{negative flows} of the $(n,1)$-type RR2T, respectively.

  We deliberately choose the indices of the above time variables $\bft$ and the constants $c_{i,k}$,
  in order to ensure that
  the dispersionless limit of this hierarchy belongs to the Principal Hierarchy of a certain generalized Frobenius manifold;
  see details in Sect.\,\ref{subsection:PH}.
  The symbols $\{\pp{t^{1,k}}, \pp{t^{0,k}}\}_{k \geq 0}$,
  which have not been introduced previously,
  are introduced as formal notations for two sequences of possible extended flows of this hierarchy;
  see Remark \ref{rmk: extended n1 RR2T} for further discussion and related conjectures.

  The RR2T of $(n,1)$-type can be further reduced to the $q$-deformed Gelfand--Dickey hierarchy
  \cite{Zhonglun Cao, double hurwitz}
  associated with the Lax operator
  \begin{equation}\label{Lax qGD}
     L_{q\mathrm{GD}} := \Lmd^n + U_1\Lmd^{n-1}+\cdots + U_{n-1}\Lmd
  \end{equation}
  by setting $U_n=W=0$ in $L$ \eqref{Lax L},
  see also in Example 2.8 of \cite{rational-reduction}.
  In the special case of $n=2$,
  the reduced hierarchy coincides with the
  \textit{$q$-deformed KdV hierarchy}, whose Lax operator is $L_{q\mathrm{KdV}}=\Lmd^2 + U \Lmd$.
  It was proved in \cite{Liu2024} that a certain extension of the $q$-deformed KdV hierarchy is Miura-equivalent
  to the topological deformation of the Principal Hierarchy of a one-dimensional generalized Frobenius manifold with non-flat unity.

\begin{rmk}
  Consider the spectral problem
  $
    L\psi = \lmd\psi
  $ 
  of the Lax equations \eqref{positive-Lax}--\eqref{negative-Lax},
  where $\psi$ and $\lmd$ are the wave function and the spectral parameter respectively.
  Let us apply a gauge transformation $\psi = \rho\phi$ to this spectral problem,
  where the function $\rho$ is determined by the relation $U_n\rho=-W\rho^{[-1]}$.
  Then the spectral problem is transformed to
  $
    \widetilde{L}\phi = \frac{1}{\lmd}\phi
  $, where
  \[
    \widetilde{L}:=
    \left(
      1 + \widetilde{U}_1\widetilde{\Lmd}^{-1} + \cdots +
      \widetilde{U}_n\widetilde{\Lmd}^{-n}
    \right)^{-1}
    \left(
      \widetilde{\Lmd} + \widetilde{W}
    \right),\quad
    \widetilde{\Lmd}:=\Lmd^{-1}
  \]
  is the Lax operator of $(1,n)$-type RR2T \eqref{Lnn'}, and
  \begin{equation}\label{n1 and 1n}
    \widetilde{W}=\frac{1}{U_n},\quad
    \widetilde{U}_k = (-1)^k\frac{U_{n-k}}{U_n}
    \prod_{\ell=1}^{k}
    \left(\frac{W}{U_n}\right)^{[\ell]},\quad 1\leq k\leq n.
  \end{equation}
Therefore the RR2T of types $(n,1)$ and $(1,n)$ are related via the transformation \eqref{n1 and 1n}
together with spatial reflection $\veps\mapsto -\veps$,
generalizing the result for $n=1$ in Sect.\,6 of \cite{AL-triham}.
\end{rmk}

\begin{ex}
The first negative flow $\pp{t^{0,-1}}$ has the form
\begin{align}
  \veps\pfrac{U_\afa}{t^{0,-1}}
&=\,
  U_{\afa-1}\left(\frac{W}{U_n}\right)^{[n+1-\afa]} - W\left(\frac{U_{\afa-1}}{U_n}\right)^{[-1]},
\quad 1\leq\afa\leq n,
  \\
  \veps\pfrac{W}{t^{0,-1}}
&=\,
  \frac{W}{U_n^{[-1]}} - \frac{W}{U_n}.
\end{align}
\end{ex}

\begin{ex}
  Note that the operators $L^k$, $\tilde{L}^k \in\mfkg^-$
  have the form
  \begin{equation}\label{coef a^k_s}
    L^k = \sum_{s\leq nk}a_s^k\Lmd^s,\quad
    \tilde L^k = \sum_{s\leq nk}\tilde a^k_s\Lmd^s
  \end{equation}
for each $k\geq 0$, where $a^k_s, \tilde a^k_s \in\mcalA$ are certain coefficients.
Under this notation, the flows $\pp{t^{n+1,k}}$ in  \eqref{n1-1}--\eqref{n1-2} can be rewritten as
\begin{align}
  \veps\pfrac{U_\afa}{t^{n+1,k}} &=\,
  \frac{1}{(k+1)!}
    \sum_{s=0}^{n-\afa}
      \left(
        \tilde a^{k+1}_s U_{\afa+s}^{[s]}
       -a^{k+1,[n-\afa-s]}_s U_{\afa+s}
      \right), \label{pUptn+1}\\
  \veps\pfrac{W}{t^{n+1,k}} &=\,
  \frac{1}{(k+1)!}
  \left(\tilde a^{k+1}_0 W - a^{k+1,[-1]}_0W \right) \label{pWptn+1}
\end{align}
for all $1\leq\afa\leq n$ and $k\geq 0$.
\end{ex}

\subsection{The first Hamiltonian operator $\mcalP_1$}

Let us study the bihamiltonian structure of the RR2T of type $(n,1)$.
We shall introduce two skew-symmetric operator-valued matrices
$\mcalP_a\in\mcalA\fps{\veps\p_x}\otimes\bbC^{(n+1)\times(n+1)}$
for $a=1,2$, such that the flows of the $(n,1)$-type RR2T
can be written in Hamiltonian form with respect to them.
It will be proved in Sect.\,\ref{section:check jacobi} that the associated brackets
\begin{align*}
  \{,\}_a\colon \mcalF\times\mcalF \to\, \mcalF, \quad 
  (F,G) \mapsto\,
  \int \sum_{i,j=1}^{n+1}
    \frac{\delta F}{\delta v^i}
    \left(
      \mcalP_a^{ij}\frac{\delta G}{\delta v^j}
    \right)\td x,\quad a=1,2
\end{align*}
are compatible Poisson brackets; in other words, $(\mcalP_1,\mcalP_2)$ forms a bihamiltonian structure.

Now we introduce the operator $\mcalP_1$.

\begin{thm}\label{thm: HamP1}
  The RR2T of $(n,1)$-type \eqref{n1-1}--\eqref{n1-4} can be represented by the following Hamiltonian systems:
  \begin{equation}\label{HamP1-eqn}
    \veps\pp{t^{i,k}}
    \begin{pmatrix}
      U_1 \\
      \vdots \\
      U_n \\
      W
    \end{pmatrix}
  =
    \mcalP_1
    \begin{pmatrix}
      \delta H_{i,k}/\delta U_1 \\
      \vdots \\
      \delta H_{i,k}/\delta U_n \\
      \delta H_{i,k}/\delta W
    \end{pmatrix}
  \end{equation}
for all $i\in\{2,3,\dots,n+1\}$, $k\geq 0$ or $i=0$, $k\leq -1$,
where the entries of $\mcalP_1$ are given by
\begin{align}
  \mcalP_1^{\afa\beta}&=\,
    \left(
     U_{\afa+\beta-n}\Lmd^{n-\afa} - \Lmd^{\beta-n}U_{\afa+\beta-n}
    \right)\notag \\
  &\qquad
    +\left(
      W\Lmd^{\beta-n-1}U_{\afa+\beta-n-1}-U_{\afa+\beta-n-1}\Lmd^{n+1-\afa}W
    \right), \label{250812-1309-1}
\\
  \mcalP_1^{\afa,n} &=\, W\Lmd^{-1}U_{\afa-1}-U_{\afa-1}\Lmd^{n+1-\afa}W, \notag
\\
  \mcalP_1^{n,n} &=\, W\Lmd^{-1}U_{n-1} - U_{n-1}\Lmd W,
\qquad
  \mcalP_1^{n,n+1} = (\Lmd-1)W,
  \notag
\\
  \mcalP_1^{\afa,n+1} &=\, \mcalP_1^{n+1,\afa}=\mcalP_1^{n+1,n+1}=0,\notag
\\
  \mcalP_1^{n,\afa}&=\, -\left(\mcalP_1^{\afa,n}\right)^\dag,\quad
  \mcalP_1^{n+1,n} = -\left(\mcalP_1^{n,n+1}\right)^\dag  \label{250812-1309-2}
\end{align}
for all $1\leq \afa,\beta\leq n-1$, and the Hamiltonians have the forms
\begin{align}
  H_{i,k}&=\,
    c_{i,k+1}
    \int\Res\left(L^{k+1+\frac{i-1}{n}}\right)\td x,\quad 2\leq i\leq n+1,\, k\geq -1, \label{Ham Hip-1}\\
  H_{0,-k-1}&=\,
    -c_{0,-k}\int\Res\left(M^k\right)\td x,\qquad k\geq 1, \label{Ham Hip-2}\\
  H_{0,-1} &=\, \int\log\frac{W}{U_n}\td x. \label{Ham Hip-3}
\end{align}
\end{thm}

  We shall prove only the case $i=n+1$ in \eqref{HamP1-eqn},
  the other cases of the above theorem are similar and are left to the reader.
  Let us first compute the variational derivatives $\frac{\delta H_{n+1,k}}{\delta U_\afa}$ and
  $\frac{\delta H_{n+1,k}}{\delta W}$
  by using \eqref{by parts-260329} and \eqref{delta relation-260329}.
  Notice that the operators $L^{k}B^{-1}$ have the forms
  \begin{equation}\label{coef b^k_s}
    L^k B^{-1}=\sum_{s\leq nk}b^k_s\Lmd^s,\quad k\geq 0,
  \end{equation}
  where $b^k_s\in\mcalA$ are certain coefficients.
  Using the above notation and \eqref{res commutator-260329}, we obtain
  \begin{align*}
  &\,
    \frac{1}{c_{n+1,k}}\delta H_{n+1,k} =
    \int\res\left(\delta L\cdot L^{k+1}\right) \td x\\
  =&\,
    \int\res\left(
      -B^{-1}\cdot\delta B\cdot L^{k+2}
      +B^{-1}\cdot\delta A\cdot L^{k+1}
    \right)\td x\\
  =&\,
    \int\res\left(
      -\delta B\cdot L^{k+2}B^{-1} + \delta A\cdot L^{k+1}B^{-1}
    \right)\td x\\
  =&\,
    \int\left(
      \delta W\cdot b_1^{k+2,[-1]}
     +\sum_{\afa=1}^{n}
       \delta U_\afa\cdot b_{\afa-n}^{k+1,[n-\afa]}
    \right)\td x,
  \end{align*}
here $b_s^{k,[\ell]}:=\Lmd^\ell(b_s^k)$,
therefore for $1\leq\afa\leq n$, we have
\begin{equation}\label{var-derivative}
  \frac{1}{c_{n+1,k}}
  \frac{\delta H_{n+1,k}}{\delta U_\afa} = b_{\afa-n}^{k+1,[n-\afa]},\quad
  \frac{1}{c_{n+1,k}}
  \frac{\delta H_{n+1,k}}{\delta W} = b_1^{k+2,[-1]}.
\end{equation}

To prove Theorem \ref{thm: HamP1}, we need the following two lemmas.

\begin{lem}
  The coefficients $a^k_s$, $\tilde a^k_s$ and $b^k_s$ defined in \eqref{coef a^k_s}, \eqref{coef b^k_s}
  satisfy the following equations for all $k\geq 0$, $1\leq\afa\leq n$ and $s\in\bbZ$:
  \begin{align}
  &
    \sum_{\ell=0}^{n-1}
      \left(
        \tilde a^k_{n-\afa-\ell}U_{n-\ell}^{[n-\afa-\ell]}
       -a^{k,[\ell]}_{n-\afa-\ell}U_{n-\ell}
      \right)
    =
      a^{k,[n]}_{-\afa}-\tilde a^k_{-\afa}, \label{coef lem 1}
\\&
    \tilde a^k_s - W^{[s+1]}\tilde a^k_{s+1}
  =
    a^k_s - Wa^{k,[-1]}_{s+1},   \label{coef lem 2}
  \end{align}
\begin{equation}\label{coef lem 3}
  \begin{cases}
    a^k_s  =  b^k_s - W^{[s+1]}b^k_{s+1}, \\
    \tilde a^k_s  =  b^k_s - W b^{k,[-1]}_{s+1},
  \end{cases}
\quad
  \begin{cases}
    a^{k+1}_s = b^k_{s-n} + \sum_{\ell=0}^{n-1}U^{[s-\ell]}_{n-\ell}b^k_{s-\ell}, \\
    \tilde a^{k+1}_s = b^{k,[n]}_{s-n} + \sum_{\ell=0}^{n-1}U_{n-\ell}b_{s-\ell}^{k,[\ell]}.
  \end{cases}
\end{equation}
Here we use the conventions that $a^k_s=\tilde a^k_s=b^k_s=0$ for $s>nk$.
\end{lem}

\begin{proof}
  The relations $\tilde L^kA=AL^k$ and $\tilde L^kB=BL^k$ imply \eqref{coef lem 1} and \eqref{coef lem 2} respectively,
  while \eqref{coef lem 3} are followed from
  \begin{align*}
     L^k=(L^kB^{-1})B,\quad
     \tilde L^k=B(L^kB^{-1}),
  \quad
     L^{k+1}=(L^kB^{-1})A,\quad
     \tilde L^{k+1}=A(L^kB^{-1}).
  \end{align*}
  The details are omitted.
\end{proof}

\begin{lem}\label{lem:b-n}
The following identity holds true for all $k\geq 0$:
\begin{equation}\label{eq:b-n}
  b^{k+1}_{-n}
=
  \frac{\Lmd-1}{\Lmd^n-1}\left(
    W\frac{\delta H_{n+1,k}}{\delta W}\right)
   +\sum_{\ell=1}^{n-1}
     \frac{\Lmd^{\ell-n}-1}{\Lmd^n-1}
     \left(
       U_\ell\frac{\delta H_{n+1,k}}{\delta U_\ell}
     \right).
\end{equation}
\end{lem}

We remark that the operator $\frac{\Lmd-1}{\Lmd^n-1}$ should be interpreted as
\[
\frac{\rme^{\veps\p_x}-1}{\rme^{n\veps\p_x}-1}
=
  \frac 1n + \left(-\frac 12 + \frac{1}{2n}\right)\veps\p_x
  +\left(\frac{n}{12} - \frac 14 + \frac{1}{6n}\right)(\veps\p_x)^2 + O(\veps^3) \in\mcalA\fps{\veps\p_x},
\]
and $\frac{\Lmd^{\ell-n}-1}{\Lmd^n-1}$ is defined similarly.

\begin{proof}
  Following from \eqref{coef lem 3}, we obtain
  \begin{align*}
  &\,
    (\Lmd^n-1)b^{k+1}_{-n} = b^{k+1,[n]}_{-n}-b^{k+1}_{-n}\\
  =&\,
    \left(
      \tilde a_0^{k+2}-\sum_{\ell=0}^{n-1}U_{n-\ell}b_{-\ell}^{k+1,[\ell]}
    \right)
   -\left(
     a^{k+2}_0-\sum_{\ell=0}^{n-1}U_{n-\ell}^{[-\ell]}b_{-\ell}^{k+1}
   \right) \\
  =&\,
    \left(
      b^{k+2}_0 - Wb_1^{k+2,[-1]}
    \right)
   -\left(
     b^{k+2}_0-W^{[1]}b_1^{k+2}
    \right) 
    +\sum_{\ell=0}^{n-1}
    \left(
      U_{n-\ell}^{[-\ell]}b^{k+1}_{-\ell}
     -U_{n-\ell}b_{-\ell}^{k+1,[\ell]}
    \right)\\
  =&\,
    \left(
      W^{[1]}-W\Lmd^{-1}
    \right)b^{k+2}_1
   +\sum_{\ell=0}^{n-1}
     \left(
       U_{n-\ell}^{[-\ell]}
      -U_{n-\ell}\Lmd^\ell
     \right) b_{-\ell}^{k+1}.
  \end{align*}
Then, applying \eqref{var-derivative}, the lemma is proved.
\end{proof}

We now proceed to prove the Theorem \ref{thm: HamP1}.

\begin{proof}[Proof of the Theorem \ref{thm: HamP1}] For each $1\leq\afa\leq n$,
from \eqref{pUptn+1}, \eqref{coef lem 1} and \eqref{coef lem 3}, we obtain
\begin{align*}
&\,
  \frac{\veps}{c_{n+1,k}}\pfrac{U_\afa}{t^{n+1,k}}
=
  \left(
  \sum_{\ell=0}^{n-1}
  -\sum_{\ell=n-\afa+1}^{n-1}
  \right)
    \left(
      U_{n-\ell}^{[n-\afa-\ell]}
      \tilde a^{k+1}_{n-\afa-\ell}
     -U_{n-\ell}\,a^{k+1,[\ell]}_{n-\afa-\ell}
    \right)\\
=&\,
  \left(
    a_{-\afa}^{k+1,[n]}
   -\tilde a_{-\afa}^{k+1}
  \right)
 -\sum_{s=1-\afa}^{-1}
   \left(
     U_{\afa+s}^{[s]}\,\tilde a^{k+1}_s
    -U_{\afa+s}\,a_s^{k+1,[n-\afa-s]}
   \right) \\
=&\,
  \left(
    b_{-\afa}^{k+1}
   -b_{-\afa+1}^{k+1}W^{[1-\afa]}
  \right)^{[n]}
 -\left(
    b^{k+1}_{-\afa} - b^{k+1,[-1]}_{1-\afa}W
  \right) \\
 &\quad
   -\sum_{s=1-\afa}^{-1}
   \left(
     U_{\afa+s}^{[s]}
     \left(
       b^{k+1}_s - b^{k+1,[-1]}_{s+1}W
     \right)
    -U_{\afa+s}\left(
      b_s^{k+1} - b^{k+1}_{s+1}W^{[s+1]}
     \right)^{[n-\afa-s]}
   \right)\\
=&\,
  (\Lmd^n-1)b_{-\afa}^{k+1}
 +\left(
    W\Lmd^{-1}-\Lmd^{n+1-\afa}\circ W \circ \Lmd^{\afa-1}
  \right) b_{1-\afa}^{k+1}\\
&\quad
  +\sum_{s=1-\afa}^{-1}
    \left(
      U_{\afa+s}\Lmd^{n-\afa-s} - U_{\afa+s}^{[s]}
    \right)b_s^{k+1}\\
&\quad
  +\sum_{s=1-\afa}^{-1}
    \left(
      WU_{\afa+s}^{[s]}\Lmd^{-1}
     -W^{[n+1-\afa]}U_{\afa+s}\Lmd^{n-\afa-s}
    \right)b^{k+1}_{s+1}.
\end{align*}
Then using \eqref{var-derivative}, \eqref{eq:b-n} and notation \eqref{short notation},
by straightforward calculation we have
\begin{align}
  \veps\pfrac{U_\afa}{t^{n+1,k}}
&=\,
  \sum_{\beta=n-\afa}^{n-1}
    \left(
      U_{\afa+\beta-n}\Lmd^{n-\afa}
     -\Lmd^{\beta-n}U_{\afa+\beta-n}
    \right)
    \frac{\delta H_{n+1,k}}{\delta U_{\beta}}\notag \\
&\quad
  +\sum_{\beta=n+1-\afa}^{n}
    \left(
      W\Lmd^{\beta-n-1}U_{\afa+\beta-n-1}
     -U_{\afa+\beta-n-1}\Lmd^{n+1-\afa}W
    \right)\frac{\delta H_{n+1,k}}{\delta U_\beta}  \label{250806-1710-1}
\end{align}
for $1\leq \afa\leq n-1$, and
\begin{equation}
  \veps\pfrac{U_n}{t^{n+1,k}} = \Big((\Lmd-1)W\Big)\frac{\delta H_{n+1,k}}{\delta W}
  +\sum_{\beta=1}^{n}
    \left(
      W\Lmd^{\beta-n-1}U_{\beta-1}-U_{\beta-1}\Lmd W
    \right)
    \frac{\delta H_{n+1,k}}{\delta U_\beta}.  \label{250806-1710-2}
\end{equation}

Finally, by using \eqref{pWptn+1}, \eqref{coef lem 2}--\eqref{coef lem 3} and \eqref{var-derivative}, we obtain
\begin{align}
&\,
  \frac{\veps}{c_{n+1,k}}\pfrac{W}{t^{n+1,k}} =
    \tilde a_0^{k+1}W - a_0^{k+1,[-1]}W  = \tilde a^{k+1}_{-1}-a^{k+1}_{-1}  \notag\\
=&\,
  \left(
    b^{k+1}_{-1}-Wb_0^{k+1,[-1]}
  \right)
 -\left(
    b^{k+1}_{-1}-Wb^{k+1}_0
  \right) \notag\\
=&\,
  W(1-\Lmd^{-1})b^{k+1}_0
=
  W(1-\Lmd^{-1})
  \left(
  \frac{1}{c_{n+1,k}}
  \frac{\delta H_{n+1,k}}{\delta U_n}
  \right). \label{250806-1710-3}
\end{align}
Equations \eqref{250806-1710-1}--\eqref{250806-1710-3} imply the case $i=n+1$ of \eqref{HamP1-eqn}.
The proof of other cases are similar, and we omit its details. The theorem is proved.
\end{proof}

\subsection{The second Hamiltonian operator $\mcalP_2$}

Let us introduce the operator $\mcalP_2$.

\begin{thm} \label{thm: operator P2}
  The RR2T of $(n,1)$-type \eqref{n1-1}--\eqref{n1-4} can be represented by the following Hamiltonian systems:
  \begin{equation}\label{HamP2-eqn}
    \veps
    \left(k+\mu_i+\frac12\right)
    \pp{t^{i,k}}
    \begin{pmatrix}
      U_1 \\
      \vdots \\
      U_n \\
      W
    \end{pmatrix}
  =
    \mcalP_2
    \begin{pmatrix}
      \delta H_{i,k-1}/\delta U_1 \\
      \vdots \\
      \delta H_{i,k-1}/\delta U_n \\
      \delta H_{i,k-1}/\delta W
    \end{pmatrix}
  \end{equation}
for all $(i,k)\in\Big(\{2,3,\dots,n+1\}\times\bbZ_{\geq 0}\Big)\cup\Big(\{0\}\times\bbZ_{\leq -1}\Big)$,
where the constants
\begin{equation}\label{P2 mu}
\mu_0 := -\frac12,\quad 
\mu_i := \frac{i-1}{n}-\frac12,\quad 2\leq i\leq n+1,
\end{equation}
the Hamiltonians $H_{i,p}$ are given as in \eqref{Ham Hip-1}--\eqref{Ham Hip-3}, and the
entries of $\mcalP_2$ are given by
\begin{equation}\label{HamP2Op1}
  \mcalP_2^{\afa\beta}=
  \begin{cases}
    U_\afa
    \frac{(\Lmd^{n-\afa}-1)(1-\Lmd^\beta)}{\Lmd^n-1}U_\beta
   +\sum\limits_{s=1}^{\beta}
     \left(
       U_{\beta-s}\Lmd^s U_{\afa+s}
      -U_{\afa+s}\Lmd^{\beta-\afa-s}U_{\beta-s}
     \right),
   \quad \text{if\, $\afa\geq\beta$},
  \\
   U_\afa
    \frac{(1-\Lmd^{-\afa})(\Lmd^{\beta-n}-1)}{1-\Lmd^{-n}}U_\beta
   +\sum\limits_{s=1}^{\afa}
    \left( U_{\afa-s}\Lmd^{\beta-\afa+s}U_{\beta+s}-
      U_{\beta+s}\Lmd^{-s}U_{\afa-s}\right),\quad \text{if\, $\afa<\beta$}
  \end{cases}
\end{equation}
for all $1\leq \afa,\beta\leq n-1$, and
\begin{align}
  \mcalP_2^{\afa,n}&=\,\mcalP_2^{n,\afa}=\mcalP_2^{n,n} = 0, \notag
\\
  \mcalP_2^{\afa,n+1} &=\,
    U_\afa\frac{(\Lmd^n-\Lmd^{n-\afa})(\Lmd-1)}{\Lmd^n-1}W, \notag
\\
  \mcalP_2^{n,n+1} &=\,
    U_n(\Lmd-1)W,\notag
\\
  \mcalP^{n+1,n+1}_2 &=\,  W\frac{(\Lmd^n-\Lmd^{-1})(\Lmd-1)}{\Lmd^n-1}W, \notag
\\
  \mcalP_2^{n+1,\afa} &=\, -\left(\mcalP_2^{\afa,n+1}\right)^\dag,
  \quad
  \mcalP_2^{n+1,n} = -\left(\mcalP_2^{n,n+1}\right)^\dag. \label{HamP2Op2}
\end{align}
\end{thm}

\begin{proof}
  We also only proof the case $i=n+1$ in \eqref{HamP2-eqn},
  the proof of other cases are similar and we omit its details.
  For each $1\leq\afa\leq n$,
  from \eqref{pUptn+1} and \eqref{coef lem 3} we have
  \begin{align}
  &\,
    \frac{\veps}{c_{n+1,k}}
    \pfrac{U_\afa}{t^{n+1,k}}
  =
    \sum_{s=0}^{n-\afa}
      \left(
        \tilde a^{k+1}_s U_{\afa+s}^{[s]}
       -a_s^{k+1,[n-\afa-s]}U_{\afa+s}
      \right) \notag \\
  =&\,
    \sum_{s=0}^{n-\afa}
      \left(
        U^{[s]}_{\afa+s}
          \left(
            b^{k,[n]}_{s-n}
           +\sum_{\ell=0}^{n-1}U_{n-\ell}b^{k,[\ell]}_{s-\ell}
          \right)
       -U_{\afa+s}
          \left(
            b^k_{s-n}
           +\sum_{\ell=0}^{n-1}U_{n-\ell}^{[s-\ell]}b^k_{s-\ell}
          \right)^{[n-\afa-s]}
      \right)\notag\\
  =&\,
    U_\afa\left(
      b_{-n}^{k,[n]} - b_{-n}^{k,[n-\afa]}
    \right)
   +\sum_{s=1}^{n-\afa}
     \left(
       U_{\afa+s}^{[s]}b_{s-n}^{k,[n]}
      -U_{\afa+s}b^{k,[n-\afa-s]}_{s-n}
     \right) \notag\\
  &\quad
    +\sum_{s=0}^{n-\afa}
       \sum_{\ell=0}^{n-1}
         \left(
           U^{[s]}_{\afa+s}U_{n-\ell}b^{k,[\ell]}_{s-\ell}
          -U_{\afa+s}U_{n-\ell}^{[n-\afa-\ell]}b^{k,[n-\afa-s]}_{s-\ell}
         \right) \notag\\
  =&\,
    U_\afa(\Lmd^n-\Lmd^{n-\afa})b^k_{-n}
   +\sum_{\beta=1}^{n-\afa}
     \left(
       U_{\afa+\beta}^{[\beta]}\Lmd^n
      -U_{\afa+\beta}\Lmd^{n-\afa-\beta}
     \right)b_{\beta-n}^k\notag\\
   &\quad
   +\left(
       \sum_{s=0}^{n-\afa}
       \sum_{\ell=s+1}^{n-1}
      +\sum_{\ell=0}^{n-\afa}
       \sum_{s=\ell}^{n-\afa}
    \right)
         \left(
           U^{[s]}_{\afa+s}U_{n-\ell}b^{k,[\ell]}_{s-\ell}
          -U_{\afa+s}U_{n-\ell}^{[n-\afa-\ell]}b^{k,[n-\afa-s]}_{s-\ell}
         \right). \label{250807-0834}
  \end{align}
Observe that the double sum $\sum\limits_{\ell=0}^{n-\afa}
       \sum\limits_{s=\ell}^{n-\afa}(\cdots)$ appearing in \eqref{250807-0834} vanishes,
       then applying \eqref{var-derivative} and \eqref{eq:b-n}, using notation \eqref{short notation},
       and by straightforward calculation, we obtain
\begin{align}
  \veps\frac{c_{n+1,k-1}}{c_{n+1,k}}
  \pfrac{U_\afa}{t^{n+1,k}}
&=\,
  \sum_{\beta=0}^{n-1}
  \left(
    U_\afa
    \frac{(\Lmd^n-\Lmd^{n-\afa})(\Lmd^{\beta-n}-1)}{\Lmd^n-1}
    U_\beta
  \right)\frac{\delta H_{n+1,k-1}}{\delta U_\beta} \notag \\
&\quad
  +\sum_{s=0}^{n-\afa}
   \sum_{\beta=s}^{n-1}
     \left(
       U_{\beta-s}\Lmd^s U_{\afa+s}
      -U_{\afa+s}\Lmd^{\beta-\afa-s}U_{\beta-s}
     \right)
     \frac{\delta H_{n+1,k-1}}{\delta U_\beta}\notag \\
&\quad
  +\left(
      U_\afa \frac{(\Lmd^n-\Lmd^{n-\afa})(\Lmd-1)}{\Lmd^n-1}W
   \right)
   \frac{\delta H_{n+1,k-1}}{\delta W} \label{250807-1026-1}
\end{align}
for all $1\leq\afa\leq n$,
and notice that $\frac{c_{n+1,k}}{c_{n,k}}=k+1=k+\mu_{n+1}+\frac12$.

In particular, when $\afa=n$, we have
\begin{equation}
  \veps
  \left(k+\mu_{n+1}+\frac12\right)
  \pfrac{U_n}{t^{n+1,k}}
=
  \Big(
    U_n(\Lmd-1)W
  \Big)\frac{\delta H_{n+1,k-1}}{\delta W},
\end{equation}
and the entries $\mcalP_2^{\afa\beta}$  for all $1\leq \afa,\beta\leq n-1$ have the form
\begin{align}
  \mcalP_2^{\afa\beta}
&=\, U_\afa\frac{(\Lmd^n-\Lmd^{n-\afa})(\Lmd^{\beta-n}-1)}{\Lmd^n-1}U_\beta
  +\sum_{s=0}^{\beta}
  \left(
    U_{\beta-s}\Lmd^sU_{\afa+s} - U_{\afa+s}\Lmd^{\beta-\afa-s}U_{\beta-s}
  \right) \notag\\
&=\,
  \left(
  U_\afa\frac{(\Lmd^n-\Lmd^{n-\afa})(\Lmd^{\beta-n}-1)}{\Lmd^n-1}U_\beta
    +U_\afa(1-\Lmd^{\beta-\afa})U_\beta
  \right) \notag\\
&\quad
  +\sum_{s=1}^{\beta}
  \left(
    U_{\beta-s}\Lmd^sU_{\afa+s} - U_{\afa+s}\Lmd^{\beta-\afa-s}U_{\beta-s}
  \right) \notag\\
&=\,
  U_\afa\frac{(\Lmd^{n-\afa}-1)(1-\Lmd^\beta)}{\Lmd^n-1}U_\beta
  +\sum_{s=1}^{\beta}
  \left(
    U_{\beta-s}\Lmd^sU_{\afa+s} - U_{\afa+s}\Lmd^{\beta-\afa-s}U_{\beta-s}
  \right).  \label{250815-1524}
\end{align}
  We notice that when $1\leq \afa<\beta\leq n-1$, the entries $\mcalP^{\afa\beta}_2$ can also be expressed as
  \begin{align}
    \mcalP^{\afa\beta}_2
  &=\,
  U_\afa\frac{(\Lmd^n-\Lmd^{n-\afa})(\Lmd^{\beta-n}-1)}{\Lmd^n-1}U_\beta
  +\sum_{s=0}^{\beta}
    U_{\beta-s}\Lmd^sU_{\afa+s}
  -\sum_{s=-\afa}^{\beta-\afa}U_{\beta-s}\Lmd^s U_{\afa+s}\notag\\
  &=\,
    U_\afa
    \frac{(1-\Lmd^{-\afa})(\Lmd^{\beta-n}-1)}{1-\Lmd^{-n}}U_\beta
   +\left(
     \sum_{s=\beta-\afa+1}^{\beta}
    -\sum_{s=-\afa}^{-1}
    \right)U_{\beta-s}\Lmd^s U_{\afa+s}\notag\\
  &=\,
    U_\afa
    \frac{(1-\Lmd^{-\afa})(\Lmd^{\beta-n}-1)}{1-\Lmd^{-n}}U_\beta
  +
    \sum_{s=1}^{\afa}
      \left(
        U_{\afa-s}\Lmd^{\beta-\afa+s}U_{\beta+s}
       -U_{\beta+s}\Lmd^{-s}U_{\afa-s}
      \right),
  \end{align}
therefore \eqref{HamP2Op1} holds true.

Finally, following from \eqref{pWptn+1}, \eqref{coef lem 3}, \eqref{var-derivative} and \eqref{eq:b-n},
we obtain
\begin{align}
&\,
  \frac{\veps}{c_{n+1,k}}
  \pfrac{W}{t^{n+1,k}}
=
  \tilde a_0^{k+1}W-a_0^{k+1,[-1]}W \notag\\
=&\,
  W\left(
    b_{-n}^{k,[n]}+\sum_{\ell=0}^{n-1}U_{n-\ell}b_{-\ell}^{k,[\ell]}
   -b_{-n}^{k,[-1]}
   -\sum_{\ell=0}^{n-1}U_{n-\ell}^{[-\ell-1]}b_{-\ell}^{k,[-1]}
  \right) \notag \\
=&\,
  W(\Lmd^n-\Lmd^{-1})b_{-n}^k
 +W\sum_{\ell=0}^{n-1}
   \left(
     U_{n-\ell}\Lmd^{\ell}
    -U_{n-\ell}^{[-\ell-1]}\Lmd^{-1}
   \right)b_{-\ell}^k \notag \\
=&\,
  \left(
    W\frac{(\Lmd^n-\Lmd^{-1})(\Lmd-1)}{\Lmd^n-1}W
  \right)
  \left(
    \frac{1}{c_{n+1,k-1}}
    \frac{\delta H_{n+1,k-1}}{\delta W}
  \right) \notag \\
&\quad
+
  \sum_{\beta=1}^{n}
    \left(
      W\frac{(1-\Lmd^{-1})(\Lmd^\beta-1)}{\Lmd^n-1}U_\beta
    \right)
    \left(
      \frac{1}{c_{n+1,k-1}}
      \frac{\delta H_{n+1,k-1}}{\delta U_\beta}
    \right). \label{250807-1026-2}
\end{align}
Equations \eqref{250807-1026-1}--\eqref{250807-1026-2} imply the case $i=n+1$ of \eqref{HamP2-eqn}.
The proof of other cases are similar, and we omit its details. The theorem is proved.
\end{proof}

\subsection{Examples}

For convenience we provide the explicit expressions of the operators $(\mcalP_1, \mcalP_2)$
given by \eqref{250812-1309-1}--\eqref{250812-1309-2}, \eqref{HamP2Op1}--\eqref{HamP2Op2}
for $n=1,2,3$.

\begin{ex}
If $n=1$, then the hierarchy \eqref{positive-Lax}--\eqref{negative-Lax} coincides with the Ablowitz--Ladik hierarchy,
see \cite{AL-triham, extend AL} and references therein. The operators $(\mcalP_1,\mcalP_2)$ have the form
\[
  \mcalP_1 =
  \begin{pmatrix}
    W\Lmd^{-1} - \Lmd W & (\Lmd-1)W \\[12pt]
    W(1-\Lmd^{-1}) & 0
  \end{pmatrix},
\quad
  \mcalP_2 =
  \begin{pmatrix}
    0 & U_1(\Lmd-1)W \\[12pt]
    W(1-\Lmd^{-1})U_1 & W(\Lmd-\Lmd^{-1})W
  \end{pmatrix},
\]
which coincides with (3.2), (3.9) of \cite{AL-triham}.
\end{ex}

\begin{ex} The case $n=2$ has been studied in \cite{2-1 RR2T},
and $(\mcalP_1, \mcalP_2)$ have the form
\begin{align*}
  \mcalP_1 &=\,
    \begin{pmatrix}
      \Lmd-\Lmd^{-1} & W\Lmd^{-1}-\Lmd^2 W & 0 \\[12pt]
      W\Lmd^{-2}-\Lmd W & W\Lmd^{-1} U_1 - U_1\Lmd W & (\Lmd-1)W \\[12pt]
      0 & W(1-\Lmd^{-1}) & 0
    \end{pmatrix},
\\[12pt]
  \mcalP_2 &=\,
    \begin{pmatrix}
      \Lmd U_2 - U_2\Lmd^{-1} + U_1\frac{1-\Lmd}{1+\Lmd}U_1 & 0 & U_1\frac{\Lmd-1}{1+\Lmd^{-1}}W \\[12pt]
      0 & 0 & U_2(\Lmd-1)W \\[12pt]
      W\frac{1-\Lmd^{-1}}{\Lmd+1}U_1 & W(1-\Lmd^{-1})U_2 & W\frac{\Lmd^2-\Lmd^{-1}}{\Lmd+1}W
    \end{pmatrix}.
\end{align*}
\end{ex}

\begin{ex} If $n=3$, then
\begin{align*}
  \mcalP_1 &=\,
    \begin{pmatrix}
      0 & \Lmd^2-\Lmd^{-1} & W\Lmd^{-1}-\Lmd^3W & 0 \\[12pt]
      \Lmd-\Lmd^{-2} & U_1\Lmd-\Lmd^{-1}U_1 + W\Lmd^{-2}-\Lmd^2 W & W\Lmd^{-1}U_1-U_1\Lmd^2 W & 0 \\[12pt]
      W\Lmd^{-3}-\Lmd W & W\Lmd^{-2}U_1-U_1\Lmd W & W\Lmd^{-1}U_2 - U_2\Lmd W & (\Lmd-1)W \\[12pt]
      0 & 0 & W(1-\Lmd^{-1}) & 0
    \end{pmatrix},
\\[12pt]
  \mcalP_2 &=\,
    \left(
      \begin{array}{c|c|}
        U_1 \frac{\Lmd^{-1}-\Lmd}{\Lmd+1+\Lmd^{-1}} U_1 + \Lmd U_2 - U_2\Lmd^{-1}
      & U_1 \frac{1-\Lmd}{\Lmd+1+\Lmd^{-1}} U_2 + \Lmd^2 U_3 - U_3\Lmd^{-1}
    \\[12pt]
        U_2 \frac{\Lmd^{-1}-1}{\Lmd+1+\Lmd^{-1}} U_1 + \Lmd U_3 - U_3 \Lmd^{-2}
      &
        U_2 \frac{\Lmd^{-1}-\Lmd}{\Lmd+1+\Lmd^{-1}} U_2 + U_1\Lmd U_3 - U_3\Lmd^{-1}U_1
    \\[12pt]
       0&0
    \\[12pt]
       W \frac{\Lmd^{-1}-\Lmd^{-2}}{\Lmd+1+\Lmd^{-1}} U_1
      &W \frac{1-\Lmd^{-2}}{\Lmd+1+\Lmd^{-1}} U_2
      \end{array}
    \right.
\\[15pt]
&\qquad \quad
  \left.
    \begin{array}{c|c}
    0&\, U_1 \frac{\Lmd^2-\Lmd}{\Lmd+1+\Lmd^{-1}} W \,\\[12pt]
    0&\, U_2 \frac{\Lmd^2-1}{\Lmd+1+\Lmd^{-1}} W \,\\[12pt]
    0&\, U_3(\Lmd-1)W \,\\[12pt]
    \, W(1-\Lmd^{-1})U_3 \, &\, W \frac{\Lmd^2-\Lmd^{-2}}{\Lmd+1+\Lmd^{-1}}W \,
    \end{array}
  \right).
\end{align*}
\end{ex}

\section{Bihamiltonian structure}
\label{section:check jacobi}

To verify that $(\mcalP_1, \mcalP_2)$ given in \eqref{HamP1-eqn} and \eqref{HamP2-eqn} form a bihamiltonian structure,
we use the theory of the Schouten bracket on the space of local functionals over the supermanifold $\hat{\mcalM}$,
where $\hat{\mcalM}=\Pi(T^*\mcalM)$ denotes the cotangent bundle of an $m$-dimensional smooth manifold $\mcalM$ with reversed fiber parity;
see \cite{Liu lecture notes, Jacobi structure} for details.

\subsection{The super-variable formalism and the Schouten bracket}

We first recall the basic notions of the super-variable formalism and the Schouten bracket
\cite{Liu lecture notes, Jacobi structure}.
Fix a local trivialization $\hat \mcalU := \mcalU\times \mathbb{R}^{0|m}$ of $\hat{\mcalM}$,
where $\mcalU$ is endowed with local coordinates $\bsv=(v^1,\dots, v^m)$,
and denote $(\theta_{1},\dots,\theta_{m})$ the family of super-variables on $\bbR^{0|m}$ dual to them.
Then
$
  \{v^{\alpha,s},\theta_{\alpha}^s|1\leq\afa\leq m,\,s\geq 0\}
$
forms a family of local coordinates on the infinite jet bundle $J^{\infty}(\hat{\mcalM})$,
where $\theta_\afa^0:=\theta_\afa$.
Note that the super-variables $\theta_{\afa}^s$ satisfy the anti-commutation relations
$\theta_{\afa}^s\theta_\beta^t + \theta_\beta^t\theta_\afa^s = 0$.

A smooth function $f\in C^{\infty}(J^{\infty }(\hat \mcalU))$ is called a differential polynomial
if it depends polynomially on the jet variables $\{v^{\alpha,s+1},\theta^s_\alpha\}_{1\leq\afa\leq m, s\geq 0}$.
The ring of differential polynomials is equipped with a natural grading $\deg_\p$,
called the \textit{differential grading}, generated by
\[
  \deg_\p f=0,\quad \deg_\p v^{\afa,s}=\deg_\p \theta_\afa^s = s
\]
for all $f\in C^\infty(\mcalU),\, 1\leq\afa\leq m$ and $s\geq 0$.
We denote $\hat\mcalA_d$ the space of homogeneous differential polynomials $f$ satisfying $\deg_\p f=d$ for $d\geq 0$.
In analogy with \eqref{completion},
one defines the completed differential polynomial ring $\hat\mcalA$,
whose elements can be written as the formal series of the form
$
   \sum_{d=0}^{\infty}\veps^df_d
$ such that $f_d\in\hat\mcalA_d$.
There is another grading $\deg_\theta$ on $\hat\mcalA$, called the \textit{super grading}, generated by
\[
  \deg_\theta f=\deg_\theta v^{\afa,s} = 0,\quad
  \deg_\theta\theta_\afa^s = 1
\]
for all $f\in C^\infty(\mcalU)$, $1\leq\afa\leq m$ and $s\geq 0$,
and then
denote $\hat\mcalA^d$ the space of homogeneous differential polynomials $f$ satisfying $\deg_\theta f=d$ for $d\geq 0$.

The operator $\p_x$ in \eqref{operator px} extends to an element of $\Der(\hat\mcalA)$ via the formula
	\[
	\partial_x=
\sum_{\afa=1}^{m}
\sum_{s\geq0}
\left(v^{\alpha,s+1}\frac{\partial}{\partial v^{\alpha,s}}+\theta_{\alpha}^{s+1}\frac{\partial}{\partial \theta_{\alpha}^s}\right).
	\]
Each element $f+\p_x\hat\mcalA$ of the quotient space
$\hat \mcalF := \hat\mcalA/\p_x\hat\mcalA$ is written as $\int f \td x$.
The integration by parts formula analogous to \eqref{by parts-260329} also holds true in $\hat\mcalF$.
We also denote the homogeneous component $\hat\mcalF^p:=\int\hat\mcalA^p\td x\subseteq\hat\mcalF$ for all $p\geq 0$.
For $F=\int f\td x\in\hat\mcalF$, its \textit{variational derivatives}
$\frac{\delta F}{\delta v^\afa}$ and  $\frac{\delta F}{\delta\theta_\afa}$
are defined by
\begin{align}
		\frac{\delta F}{\delta v^\alpha} = \sum_{s \geq 0} (-\partial_x)^s \frac{\partial f}{\partial v^{\alpha,s}},
		\quad
		\frac{\delta F}{\delta \theta_\alpha} = \sum_{s \geq 0} (-\partial_x)^s \frac{\partial f}{\partial \theta^s_\alpha},
\end{align}
which are uniquely determined by the relation
\begin{equation}\label{prop of VD}
  \delta F = \int\sum_{\afa=1}^{m}
    \left(
      \delta v^\afa\cdot\frac{\delta F}{\delta v^\afa}
     +\delta\theta_\afa\cdot\frac{\delta F}{\delta\theta_\afa}
    \right)\td x.
\end{equation}

There is a bilinear map, called the \textit{Schouten bracket} on $\hat\mcalF$, which is defined by
\begin{align*}
  [,]\colon \hat\mcalF^p\times\hat\mcalF^q \to\, \hat\mcalF^{p+q-1},\quad 
  [F, G] :=\,
    \int\sum_{\afa=1}^{m}
      \left(
        \frac{\delta F}{\delta\theta_\afa}
        \frac{\delta G}{\delta v^\afa}
       +(-1)^p\frac{\delta F}{\delta v^\afa}
       \frac{\delta G}{\delta\theta_\afa}
      \right)\td x
\end{align*}
for all $ F\in\hat\mcalF^p$ and $ G\in\hat\mcalF^q$.
This bracket satisfies the following (modified) graded commutative law and graded Jacobi identity
\begin{align}
  [P, Q] &=\, (-1)^{pq}[Q,P], \label{graded commutative law}\\
  (-1)^{p-1}[[P,Q],R] &=\, [P, [Q,R]] - (-1)^{(p-1)(q-1)}[Q,[P,R]] \label{graded Jacobi identity}
\end{align}
for all $P\in\hat\mcalF^p$,  $Q\in\hat\mcalF^q$ and $R\in\hat\mcalF$.

For each skew-symmetric matrix-valued operator
$\mcalP\in\mcalA\fps{\veps\p_x}\otimes\bbC^{m\times m}$,
one introduce the corresponding local functional $\iota(\mcalP)\in\hat\mcalF^2$ via the formula
	\begin{align}\label{250815-1422}
		\iota(\mcalP) :=
    \frac{1}{2} \int
      \sum_{\afa,\beta=1}^{m}
    \theta_\alpha \left(\mcalP^{\afa\beta}\theta_\beta \right) \td x.
	\end{align}
It has been proved that (see \cite{Liu lecture notes} for example) the bracket
\[
\{,\}\colon \mcalF\times\mcalF\to\mcalF,\quad
(F,G)\mapsto \int \sum_{\afa,\beta=1}^{m}
  \frac{\delta F}{\delta v^\afa}
  \left(
    \mcalP^{\afa\beta}\frac{\delta G}{\delta v^\beta}
  \right)\td x,\quad \forall\, F,G\in\mcalF
\]
satisfies the Jacobi identity, i.e. $\mcalP$ is a Hamiltonian structure, if and only if
the following Schouten bracket vanishes:
$
  [\iota(\mcalP),\iota(\mcalP)]=0\in\hat\mcalF^3
$.

Now, in our case of $(n,1)$-type RR2T, we set $m:=n+1$.
The super-variables dual to the local coordinates
$(U_1,\dots,U_n; W)$
are denoted by $(\theta_1,\dots,\theta_n; \fai)$.
We emphasize that \eqref{short notation} will be used consistently,
and in particular, $U_{n+1}:=0$. For each $f\in\hat\mcalA$ and $s\in\bbZ$, denote
\[
   f^{[s]}:=\Lmd^s f = \rme^{\veps s\p_x} f = \sum_{k\geq 0}
  \frac{(\veps s)^k}{k!}\p_x^k f \in \hat\mcalA.
\]
In particular, we have discrete jet variables
$\{U_\afa^{[s]},\,W^{[s]},\,\theta_\afa^{[s]},\,\fai^{[s]}\}_{1\leq\afa\leq n,\,s\in\bbZ}$.
We now state the main theorem of this section.
\begin{thm}
The operators $\mcalP_1,\mcalP_2$ introduced in Theorem \ref{thm: HamP1} and Theorem \ref{thm: operator P2}
form a bihamiltonian structure of the RR2T of type $(n,1)$.
\end{thm}
\begin{proof}
  Denote $I_a:=\iota(\mcalP_a)\in\hat\mcalF^2$ the corresponding local functionals for $a=1,2$.
  It is sufficient to show that
$
  [I_1,I_1]=[I_2,I_2]=[I_1,I_2]=0
$,
which will be verified in Sect.\,\ref{subsection P1 ham}, \ref{subsection P2 ham} and \ref{subsection P1P2 compatible}, in turn.
Hence the theorem is proved.
\end{proof}

\subsection{$\mcalP_1$ is a Hamiltonian structure}
\label{subsection P1 ham}

Now we proceed to verify that $[I_1,I_1]=0$.
By straightforward calculation, we have
  \begin{align}
    I_1
  &=\,
    \frac12\int\left(
      \sum_{\afa,\beta=1}^{n}
        \theta_\afa(\mcalP_1^{\afa\beta}\theta_\beta)
       +\sum_{\afa=1}^{n}
         \left(
           \theta_\afa(\mcalP_1^{\afa,n+1}\fai)
          +\fai(\mcalP_1^{n+1,\afa}\theta_\afa)
         \right)
       +\fai(\mcalP_1^{n+1,n+1}\fai)
    \right)\td x\notag \\
  &=\,
    \sum_{\afa,\beta=1}^{n-1}
      \int U_{\afa+\beta-n}\theta_\afa\theta_\beta^{[n-\afa]}\td x
  -\sum_{\afa,\beta=1}^{n}\int
      U_{\afa+\beta-n-1}\theta_\afa\tilde\theta_\beta^{[n+1-\afa]}\td x \notag\\
&\quad
  +\int\theta_n\cdot (\Lmd-1)\tilde\fai\td x, \label{I1}
  \end{align}
here we introduce the notation
  \begin{equation}\label{tilde theta fai}
    \tilde\theta_\afa := W\theta_\afa,\quad
    \tilde\fai := W\fai
  \end{equation}
for all $1\leq\afa\leq n$,
and its variational derivatives have the form
\begin{equation}
  \frac{\delta I_1}{\delta\fai} = W(1-\Lmd^{-1})\theta_n,\quad
  \frac{\delta I_1}{\delta W} =
    \sum_{\afa,\beta=1}^{n}\theta_\afa(U_{\afa+\beta-n-1}\theta_\beta)^{[\beta-n-1]}
   +\fai(1-\Lmd^{-1})\theta_n,
\end{equation}
\begin{equation}
\frac{\delta I_1}{\delta\theta_n}
=\sum_{\afa=1}^{n}
  \left(
    W(U_{\afa-1}\theta_\afa)^{[\afa-n-1]}
   -U_{\afa-1}\tilde\theta_\afa^{[1]}
  \right) + (\Lmd-1)\tilde\fai, \quad
\frac{\delta I_1}{\delta U_n}=0,
\end{equation}
\begin{align}
\frac{\delta I_1}{\delta\theta_\gamma}
&=\,
  \sum_{\afa=1}^{n-1}
  \left(
    U_{\gamma+\afa-n}\theta_\afa^{[n-\gamma]}
    -
   \big(U_{\gamma+\afa-n}\theta_\afa\big)^{[\afa-n]}
  \right)
\notag\\
&\quad
  +\sum_{\afa=1}^{n}
  \left(
     W\big(U_{\gamma+\afa-n-1}\theta_\afa\big)^{[\afa-n-1]}
    -
     U_{\gamma+\afa-n-1}\tilde\theta_\afa^{[n+1-\gamma]}
  \right),
\\
  \frac{\delta I_1}{\delta U_\gamma}
&=\,
  \sum_{\beta=\gamma+1}^{n-1}
    \theta_\beta\theta_{\gamma+n-\beta}^{[n-\beta]}
 -\sum_{\beta=\gamma+1}^{n}
    \theta_\beta\tilde\theta_{\gamma+n+1-\beta}^{[n+1-\beta]}
\end{align}
for each $\gamma\in\{1,\dots,n-1\}$.
It is clear that $\int \frac{\delta I_1}{\delta\theta_n}
  \frac{\delta I_1}{\delta U_n}\td x=0$, and
\begin{align}
&\,
  \int\frac{\delta I_1}{\delta\fai}\frac{\delta I_1}{\delta W}\td x =
  \int W(1-\Lmd^{-1})\theta_n\cdot
    \sum_{\afa,\beta=1}^{n}
      \theta_\afa\big(U_{\afa+\beta-n-1}\theta_\beta\big)^{[\beta-n-1]}\td x
\notag\\
=&\,
  \int\sum_{\afa,\beta=1}^{n}
    U_{\afa+\beta-n-1}\theta_\beta
    \left(
       (1-\Lmd^{-1})\theta_n^{[n+1-\beta]}
    \right)
    \tilde\theta_\afa^{[n+1-\beta]}\td x
\notag\\
=&\,
  \int\sum_{k=0}^{n-2}
  \sum_{\afa=k+1}^{n}
    U_k\theta_\afa
    \left(
       (1-\Lmd^{-1})\theta_n^{[n+1-\afa]}
    \right)
    \tilde\theta_{k+n+1-\afa}^{[n+1-\afa]}\td x. \label{250812-1548-1}
\end{align}
Moreover, we also have 
$
  \int\sum_{\gamma=1}^{n-1}
    \frac{\delta I_1}{\delta\theta_\gamma}
    \frac{\delta I_1}{\delta U_\gamma}\td x
=
  X_0 + X_1 + X_2
$, 
where
\begin{align*}
  X_0&:=\,
    \int
      \sum_{\gamma=1}^{n-1}
      \sum_{\afa=1}^{n-1}
      \sum_{\beta=\gamma+1}^{n-1}
        \left(
           U_{\gamma+\afa-n}\theta_\afa^{[n-\gamma]}
             -
           \big(U_{\gamma+\afa-n}\theta_\afa\big)^{[\afa-n]}
        \right)
        \theta_\beta\theta_{\gamma+n-\beta}^{[n-\beta]}\td x,
\\
  X_1&:=\,
    \int
      \sum_{\gamma=1}^{n-1}
      \sum_{\afa=1}^{n}
      \sum_{\beta=\gamma+1}^{n-1}
        \left(
           W\big(U_{\gamma+\afa-n-1}\theta_\afa\big)^{[\afa-n-1]}
             -
           U_{\gamma+\afa-n-1}\tilde\theta_\afa^{[n+1-\gamma]}
        \right)
        \theta_\beta\theta_{\gamma+n-\beta}^{[n-\beta]}\td x
\\
&\qquad
  -\int
      \sum_{\gamma=1}^{n-1}
      \sum_{\afa=1}^{n-1}
      \sum_{\beta=\gamma+1}^{n}
        \left(
           U_{\gamma+\afa-n}\theta_\afa^{[n-\gamma]}
             -
           \big(U_{\gamma+\afa-n}\theta_\afa\big)^{[\afa-n]}
        \right)
        \theta_\beta\tilde\theta_{\gamma+n+1-\beta}^{[n+1-\beta]}\td x,
\\
X_2&:=\,
  -\int
    \sum_{\gamma=1}^{n-1}
    \sum_{\afa=1}^{n}
    \sum_{\beta=\gamma+1}^{n}
      \left(
           W\big(U_{\gamma+\afa-n-1}\theta_\afa\big)^{[\afa-n-1]}
             -
           U_{\gamma+\afa-n-1}\tilde\theta_\afa^{[n+1-\gamma]}
      \right)
      \theta_\beta\tilde\theta_{\gamma+n+1-\beta}^{[n+1-\beta]}\td x.
\end{align*}
Therefore by straightforward calculation and using the following changes of indices
\[
  \begin{cases}
    \gamma'=k+\beta-\gamma, \\
    \beta' =n+k-\gamma,
  \end{cases}
\quad
  \begin{cases}
    \gamma''=k+\beta-\gamma, \\
    \beta'' =n+1+k-\gamma,
  \end{cases}
\]
we obtain
\begin{align}
&\, X_0+X_2\notag\\
=&\,
  \int\sum_{\gamma=1}^{n-1}
      \sum_{\afa=1}^{n-1}
      \sum_{\beta=\gamma+1}^{n-1}
      U_{\gamma+\afa-n}
      \left(
        \theta_\afa^{[n-\gamma]}
        \theta_\beta
        \theta_{\gamma+n-\beta}^{[n-\beta]}
       -\theta_\afa
        \theta_\beta^{[n-\afa]}
        \theta_{\gamma+n-\beta}^{[2n-\afa-\beta]}
      \right)\td x
\notag\\
&\quad+
  \int\sum_{\gamma=1}^{n-1}
      \sum_{\afa=1}^{n}
      \sum_{\beta=\gamma+1}^{n}
      U_{\gamma+\afa-n-1}
      \left(
        \theta_\beta
        \tilde\theta_{n+1+\gamma-\beta}^{[n+1-\beta]}
        \tilde\theta_{\afa}^{[n+1-\gamma]}
       -\theta_\afa
        \tilde\theta_\beta^{[n+1-\afa]}
        \tilde\theta_{n+1+\gamma-\beta}^{[2n+2-\afa-\beta]}
      \right)\td x
\notag\\
=&\,
  \sum_{k=0}^{n-2}\int U_k
  \left(
    \sum_{\gamma=k+1}^{n-2}
    \sum_{\beta=\gamma+1}^{n-1}
        \theta_\beta
        \theta_{n+\gamma-\beta}^{[n-\beta]}
        \theta_{n+k-\gamma}^{[n-\gamma]}
    -
    \sum_{\gamma'=k+1}^{n-2}
    \sum_{\beta'=\gamma'+1}^{n-1}
        \theta_{\beta'}
        \theta_{n+\gamma'-\beta'}^{[n-\beta']}
        \theta_{n+k-\gamma'}^{[n-\gamma']}
    \right)\td x
\notag\\
&\quad+
  \sum_{k=0}^{n-2}\int U_k
  \left(
    \sum_{\gamma=k+1}^{n-1}
    \sum_{\beta=\gamma+1}^{n}
        \theta_\beta
        \tilde\theta_{\gamma+n+1-\beta}^{[n+1-\beta]}
        \tilde\theta_{n+1+k-\gamma}^{[n+1-\gamma]}
  \right.
\notag\\
&\quad\quad\quad -
  \left.
    \sum_{\gamma''=k+1}^{n-1}
    \sum_{\beta''=\gamma''+1}^{n}
        \theta_{\beta''}
        \tilde\theta_{\gamma''+n+1-\beta''}^{[n+1-\beta'']}
        \tilde\theta_{n+1+k-\gamma''}^{[n+1-\gamma'']}
  \right)\td x
\notag\\
=&\,
  0+0=0,
\end{align}
and
\begin{align}
  X_1
&=\,
  -\sum_{\gamma=1}^{n-1}\sum_{\afa=1}^{n}\sum_{\beta=\gamma+1}^{n-1}
     \int U_{\gamma+\afa-n-1}\theta_\afa
       \theta_{n+\gamma-\beta}^{[2n+1-\afa-\beta]}
       \tilde\theta_\beta^{[n+1-\afa]}
     \td x
\notag\\
&\qquad
  -\sum_{\gamma=1}^{n-1}\sum_{\afa=\gamma+1}^{n-1}\sum_{\beta=1}^{n}
     \int U_{\gamma+\beta-n-1}\theta_\afa
       \theta_{n+\gamma-\afa}^{[n-\afa]}
       \tilde\theta_\beta^{[n+1-\gamma]}
     \td x
\notag\\
&\qquad
  +\sum_{\gamma=1}^{n-1}\sum_{\afa=1}^{n-1}\sum_{\beta=\gamma+1}^{n}
     \int U_{\gamma+\afa-n}\theta_\afa
       \theta_{\beta}^{[n-\afa]}
       \tilde\theta_{n+1+\gamma-\beta}^{[2n+1-\afa-\beta]}
     \td x
\notag\\
&\qquad
  +\sum_{\gamma=1}^{n-1}\sum_{\afa=\gamma+1}^{n}\sum_{\beta=1}^{n-1}
     \int U_{\gamma+\beta-n}\theta_\afa
       \theta_{\beta}^{[n-\gamma]}
       \tilde\theta_{n+1+\gamma-\afa}^{[n+1-\afa]}
     \td x
\notag\\
&=\,
  -\sum_{k=0}^{n-2}\sum_{\afa=k+2}^{n}
    \int U_k\theta_\afa
      \sum_{\gamma=k+1}^{\afa-2}
        \theta_{n+k-\gamma}^{[n-\gamma]}
        \tilde\theta_{n+1+\gamma-\afa}^{[n+1-\afa]}\td x
\notag\\&\qquad
  -\sum_{k=0}^{n-2}\sum_{\afa=k+2}^{n-1}
    \int U_k\theta_\afa
      \sum_{\beta=n+1+k-\afa}^{n-1}
        \theta_{\beta}^{[n-\afa]}
        \tilde\theta_{2n+1+k-\afa-\beta}^{[2n+1-\afa-\beta]}\td x
\notag\\&\qquad
  +\sum_{k=0}^{n-2}\sum_{\afa=k+1}^{n-1}
    \int U_k\theta_\afa
      \sum_{\beta=n+1+k-\afa}^{n}
        \theta_{\beta}^{[n-\afa]}
        \tilde\theta_{2n+1+k-\afa-\beta}^{[2n+1-\afa-\beta]}\td x
\notag\\&\qquad
  +\sum_{k=0}^{n-2}\sum_{\afa=k+2}^{n}
    \int U_k\theta_\afa
      \sum_{\gamma=k+1}^{\afa-1}
        \theta_{n+k-\gamma}^{[n-\gamma]}
        \tilde\theta_{n+1+\gamma-\afa}^{[n+1-\afa]}\td x
\notag\\
&=\,
  \int
    \sum_{k=0}^{n-2}\sum_{\afa=k+1}^{n}
      U_k\theta_\afa
      \left(
        (\Lmd^{-1}-1)\theta_n^{[n+1-\afa]}
      \right)
      \tilde\theta_{n+1+k-\afa}^{[n+1-\afa]} \td x.  \label{250812-1548-2}
\end{align}
Therefore, from \eqref{250812-1548-1}--\eqref{250812-1548-2} we obtain 
\begin{equation}
  [I_1,I_1] =
2\int\left(
  \frac{\delta I_1}{\delta\fai}
  \frac{\delta I_1}{\delta W}
 +\frac{\delta I_1}{\delta\theta_n}
  \frac{\delta I_1}{\delta U_n}
 +\sum_{\gamma=1}^{n-1}
  \frac{\delta I_1}{\delta\theta_\gamma}
  \frac{\delta I_1}{\delta U_\gamma}
\right)\td x = 0.
\end{equation}

\subsection{$\mcalP_2$ is a Hamiltonian structure}
\label{subsection P2 ham}

To compute $[I_1, I_2]$ and $[I_2, I_2]$, 
we need to deal with operators of the form $\mathsf R(\Lmd)=\mathsf R(\rme^{\veps\p_x})$,
where $\mathsf R\in\bbR(x)$ is a rational function,
and then find the variational derivatives of $I_2$.
Notice that the adjoint of $\mathsf R(\Lmd)$ for $\mathsf R\in\bbR(x)$ has the form
  \[
    \mathsf R(\Lmd)^\dag = \mathsf R(\rme^{\veps\p_x})^\dag = \mathsf R(\rme^{-\veps\p_x}) = \mathsf R(\Lmd^{-1}).
  \]
  Then using \eqref{250815-1524} and the integration by parts formula analogous to \eqref{by parts-260329},
  we have
  \begin{equation}
    I_2 = I_{2,1} + I_{2,2},
  \end{equation}
where
  \begin{align}
    I_{2,1}&:=\,
      \frac12\int\sum_{\afa,\beta=1}^{n-1}
        \theta_\afa U_\afa
        \frac{(1-\Lmd^{n-\afa})(1-\Lmd^\beta)}{1-\Lmd^n}(U_\beta\theta_\beta)\td x \notag\\
    &\quad
      +\frac12
        \int\sum_{\afa,\beta=1}^{n-1}
        \sum_{s=1}^{\beta}
          \theta_\afa
            \left(
              U_{\beta-s}\Lmd^sU_{\afa+s}
             -U_{\afa+s}\Lmd^{\beta-\afa-s}U_{\beta-s}
            \right)\theta_\beta\td x,  \notag
  \\
    I_{2,2}&:=\,
      \int\sum_{\afa=1}^{n}
        \theta_\afa U_\afa
        \frac{(1-\Lmd)(\Lmd^n-\Lmd^{n-\afa})}{1-\Lmd^n}(W\fai)\td x
    +
      \int W\fai\frac{1-\Lmd^{-1}}{1-\Lmd^n}(W\fai)\td x. \notag
  \end{align}
Then the variational derivatives of $I_2=I_{2,1}+I_{2,2}$
can be computed by using \eqref{prop of VD}.
By straightforward calculation, we obtain
\begin{align}
  \frac{\delta I_2}{\delta\theta_n}
&=\,
  U_n \left(\Lmd -1 \right) \left( W\fai\right) ,  \label{260119-1659-1}
\\
  \frac{\delta I_2}{\delta U_n}
&=\,
  \theta_n \left(\Lmd -1 \right) \left( W\fai\right)
 -\sum_{\afa,\beta = 1}^{n-1}
   \theta_\afa
     \left(
       U_{\afa+\beta-n}\theta_\beta
     \right)^{[\beta - n]},   \label{260119-1659-2}
\\
  \frac{\delta I_2}{\delta\fai}
&=\,
  W \frac{1-\Lmd^{-1}}{1-\Lmd^n}
  \sum_{\afa=1}^{n}
    \left(1-\Lmd^\afa\right)
    \left( U_\afa\theta_\afa\right)
 +W\frac{(1-\Lmd)(\Lmd^n-\Lmd^{-1})}{1-\Lmd^n}
   \left(W\fai\right),  \label{260119-1659-3}
\\
  \frac{\delta I_2}{\delta W}
&=\,
  \fai \frac{1-\Lmd^{-1}}{1-\Lmd^n}
  \sum_{\afa=1}^{n}
    \left(1-\Lmd^\afa\right)
    \left( U_\afa\theta_\afa\right)
 +\fai\frac{(1-\Lmd)(\Lmd^n-\Lmd^{-1})}{1-\Lmd^n}
   \left(W\fai\right). \label{260119-1659-4}
\end{align}

Moreover, we also have the following lemma:
\begin{lem}
The following equations hold true for all $1\leq\gamma\leq n-1$:
\begin{align}
  \frac{\delta I_2}{\delta\theta_\gamma}
&=\,
    U_\gamma \xi_\gamma
  + \sum_{\afa=1}^{n-1}
    \left(
      \sum_{s=0}^{\gamma}-\sum_{s=\afa}^{\afa+\gamma}
    \right)
      U_s\left(U_{\gamma+\afa-s}\theta_\afa\right)^{[\afa-s]}, \label{del I2 del theta gamma}
\\
  \frac{\delta I_2}{\delta U_\gamma} &=\,
      \theta_\gamma \xi_\gamma
  +\left(
    \sum_{\beta,\beta'=\gamma}^{n-1}
   -\sum_{\beta,\beta'=1}^{\gamma-1}
  \right)
    \theta_\beta
    \left(
      U_{\beta+\beta'-\gamma}\theta_{\beta'}
    \right)^{[\beta'-\gamma]} , \label{del I2 del U gamma}
\end{align}
here for each $1\leq\gamma\leq n-1$, we denote
  \begin{equation} \label{def of xi gamma}
    \xi_\gamma:=
      \frac{(1-\Lmd)(\Lmd^n-\Lmd^{n-\gamma})}{1-\Lmd^n}(W\fai)
      +\sum_{\afa=1}^{n-1}
        \frac{(1-\Lmd^\afa)(1-\Lmd^{n-\gamma})}{1-\Lmd^n}(U_\afa\theta_\afa).
  \end{equation}
\end{lem}

\begin{proof}
For each $1\leq\gamma\leq n-1$, it is clear that
\begin{equation}\label{250815-1647-1}
  \frac{\delta I_{2,2}}{\delta\theta_\gamma}
=
  U_\gamma\frac{(1-\Lmd)(\Lmd^n-\Lmd^{n-\gamma})}{1-\Lmd^n}(W\fai),
\quad
  \frac{\delta I_{2,2}}{\delta U_\gamma}
=
  \theta_\gamma
  \frac{(1-\Lmd)(\Lmd^n-\Lmd^{n-\gamma})}{1-\Lmd^n}(W\fai).
\end{equation}
Moreover, by straightforward calculation, we obtain
\begin{align}
\frac{\delta I_{2,1}}{\delta\theta_\gamma}
&=\,
  \frac12\sum_{\beta=1}^{n-1}
    U_\gamma\frac{(1-\Lmd^{n-\gamma})(1-\Lmd^\beta)}{1-\Lmd^n}(U_\beta\theta_\beta)
  -\frac12\sum_{\afa=1}^{n-1}
    U_\gamma
    \left(
      \frac{(1-\Lmd^{n-\afa})(1-\Lmd^\gamma)}{1-\Lmd^n}
    \right)^\dag(U_\afa\theta_\afa) \notag \\
&\quad
  +\frac12\sum_{\beta=1}^{n-1}\sum_{s=1}^{\beta}
    \left(
      U_{\beta-s}\Lmd^s U_{\gamma+s}
     -U_{\gamma+s}\Lmd^{\beta-\gamma-s}U_{\beta-s}
    \right)\theta_\beta \notag\\
&\quad
  -\frac12\sum_{\afa=1}^{n-1}\sum_{s=1}^{\gamma}
   \left(
     U_{\gamma-s}\Lmd^s U_{\afa+s}
    -U_{\afa+s}\Lmd^{\gamma-\afa-s} U_{\gamma-s}
   \right)^\dag\theta_\afa \notag\\
=&\,
  \sum_{\afa=1}^{n-1}U_\gamma
    \frac{(1-\Lmd^\afa)(1-\Lmd^{n-\gamma})}{1-\Lmd^n}(U_\afa\theta_\afa)
+
  \sum_{\afa=1}^{n-1}
    \left(
      \sum_{s=0}^{\gamma}-\sum_{s=\afa}^{\gamma+\afa}
    \right)U_s
    \left(
      U_{\gamma+\afa-s}\theta_\afa
    \right)^{[\afa-s]}.  \label{250815-1647-2}
\end{align}
By a similar argument, we can also obtain
\begin{equation}
\frac{\delta I_{2,1}}{\delta U_\gamma}
=
\sum_{\afa=1}^{n-1}
  \theta_\gamma
  \frac{(1-\Lmd^\afa)(1-\Lmd^{n-\gamma})}{1-\Lmd^n}(U_\afa\theta_\afa)
+
  \left(
    \sum_{\beta,\beta'=\gamma}^{n-1}
   -\sum_{\beta,\beta'=1}^{\gamma-1}
  \right)
    \theta_\beta
    \left(
      U_{\beta+\beta'-\gamma}\theta_{\beta'}
    \right)^{[\beta'-\gamma]}. \label{250815-1647-3}
\end{equation}
This lemma follows from \eqref{250815-1647-1}--\eqref{250815-1647-3}.
\end{proof}

Now we are to compute the Schouten bracket $[I_2, I_2]$. If we introduce
\begin{align}
  A_\gamma &:=\,
    \sum_{\afa=1}^{n-1}
      \left(
        \sum_{s=0}^{\gamma}-\sum_{s=\afa}^{\afa+\gamma}
      \right)
        U_s\left(U_{\gamma+\afa-s}\theta_\afa\right)^{[\afa-s]}, \label{def of A gamma}\\
  B_\gamma &:=\,
    \left(
      \sum_{\beta,\beta'=\gamma}^{n-1}
     -\sum_{\beta,\beta'=1}^{\gamma-1}
    \right)
      \theta_\beta\left(U_{\beta+\beta'-\gamma}\theta_{\beta'}\right)^{[\beta'-\gamma]}, \label{def of B gamma}
\end{align}
\begin{equation}\label{Theta and Psi}
  \Theta:=\sum_{\afa=1}^{n-1}
    \frac{1-\Lmd^\afa}{1-\Lmd^n}
      \left(U_\afa\theta_\afa\right),
  \quad
  \Psi:=\frac{1-\Lmd}{\Lmd^{-n}-1}(W\fai)
\end{equation}
for each $1\leq\gamma\leq n-1$, then $A_\gamma, \Theta, \Psi\in\hat\mcalA^1$ are odd variables, $B_\gamma\in\hat\mcalA^2$ are even variables,
and equations \eqref{del I2 del theta gamma}--\eqref{def of xi gamma} can be rewritten as
\begin{equation}\label{260119-1730-1}
\frac{\delta I_2}{\delta\theta_\gamma} = U_\gamma\xi_\gamma+A_\gamma,\quad
\frac{\delta I_2}{\delta U_\gamma} = \theta_\gamma\xi_\gamma+B_\gamma,
\end{equation}
\begin{equation}\label{xi = Psi and Theta}
  \xi_\gamma = \left(1-\Lmd^{-\gamma}\right)\Psi + (1-\Lmd^{n-\gamma})\Theta
\end{equation}
for $1\leq\gamma\leq n-1$.
The following two lemmas are useful:

\begin{lem}
  Under the above notations, the following identities hold true:
  \begin{align}
    \sum_{\gamma=1}^{n-1}\left(U_\gamma B_\gamma-\theta_\gamma A_\gamma\right)
  &=\,
    \sum_{\gamma,\afa=1}^{n-1}
      \left(
        U_n\theta_\gamma\left(U_{\gamma+\afa-n}\theta_\afa\right)^{[\afa-n]}
       -\theta_\gamma\left(U_{\gamma+\afa}\theta_\afa\right)^{[\afa]}
      \right),\label{UBTA-1}
\\
    \sum_{\gamma=1}^{n-1}\left(U_\gamma B_\gamma-\theta_\gamma A_\gamma\right)^{[\gamma]}
  &=\,
    \sum_{\gamma,\afa=1}^{n-1}
      \left(
        \left(U_n\theta_\gamma\right)^{[n]}\left(U_{\gamma+\afa-n}\theta_\afa\right)^{[\afa]}
       -\theta_\gamma\left(U_{\gamma+\afa}\theta_\afa\right)^{[\afa]}
      \right).\label{UBTA-2}
  \end{align}
\end{lem}
\begin{proof}
    By straightforward computation, one obtains
    \begin{align*}
      \sum_{\gamma=1}^{n-1}\theta_\gamma A_\gamma
  &=\,
    \sum_{\afa,\gamma=1}^{n-1}
    \left(
        \sum_{s=0}^{\gamma}-\sum_{s=\afa}^{\afa+\gamma}
      \right)
        U_s\theta_\gamma\left(U_{\gamma+\afa-s}\theta_\afa\right)^{[\afa-s]}, \\
      \sum_{\gamma=1}^{n-1} U_\gamma B_\gamma
  &=\,
    \sum_{\afa,\gamma=1}^{n-1}
       \left(
         \sum_{s=1}^{\gamma}
        -\sum_{s=\afa+1}^{n-1}
       \right)
       U_s\theta_\gamma\left(U_{\gamma+\afa-s}\theta_{\afa}\right)^{[\afa-s]}.
    \end{align*}
Therefore, we arrive at
\begin{align*}
&\,
  \sum_{\gamma=1}^{n-1}
    \left(
      U_\gamma B_\gamma - \theta_\gamma A_\gamma
    \right)
\\
=&\,
  \sum_{\afa,\gamma=1}^{n-1}
  \left(
    -\theta_\gamma\left(U_{\gamma+\afa}\theta_\afa\right)^{[\afa]}
    +U_\afa\theta_\gamma\cdot U_\gamma\theta_\afa
  \right)
  +\sum_{\afa,\gamma=1}^{n-1}
    \left(
      \sum_{s=\afa+1}^{\afa+\gamma}
     -\sum_{s=\afa+1}^{n-1}
    \right)U_s\theta_\gamma\left(U_{\gamma+\afa-s}\theta_{\afa}\right)^{[\afa-s]}
\\
=&\,
  -\sum_{\afa,\gamma=1}^{n-1}\theta_\gamma\left(U_{\gamma+\afa}\theta_\afa\right)^{[\afa]}
  +\sum_{1\leq\afa,\gamma\leq n-1 \atop \afa+\gamma\geq n}
    U_n\theta_\gamma
      \left(
        U_{\gamma+\afa-n}\theta_\afa
      \right)^{[\afa-n]}
\\
&\quad
 -\sum_{1\leq\afa,\gamma\leq n-1 \atop \afa+\gamma < n}
  \sum_{s=\gamma+\afa+1}^{n-1}
    U_s\theta_\gamma\left(U_{\gamma+\afa-s}\theta_{\afa}\right)^{[\afa-s]}
\\
=&\,
\sum_{\gamma,\afa=1}^{n-1}
      \left(
        U_n\theta_\gamma\left(U_{\gamma+\afa-n}\theta_\afa\right)^{[\afa-n]}
       -\theta_\gamma\left(U_{\gamma+\afa}\theta_\afa\right)^{[\afa]}
      \right),
\end{align*}
hence \eqref{UBTA-1} holds true.
The equation \eqref{UBTA-2} can be verified in the same way, and the details are omitted.
\end{proof}

\begin{lem}\label{CZL main lemma}
  The following identity holds true:
  \begin{equation}\label{CZL main identity}
    \int
      \sum_{\gamma=1}^{n-1}
        A_\gamma B_\gamma \td x
  =
  -\int\sum_{\beta=1}^{n-1}
    \left(\Lmd^\beta-1\right)
    \left(U_\beta\theta_\beta\right)
    \cdot
    \sum_{\gamma,\afa=1}^{n-1}
      \theta_\gamma \left(U_{\gamma+\afa}\theta_\afa\right)^{[\afa]}
    \td x.
  \end{equation}
\end{lem}
\begin{proof}
  This lemma follows from tedious computations; see Appendix for details.
\end{proof}


From \eqref{260119-1659-1}--\eqref{260119-1659-4}, it is clear that
\begin{align}\label{260119-1729-1}
  \frac{\delta I_2}{\delta\theta_n}\frac{\delta I_2}{\delta U_n}
&=\,
  -U_n(\Lmd -1) (W\fai) \cdot
  \sum_{\afa,\beta = 1}^{n-1}
    \theta_\afa
    \left(
      U_{\afa+\beta-n}\theta_\beta
    \right)^{[\beta - n]},
\quad
  \frac{\delta I_2}{\delta\fai}\frac{\delta I_2}{\delta W} = 0.
\end{align}
Next, from \eqref{Theta and Psi} and \eqref{UBTA-1}--\eqref{UBTA-2}, we obtain
\begin{align*}
&\,
  \int \sum_{\gamma=1}^{n-1}
    \Big(
      U_\gamma\left(1-\Lmd^{n-\gamma}\right)\Theta\cdot B_\gamma
     +A_\gamma\theta_\gamma\left(1-\Lmd^{n-\gamma}\right)\Theta
    \Big)\td x
\\
=&\,
  \int \Theta\cdot\sum_{\gamma=1}^{n-1}
    \left(
      1-\Lmd^{\gamma-n}
    \right)
    \left(
      U_\gamma B_\gamma - \theta_\gamma A_\gamma
    \right)\td x
=
  \int\Theta\cdot\left(\Lmd^{-n}-1\right)
    \sum_{\gamma,\afa=1}^{n-1}
      \theta_\gamma \left(U_{\gamma+\afa}\theta_\afa\right)^{[\afa]}
  \td x
\\
=&\,
  \int\left(\Lmd^n-1\right)\Theta\cdot
    \sum_{\gamma,\afa=1}^{n-1}
      \theta_\gamma \left(U_{\gamma+\afa}\theta_\afa\right)^{[\afa]}
  \td x
=
  \int\sum_{\beta=1}^{n-1}
    \left(\Lmd^\beta-1\right)
    \left(U_\beta\theta_\beta\right)
    \cdot
    \sum_{\gamma,\afa=1}^{n-1}
      \theta_\gamma \left(U_{\gamma+\afa}\theta_\afa\right)^{[\afa]}
    \td x.
\end{align*}
Similarly, it can also be verified that
\begin{align*}
&\,
  \int\sum_{\gamma=1}^{n-1}
    \Big(
      U_\gamma\left(1-\Lmd^{-\gamma}\right)\Psi\cdot B_\gamma
     +A_\gamma \theta_\gamma\left(1-\Lmd^{-\gamma}\right)\Psi
    \Big)\td x \\
=&\,
  \int \left(\Lmd-1\right)(W\fai)
  \cdot
  \sum_{\gamma,\afa=1}^{n-1}
    U_n\theta_\gamma
    \left(
      U_{\gamma+\afa-n}\theta_\afa
    \right)^{[\afa-n]}
  \td x,
\end{align*}
then from \eqref{xi = Psi and Theta} we arrive at
\begin{align}
  \int\sum_{\gamma=1}^{n-1}
    \left(
      U_\gamma\xi_\gamma B_\gamma + A_\gamma\theta_\gamma\xi_\gamma
    \right)\td x
&=\,
  \int\sum_{\beta=1}^{n-1}
    \left(\Lmd^\beta-1\right)
    \left(U_\beta\theta_\beta\right)
    \cdot
    \sum_{\gamma,\afa=1}^{n-1}
      \theta_\gamma \left(U_{\gamma+\afa}\theta_\afa\right)^{[\afa]}
    \td x
\notag\\
&\quad
+
 \int \left(\Lmd-1\right)(W\fai)
  \cdot
  \sum_{\gamma,\afa=1}^{n-1}
    U_n\theta_\gamma
    \left(
      U_{\gamma+\afa-n}\theta_\afa
    \right)^{[\afa-n]}
  \td x. \label{260119-1732}
\end{align}

Therefore, from \eqref{260119-1730-1} and \eqref{CZL main identity}--\eqref{260119-1732}, we obtain
\begin{align*}
&\,
\frac 12 [I_2, I_2]
=
  \int
    \left(
      \frac{\delta I_2}{\delta\theta_n}
      \frac{\delta I_2}{\delta U_n}
     +
      \frac{\delta I_2}{\delta \fai}
      \frac{\delta I_2}{\delta W}
     +
      \sum_{\gamma = 1}^{n-1}
      \frac{\delta I_2}{\delta\theta_\gamma}
      \frac{\delta I_2}{\delta U_\gamma}
    \right)\td x
=0,
\end{align*}
here we notice that $\xi_\gamma\xi_\gamma = 0$, because $\xi_\gamma\in\hat\mcalA^1$ are odd variables.

\subsection{The Hamiltonian structures $\mcalP_1$ and $\mcalP_2$ are compatible}
\label{subsection P1P2 compatible}

  Let us show that $\mcalP_1$ and $\mcalP_2$ are compatible,
  i.e. $[I_1, I_2]=0$.
  Although this can also be directly verified through straightforward computation,
  it is noteworthy that by making the following observation, one can significantly reduce the computational effort.
  If we consider
  \[
    J:= \int \theta_n \td x\, \in \hat\mcalF^1,
  \]
  then from \eqref{260119-1659-2} and
  in comparison with \eqref{I1}--\eqref{tilde theta fai}, we obtain
  \begin{align}
    &\,
    [J, I_2] = \int \frac{\delta I_2}{\delta U_n} \td x
=
    I_1 + \underbrace{
      \int \sum_{\afa,\beta = 1}^{n}
        U_{\afa+\beta-n-1}\theta_\afa\tilde\theta_\beta^{[n+1-\afa]}\td x
      }_{:=\tilde I_1}.
  \end{align}
  Then $[I_2, I_2]=0$ and the properties \eqref{graded commutative law}--\eqref{graded Jacobi identity} of the Schouten bracket imply
  \begin{align}\label{J-Schouten}
    0&=\,\frac12[J,[I_2, I_2]]
  = [[J,I_2],I_2] = [I_1+\tilde I_1, I_2].
  \end{align}
Therefore we only need to verify that $[\tilde I_1, I_2]=0$.

Now we proceed to compute $[\tilde I_1, I_2]$.
The variational derivatives of $\tilde I_1$ are given by
\begin{align}
\frac{\delta\tilde I_1}{\delta\fai} &= 0,\quad
\frac{\delta\tilde I_1}{\delta W} =
  -\sum_{\afa,\beta=1}^{n}
    \theta_\afa
    \left(U_{\afa+\beta-n-1}\theta_\beta\right)^{[\beta-n-1]},
\label{var der til I1 - 1}\\
  \frac{\delta\tilde I_1}{\delta\theta_n}
&=\,
  \sum_{\afa=1}^{n}
    \left(
      U_{\afa-1}\tilde\theta_\afa^{[1]}
     -W\left(U_{\afa-1}\theta_\afa\right)^{[\afa-n-1]}
    \right),\quad
\frac{\delta\tilde I_1}{\delta U_n} = 0,  \label{var der til I1 - 2}
\\
  \frac{\delta\tilde I_1}{\delta\theta_\gamma}
&=\,
  \sum_{\afa=1}^{n}
    \left(
      U_{\gamma+\afa-n-1}\tilde\theta_\afa^{[n+1-\gamma]}
     -W\left(U_{\gamma+\afa-n-1}\theta_\afa\right)^{[\afa-n-1]}
    \right),
\\
  \frac{\delta\tilde I_1}{\delta U_\gamma}
&=\,
  \sum_{\afa=\gamma+1}^{n}
    \theta_\afa\tilde\theta_{n+1+\gamma-\afa}^{[n+1-\afa]} \label{var der til I1 - 4}
\end{align}
for all $1\leq\gamma\leq n-1$, where the short notation $\tilde\theta_\afa:= W\theta_\afa$ is introduced in \eqref{tilde theta fai}.
Then by straightforward computation, we have
\begin{align*}
  \sum_{\gamma=1}^{n-1}
  \left(
    \frac{\delta\tilde I_1}{\delta\theta_\gamma}\theta_\gamma
   +\frac{\delta\tilde I_1}{\delta U_\gamma} U_\gamma
  \right)
&=\,
  -\sum_{\gamma=1}^{n-1}\sum_{\afa=1}^{n}
    \left( U_{\gamma+\afa-n-1}\theta_\afa\right)^{[\afa-n-1]}\tilde\theta_\gamma
   +\sum_{\afa=1}^{n}
   \left( U_{\afa-1}\theta_n\tilde\theta_\afa^{[1]} - \theta_{n+1-\afa}\tilde\theta_\afa^{[\afa]}\right),
\\
  \sum_{\gamma=1}^{n-1}
  \left(
    \frac{\delta\tilde I_1}{\delta\theta_\gamma}\theta_\gamma
   +\frac{\delta\tilde I_1}{\delta U_\gamma} U_\gamma
  \right)^{[\gamma]}
&=\,
  -\sum_{\gamma=1}^{n-1}\sum_{\afa=1}^{n}
    \left( U_{\gamma+\afa-n-1}\theta_\gamma\right)^{[\gamma]}\tilde\theta_\afa^{[n+1]} \\
&\qquad
   +\sum_{\afa=1}^{n}
   \left(
     \left(U_{\afa-1}\theta_\afa\right)^{[\afa-1]}\tilde\theta_n^{[n]} -
     \theta_{n+1-\afa}\tilde\theta_\afa^{[\afa]}\right),
\end{align*}
and therefore, it can be verified directly that
\begin{align}
&\,
  \left(
    1-\Lmd^{n+1}\right)
    \left(\frac{\delta\tilde I_1}{\delta W}W\right)
  +
  \sum_{\gamma=1}^{n-1}
    \left( 1-\Lmd^\gamma\right)
    \left(
    \frac{\delta\tilde I_1}{\delta\theta_\gamma}\theta_\gamma
   +\frac{\delta\tilde I_1}{\delta U_\gamma} U_\gamma
  \right) \notag \\
=&\,
  \left(
    1-\Lmd^n
  \right)
  \left(
    \sum_{\afa=1}^{n}
      \left(
        U_{\afa-1}\theta_n\tilde\theta_\afa^{[1]}
       +\left(U_{\afa-1}\theta_\afa\right)^{[\afa-n-1]}\tilde\theta_n
      \right)
  \right), \label{260203-1640-1}
\\
  &\,
  \left(
    1-\Lmd\right)
    \left(\frac{\delta\tilde I_1}{\delta W}W\right)
  +
  \sum_{\gamma=1}^{n-1}
    \left( 1-\Lmd^{\gamma-n}\right)
    \left(
    \frac{\delta\tilde I_1}{\delta\theta_\gamma}\theta_\gamma
   +\frac{\delta\tilde I_1}{\delta U_\gamma} U_\gamma
  \right) \notag \\
=&\,
  \left(
    \Lmd^{-n}-1
  \right)
  \left(
    \sum_{\afa=1}^{n}
      \theta_{n+1-\afa}\tilde\theta_\afa^{[\afa]}
  \right). \label{260203-1640-2}
\end{align}

Let us recall the notations $\Theta, \Psi$ and $\xi_\gamma$
introduced in \eqref{def of xi gamma} and \eqref{Theta and Psi}--\eqref{xi = Psi and Theta},
which appear in the variational derivatives of $I_2$. If we additionally introduce
\[
  \xi_{n+1}:= \left(1-\Lmd^{-n-1}\right)\Psi + \left(1-\Lmd^{-1}\right)\Theta,
\]
then \eqref{260119-1659-3} can be rewritten as
\begin{equation}
  \frac{\delta I_2}{\delta\fai}
=
  W\xi_{n+1} + W\left(1-\Lmd^{-1}\right)
  \left( U_n\theta_n\right).
\end{equation}
From \eqref{var der til I1 - 2},
\eqref{260203-1640-1}--\eqref{260203-1640-2}
and the integration by parts formula analogous to \eqref{by parts-260329}, we obtain
\begin{align}
&\,
  \int\left(
    \frac{\delta\tilde I_1}{\delta W}\cdot W\xi_{n+1}
   +\sum_{\gamma=1}^{n-1}
     \left(
       \frac{\delta\tilde I_1}{\delta\theta_\gamma}\cdot \theta_\gamma\xi_\gamma
      +\frac{\delta\tilde I_1}{\delta U_\gamma}\cdot U_\gamma\xi_\gamma
     \right)
  \right)\td x \notag
\\
=&\,
  \int\left(
    \left(1-\Lmd^{n+1}\right)
    \left(\frac{\delta\tilde I_1}{\delta W}W\right)
   +\sum_{\gamma=1}^{n-1}
      \left(1-\Lmd^{\gamma}\right)
      \left(
        \frac{\delta\tilde I_1}{\delta\theta_\gamma}\theta_\gamma
       +\frac{\delta\tilde I_1}{\delta U_\gamma}U_\gamma
      \right)
  \right)\cdot\Psi\td x  \notag
\\
&\quad
  +\int
    \left(
      \left(
    1-\Lmd\right)
    \left(\frac{\delta\tilde I_1}{\delta W}W\right)
  +
  \sum_{\gamma=1}^{n-1}
    \left( 1-\Lmd^{\gamma-n}\right)
    \left(
    \frac{\delta\tilde I_1}{\delta\theta_\gamma}\theta_\gamma
   +\frac{\delta\tilde I_1}{\delta U_\gamma} U_\gamma
  \right)
    \right)\cdot\Theta\td x \notag
\\
=&\,
  \int
    \left(
    \sum_{\afa=1}^{n}
      \left(
        U_{\afa-1}\theta_n\tilde\theta_\afa^{[1]}
       +\left(U_{\afa-1}\theta_\afa\right)^{[\afa-n-1]}\tilde\theta_n
      \right)
  \right) \cdot
    \left( 1-\Lmd^{-n}\right)\Psi \td x \notag
\\
&\quad
  +\int \left(
    \sum_{\afa=1}^{n}\theta_{n+1-\afa}\tilde\theta_\afa^{[\afa]}
  \right)\cdot \left(\Lmd^n-1\right)\Theta\td x \notag\\
=&\,
  -\int\frac{\delta\tilde I_1}{\delta\theta_n}\cdot\theta_n(\Lmd-1)(W\fai)\td x
  +\int \left(
    \sum_{\afa=1}^{n}\theta_{n+1-\afa}\tilde\theta_\afa^{[\afa]}
  \right)
  \left(
    \sum_{\beta=1}^{n-1}
      \left(\Lmd^\beta - 1\right)\left( U_\beta\theta_\beta\right)
  \right)
  \td x. \label{kill Lmd 2}
\end{align}

By using
\eqref{def of A gamma}--\eqref{260119-1730-1},
\eqref{var der til I1 - 1}--\eqref{var der til I1 - 4} and \eqref{kill Lmd 2}, we have
\begin{align}
  [\tilde I_1, I_2]
&=\,
  \int\left(
    \frac{\delta\tilde I_1}{\delta\theta_n}
    \frac{\delta       I_2}{\delta     U_n}
   +\frac{\delta\tilde I_1}{\delta     U_n}
    \frac{\delta       I_2}{\delta\theta_n}
   +\frac{\delta\tilde I_1}{\delta       W}
    \frac{\delta       I_2}{\delta    \fai}
   +\frac{\delta\tilde I_1}{\delta    \fai}
    \frac{\delta       I_2}{\delta       W}
  +\sum_{\gamma=1}^{n-1}
      \left(
        \frac{\delta\tilde I_1}{\delta\theta_\gamma}
        \frac{\delta       I_2}{\delta     U_\gamma}
       +\frac{\delta\tilde I_1}{\delta     U_\gamma}
        \frac{\delta       I_2}{\delta\theta_\gamma}
      \right)
  \right)\td x \notag
\\
=&\,
  \int \frac{\delta\tilde I_1}{\delta\theta_n}\cdot\theta_n(\Lmd-1)(W\fai)\td x
 -\int \frac{\delta\tilde I_1}{\delta\theta_n}
    \sum_{\afa,\beta=1}^{n-1}
      \theta_\afa\left(U_{\afa+\beta-n}\theta_\beta\right)^{[\beta-n]}\td x \notag
\\
&\quad
  +\int\frac{\delta\tilde I_1}{\delta W} W\left(1-\Lmd^{-1}\right)\left(U_n\theta_n\right)\td x
  +\int\sum_{\gamma=1}^{n-1}
    \left(
      \frac{\delta\tilde I_1}{\delta\theta_\gamma} B_\gamma
     +\frac{\delta\tilde I_1}{\delta U_\gamma} A_\gamma
    \right) \td x \notag
\\
&\quad +
  \int\left(
    \frac{\delta\tilde I_1}{\delta W}\cdot W\xi_{n+1}
   +\sum_{\gamma=1}^{n-1}
     \left(
       \frac{\delta\tilde I_1}{\delta\theta_\gamma}\cdot \theta_\gamma\xi_\gamma
      +\frac{\delta\tilde I_1}{\delta U_\gamma}\cdot U_\gamma\xi_\gamma
     \right)
  \right)\td x \notag \\
=&\,
  \int\sum_{\gamma=1}^{n-1}
    \left(
      \frac{\delta\tilde I_1}{\delta\theta_\gamma} B_\gamma
     +\frac{\delta\tilde I_1}{\delta U_\gamma} A_\gamma
    \right) \td x
 -\int \frac{\delta\tilde I_1}{\delta\theta_n}\sum_{\afa,\beta=1}^{n-1}\theta_\afa\left(U_{\afa+\beta-n}\theta_\beta\right)^{[\beta-n]}\td x \notag
\\
&\quad
  +\int\left(
    \sum_{\afa=1}^{n}\theta_{n+1-\afa}\tilde\theta_\afa^{[\afa]}
    \right)
    \left(
      \sum_{\beta=1}^{n-1}
        \left(\Lmd^{\beta}-1\right)\left(U_\beta\theta_\beta\right)
    \right)\td x \notag \\
&\quad
  -\int\sum_{\afa,\beta=1}^{n}
    \tilde\theta_\afa
    \left(
      U_{\afa+\beta-n-1}\theta_\beta
    \right)^{[\beta-n-1]}
    \left(1-\Lmd^{-1}\right)
    \left(U_n\theta_n\right) \td x \notag \\
=&\,
  \int\sum_{\gamma=0}^{n}
      \sum_{\afa=1}^{n}
      \left(
        \sum_{\beta,\beta'=\gamma}^{n-1}
       -\sum_{\beta,\beta'=1}^{\gamma-1}
      \right) \notag\\
&\quad \times
      \left(
        U_{\gamma+\afa-n-1}\tilde\theta_\afa^{[n+1-\gamma]}
       -W\left(U_{\gamma+\afa-n-1}\theta_\afa\right)^{[\afa-n-1]}
      \right)
      \theta_\beta\left( U_{\beta+\beta'-\gamma}\theta_{\beta'}\right)^{[\beta'-\gamma]}\td x\notag\\
&\quad +\int
  \sum_{\gamma=0}^{n}\sum_{\beta=\gamma+1}^{n}
  \sum_{\afa=1}^{n-1}
  \left(
    \sum_{s=0}^{\gamma} - \sum_{s=\afa}^{\afa+\gamma}
  \right)
    \theta_{n+1+\gamma-\beta}\tilde\theta_\beta^{[\beta-\gamma]}
    U_s\left(U_{\gamma+\afa-s}\theta_\afa\right)^{[\afa-s]}\td x  \notag
\\
&\quad
  -\int\sum_{\afa,\beta=1}^{n}
    \tilde\theta_\afa
    \left(
      U_{\afa+\beta-n-1}\theta_\beta
    \right)^{[\beta-n-1]}
    \left(1-\Lmd^{-1}\right)
    \left(U_n\theta_n\right) \td x \notag \\
=&\,
  \int \sum_{\gamma=0}^{n}\sum_{\afa=1}^{n}
    \left(
      \sum_{\beta=1}^{n-1}\sum_{\beta'=\gamma}^{n-1}
     -\sum_{\beta=1}^{\gamma-1}\sum_{\beta'=1}^{n-1}
    \right)
  \underbrace{
  \tilde\theta_\afa
  \left(
    U_{\gamma+\afa-n-1}\theta_\beta
  \right)^{[\gamma-n-1]}
  \left(
    U_{\beta+\beta'-\gamma}\theta_{\beta'}
  \right)^{[\beta'-n-1]}}_{:=\mathbb F(\afa,\gamma,\beta,\beta')} \td x \notag \\
&\quad
  -\int \underbrace{
  \sum_{\gamma=0}^{n}\sum_{\afa=1}^{n}
    \left(\sum_{\beta,\beta'=\gamma}^{n-1}
          -\sum_{\beta,\beta'=1}^{\gamma-1}\right)
    \tilde\theta_\beta
    \left(U_{\beta+\beta'-\gamma}\theta_{\beta'}\right)^{[\beta'-\gamma]}
    \left(
      U_{\gamma+\afa-n-1}\theta_\afa
    \right)^{[\afa-n-1]} }_{:= Z_1}\td x \notag \\
&\quad
  +\int \underbrace{
  \sum_{\gamma=0}^{n}\sum_{\beta=\gamma+1}^{n}\sum_{\afa=1}^{n-1}
    \left(\sum_{s=0}^{\gamma}-\sum_{s=\afa}^{\afa+\gamma}\right)
    \tilde\theta_\beta
    \left(
      U_{\gamma+\afa-s}\theta_\afa
    \right)^{[\gamma+\afa-s-\beta]}
    \left(
      U_s\theta_{n+1+\gamma-\beta}
    \right)^{[\gamma-\beta]}}_{:= Z_2} \td x \notag
\\
&\quad
  -\int\sum_{\afa,\beta=1}^{n}
    \tilde\theta_\afa
    \left(
      U_{\afa+\beta-n-1}\theta_\beta
    \right)^{[\beta-n-1]}
    \left(1-\Lmd^{-1}\right)
    \left(U_n\theta_n\right) \td x. \label{260215-1854}
\end{align}

Applying the following changes of summation indices
\[
  \begin{cases}
    \hat\afa = \beta, \\
    \hat\gamma = n+1+\beta'-\gamma, \\
    \hat\beta = \beta', \\
    \hat\beta' = \afa,
  \end{cases}
\qquad
  \begin{cases}
    \tilde \afa = \beta, \\
    \tilde\gamma = n+1+\gamma+\afa-s-\beta, \\
    \tilde\beta = \afa, \\
    \tilde\beta' = n+1+\gamma-\beta
  \end{cases}
\]
to $Z_1$ and $Z_2$, respectively, in the above equation, one can verify that
\begin{align*}
  Z_1 &=\,
      \sum_{\hat\afa = 1}^{n-1}
    \left(
      \sum_{\hat\gamma = n+1}^{2n}
      \sum_{\hat\beta=1}^{\min\{n-1, \hat\gamma+\hat\afa-n-1\}}
      \sum_{\hat\beta'=1}^{n}
     -
      \sum_{\hat\gamma =0}^{n}
      \sum_{\hat\beta=\hat\gamma + \hat\afa -n}^{\hat\gamma-1}
      \sum_{\hat\beta'=1}^{n}
    \right)
    \mathbb F(\hat\afa,\hat\gamma,\hat\beta,\hat\beta'),
\\
  Z_2 &=\,
      \sum_{\tilde\afa = 1}^{n}
   \left(
      \sum_{\tilde\gamma = 0}^{2n}
      \sum_{\tilde\beta=1}^{\min\{n-1, \tilde\gamma+\tilde\afa-n-1\}}
      \sum_{\tilde\beta'=1}^{n}
     -\sum_{\tilde\gamma = 0}^{n}
      \sum_{\tilde\beta=1}^{n-1}
      \sum_{\tilde\beta'=\tilde\gamma}^{n}
   \right)\mathbb F(\tilde\afa,\tilde\gamma,\tilde\beta,\tilde\beta').
\end{align*}
Substituting these two expressions into the right-hand side of \eqref{260215-1854},
we obtain $[\tilde I_1, I_2] = 0$.
Therefore, combining it with \eqref{J-Schouten}, we have $[I_1, I_2]=0$.

\section{The associated generalized Frobenius manifold}

We are to construct a generalized Frobenius manifold structure $\mcalM_{(n,1)}$ on the space of a particular class of rational functions,
and show that the dispersionless limits of the flows of the $(n,1)$-type RR2T belong to the Principal Hierarchy of this generalized Frobenius manifold.

\subsection{Generalized Frobenius manifolds}

Let us recall some basic notions of generalized Frobenius manifolds.
\begin{defn}(\cite{GFM}). Let $\mcalM$ be a smooth or analytic manifold.
A generalized Frobenius structure on $\mcalM$ is a quadruple $(\eta,c,e,E)$,
where $e$ and $E$ are vector fields on $\mcalM$,
$\eta$ is a non-degenerate symmetric $(0,2)$-tensor on $\mcalM$,
and $c$ is a symmetric $(0,3)$-tensor on $\mcalM$, such that:
\begin{enumerate}
  \item At each point $p\in\mcalM$, $(\,\cdot\,,\eta,e)$ forms a Frobenius algebra structure on $T_p\mcalM$
  with unity $e$, where the multiplication $\cdot\,\colon T\mcalM\times T\mcalM\to T\mcalM$ is uniquely determined by
  \[
    \eta(\p'\cdot\p'', \p''') = c(\p', \p'', \p''')
  \]
  for all vector fields $\p',\p'',\p'''$ on $\mcalM$;
  \item The metric $\eta$ is flat, i.e. the Riemann curvature of $\eta$ vanishes;
  \item The $(0,4)$-tensor $\nabla c$ is symmetric, where $\nabla$ is the Levi-Civita connection of $\eta$;
  \item $\nabla\nabla E=0$, $\nabla E$ is diagonalizable, and there exists $d\in\bbC$ such that
  \begin{align*}
    \mcalL_E\left(\p'\cdot\p''\right) &=\, \left(\mcalL_E\p'\right)\cdot\p''
        + \p'\cdot\left(\mcalL_E\p''\right) + \p'\cdot\p'', \\
    \mcalL_E\left(\eta(\p', \p'')\right) &=\,
      \eta\left(\mcalL_E\p', \p''\right)
     +\eta\left(\p', \mcalL_E\p''\right)
     +(2-d)\eta\left(\p', \p''\right)
  \end{align*}
  for all vector fields $\p',\p''$ on $\mcalM$, where $\mcalL_E$ is the Lie derivative along $E$.
\end{enumerate}
\end{defn}
\noindent
The manifold $\mcalM$ equipped with such $(\eta,c,e,E)$ satisfying the above axioms is called a \textit{generalized Frobenius manifold},
the vector fields $e$, $E$ are called the unit vector field and the Euler vector field, respectively,
and the constant $d$ is called the \textit{charge} of $\mcalM(\eta,c,e,E)$.
Compared with the definition of the ordinary Dubrovin--Frobenius manifolds \cite{Dubrovin1996,Normal forms},
the only modification here is that the flatness condition $\nabla e = 0$ for the unit vector field is removed.

Let $\mcalM(\eta,c,e,E)$ be an $m$-dimensional generalized Frobenius manifold,
one can fix a family of flat local coordinates $\bsz=(z^1,\dots,z^m)$ of the metric $\eta$,
and then the $4$-symmetry of $\nabla c$ implies that there locally exists a function $F$ on $\mcalM$ such that
\[
  c\left(\pp{z^\afa}, \pp{z^\beta}, \pp{z^\gamma}\right)
=
  \frac{\p^3 F}{\p z^\afa \p z^\beta \p z^\gamma},\quad 1\leq \afa,\beta,\gamma\leq m.
\]
Such a function $F$ is called the \textit{potential} of $\mcalM(\eta,c,e,E)$,
and satisfies the Witten--Dijkgraaf--Verlinde--Verlinde (WDVV) equation.


To study the generalized Frobenius manifold associated with the RR2T of $(n,1)$-type,
let us consider the space of Landau--Ginzburg superpotentials
consisting of rational functions of the form
\begin{equation}\label{superpotential lmd p}
  \lmd(p) := \frac{p}{p-W}\left(p^n + U_1 p^{n-1} + \cdots + U_n\right),
\end{equation}
which correspond to the Lax operator $L$ \eqref{Lax L}.
The variables $(v^1,\dots,v^{n+1}):=(U_1,\dots,U_n;W)$ form local coordinates on this space.
We still use the additional notation introduced in \eqref{short notation}.
Introduce the notation $(\cdot)':=\frac{\td}{\td p}$, and then notice that
\begin{equation} \label{lmd' p}
  \lmd'(p) = \frac{\tilde\lmd(p)}{(p-W)^2},
\end{equation}
here the polynomial $\tilde\lmd(p)\in C^\infty(\mcalM)[p]$ has the form
\begin{equation}\label{lmd tilde p}
\begin{split}
\tilde\lmd(p) &:=\,
  \sum_{k=0}^{n}(n-k)U_kp^{n+1-k}
  -\sum_{k=0}^{n}(n+1-k)WU_kp^{n-k}
=
  n\prod_{s=1}^{n+1} (p-p_s),
\end{split}
\end{equation}
where $p_1,\cdots,p_{n+1}$ are all the roots of $\tilde\lmd(p)$.
Let $\mcalU$ be the open subset of this function space such that $\{p_1,\dots,p_{n+1}; W, 0\}$ are all distinct.
Then there are a $(0,2)$-tensor $\eta$ and a $(0,3)$-tensor $c$ on $\mcalU$
given by the following Landau--Ginzburg formulae
\begin{align}
  \eta\left(
    \p', \p''
  \right)
&:=\,
  \sum_{s=1}^{n+1}\Res_{p=p_s}
  \left(
    \frac{\p'\lmd(p)\,\p''\lmd(p)}{\lmd'(p)}
    \frac{\td p}{p^2}
  \right),
\label{LG eta}\\
   c\left(
    \p', \p'', \p'''
  \right)
&:=\,
  \sum_{s=1}^{n+1}\Res_{p=p_s}
  \left(
    \frac{\p'\lmd(p)\,\p''\lmd(p)\,\p'''\lmd(p)}{\lmd'(p)}
    \frac{\td p}{p^2}
  \right),
\label{LG c}
\end{align}
respectively, here $\p',\p''$ and $\p'''$ are vector fields on $\mcalU$.
We shall construct a generalized Frobenius structure $\mcalM_{(n,1)}$ on $\mcalU$
such that the metric and the Frobenius multiplication coincide with the above $\eta$ and $c$ respectively.

\begin{rmk}
The above formulae for the Landau--Ginzburg superpotential are very similar to that given by (4.1)--(4.2) and (4.5) in \cite{Yemo Wu}.
However, it turns out that the former construction yields a generalized Frobenius manifold with non-flat unity,
while the latter gives a Dubrovin–Frobenius manifold with flat unity.
\end{rmk}

\subsection{The flat metric on the space of a particular class of rational functions}
\label{subsection:flat metric}

Before constructing the generalized Frobenius structure,
we first investigate the relation between the above metric $\eta$\,\eqref{LG eta} and the Hamiltonian structure $\mcalP_1$
of the $(n,1)$-type RR2T introduced in \eqref{250812-1309-1}--\eqref{250812-1309-2}.
Note that the entries of $\mcalP_1$ has the form
\[
  \veps^{-1}\mcalP_1^{ij} =
  \left(
  g_1^{ij}\p_x + \sum_{k=1}^{n+1}\Gamma_{1;k}^{ij}v_x^k
  \right) + O(\veps), \quad 1\leq i,j\leq n+1
\]
for certain coefficients $g_1^{ij}$ and $\Gamma_{1;k}^{ij}$.
The coefficients $g_1^{ij}$ yield a contravariant metric
\begin{equation} \label{contravariant g1 from P1}
  g_1 = \sum_{i,j=1}^{n+1}g_1^{ij}\p_{v^i}\otimes\p_{v^j}
\end{equation}
on $\mcalM$. From \eqref{250812-1309-1}--\eqref{250812-1309-2} and by straightforward computation,
one can verify that the coefficients $g_1^{ij}$ of this contravariant metric
with respect to the local coordinates $(v^1,\dots,v^n; v^{n+1})=(U_1,\dots,U_n; W)$ are
\begin{equation}\label{g1 afa beta}
\left\{
\begin{split}
  g_1^{\afa\beta} &=\, (2n-\afa-\beta)U_{\afa+\beta-n} - (2n+2-\afa-\beta)WU_{\afa+\beta-n-1}, \\
  g_1^{\afa,n} &=\, g_1^{n,\afa} = (\afa-n-2)WU_{\afa-1}, \\
  g_1^{n,n} &=\, -2 WU_{n-1},\quad g_1^{n,n+1} = g_1^{n+1,n} = W, \\
  g_1^{\afa,n+1} &=\, g_1^{n+1,\afa}=g_1^{n+1,n+1} = 0
\end{split}
\right.
\end{equation}
for all indices $1\leq\afa,\beta\leq n-1$.

\begin{prop}\label{prop: eta and P1}
  The inverse of the contravariant metric $g_1$ \eqref{contravariant g1 from P1} induced by the Hamiltonian structure $\mcalP_1$
  coincides with the metric $\eta$ introduced by \eqref{LG eta}.
  Moreover, $\eta$ is non-degenerate, and the Riemann curvature tensor of $\eta$ vanishes, i.e. $\eta$ is flat.
\end{prop}

\begin{proof}
  Let $\bsg_1=(g_1^{ij})$ and $\bseta=(\eta_{ij})$ be the coefficient matrices of
  $g_1$ and $\eta$ with respect to the local coordinates $(v^1,\dots,v^n; v^{n+1})=(U_1,\dots,U_n; W)$.
  It suffices to verify that $\bsg_1\bseta = \bsI_{n+1}$.
  Now let us compute the entries $\eta_{ij}:=\eta\left(\p_{v^i}, \p_{v^j}\right)$.
  It is clear that
  \[
    \pfrac{\lmd(p)}{U_i} = \frac{p^{n+1-i}}{p-W},\quad
    \pfrac{\lmd(p)}{W} = \frac{\lmd(p)}{p-W},\quad 1\leq i\leq n.
  \]
  Introduce the following polynomials
  \begin{align}
    \hat\lmd(p)
    &:=\,
      \sum_{k=0}^{n}(n-k)U_kp^k - \sum_{k=0}^{n}(n+1-k)WU_kp^{k+1},
  \qquad 
    f(p) :=\, \sum_{k=0}^{n} U_k p^{n-k},  \label{polynomial fp}
  \end{align}
  and recall the polynomial $\tilde\lmd(p)$ introduced in \eqref{lmd tilde p}.
  Then it is clear that
  \begin{equation} \label{tilde lmd W}
    \tilde\lmd\left(\frac{1}{p}\right) = \frac{1}{p^{n+1}}\hat\lmd(p),\quad
    \tilde\lmd(W) = -Wf(W),\quad \tilde\lmd'(W) = 0.
  \end{equation}
By the residue theorem, the formula \eqref{LG eta} can be rewritten as
\begin{equation}\label{LG eta 2}
  \eta\left(
    \p', \p''
  \right)
=
  -\left(\Res_{p=p_s}+\Res_{p=\infty}\right)
  \left(
    \frac{\p'\lmd(p)\,\p''\lmd(p)}{\lmd'(p)}
    \frac{\td p}{p^2}
  \right).
\end{equation}
Therefore, for all $1\leq i,j\leq n$, it can be verified that
\begin{align}
\eta_{ij} &=\,
  \Res_{q=0}
  \left(
    \frac{q^{i+j-n-1}}{\hat\lmd(q)}\td q
  \right),
\\
\eta_{i,n+1} &=\,\eta_{n+1,i} =
  W^{n-1-i} + \Res_{q=0}
  \left(
    \frac{q^{i-n}\td q}{(1-Wq)\hat\lmd(q)}
    \sum_{k=0}^{n}U_kq^k
  \right),
\\
\eta_{n+1,n+1}
&=\,
  \frac{2 f'(W)}{W}
 +\Res_{q=0}\left(
     \frac{q^{1-n}\td q}{(1-Wq)^2\hat\lmd(q)}
     \left(
       \sum_{k=0}^{n}U_k q^k
     \right)^2
  \right),
\end{align}
here we use the substitution $p\mapsto\frac{1}{q}$ to compute the residue at $p=\infty$.
Then, for each indices $1\leq \afa,\beta\leq n-1$, we obtain
\begin{align*}
  &\, (\bsg_1\bseta)^\afa_\beta =
  \sum_{\gamma=1}^{n-1}g_1^{\afa\gamma}\eta_{\gamma\beta} + g_1^{\afa,n}\eta_{n,\beta} + g_1^{\afa,n+1}\eta_{n+1,\beta}\\
=&\,
  \sum_{k=0}^{\afa-1}(n-k)U_k
  \Res_{q=0}
  \left(
    \frac{q^{k+\beta-\afa-1}\td q}{\hat\lmd(q)}
  \right)
-
  \sum_{k=0}^{\afa-2}(n+1-k)WU_k
  \Res_{q=0}
  \left(
    \frac{q^{k+\beta-\afa}\td q}{\hat\lmd(q)}
  \right)\\
=&\,
  \Res_{q=0}
  \left(
    \hat\lmd(q)\cdot
    \frac{q^{\beta-\afa-1}\td q}{\hat\lmd(q)}
  \right)
=
  \delta^\afa_\beta.
\end{align*}
Moreover, by straightforward computation, we also have
\begin{align*}
&\,(\bsg_1\bseta)^n_{n+1} = \sum_{k=1}^{n-1}(k-n-2)WU_{k-1}\\
&\quad
  \times\left(
    W^{n-1-k} + \Res_{q=0}
    \left(
      \frac{q^{k-n}\td q}{(1-Wq)\hat\lmd(q)}
      \sum_{\ell=0}^{n}U_\ell q^\ell
    \right)
  \right) -2WU_{n-1}\times\frac 1W\\
&\quad
  +W\left(
    \frac{2f'(W)}{W}
   +\Res_{q=0}\left(
     \frac{q^{1-n}\td q}{(1-Wq)^2\hat\lmd(q)}
     \left(
       \sum_{\ell=0}^{n}U_\ell q^\ell
     \right)^2
   \right)
  \right)
\\
=&\,
  \sum_{k=1}^{n-1}
    (k-n-2)U_{k-1}W^{n-k} - 2U_{n-1} + 2f'(W) \\
&\quad
  +W\Res_{q=0}\left[
    \frac{q^{1-n}\td q}{(1-Wq)^2\hat\lmd(q)}
    \left(\sum_{\ell=0}^{n}U_\ell q^\ell\right)\right.\\
&\quad\quad \times\left.
    \left(
      (1-Wq)\sum_{k=1}^{n-1}
        (k-n-2)U_{k-1}q^{k-1}
       +\sum_{k=0}^{n}U_kq^k
    \right)
   \right]
\\
=&\,
  \sum_{k=1}^{n-1}(n-k)U_{k-1}W^{n-k}
 -W\Res_{q=0}\left(
   \frac{q^{1-n}\td q}{(1-Wq)^2\hat\lmd(q)}
   \left(\sum_{\ell=0}^{n}U_\ell q^\ell\right)
   \cdot\hat\lmd(q)
 \right) = 0,
\end{align*}
and using similar methods, it can also be verified that
\[
  (\bsg_1\bseta)^\afa_{n+1}
 =(\bsg_1\bseta)^n_{\afa} = 0,\quad 1\leq\afa\leq n-1.
\]
The details are left to the reader.
The validity of relations
\[
  (\bsg_1\bseta)^n_{n}=(\bsg_1\bseta)^{n+1}_{n+1}=1,\quad
  (\bsg_1\bseta)^{\afa}_{n} =
 (\bsg_1\bseta)^{n+1}_{\afa} = (\bsg_1\bseta)^{n+1}_{n} = 0,\quad 1\leq\afa\leq n-1
\]
are clear. Hence $\bsg_1\bseta=\bsI_{n+1}$, and therefore, $\eta$ is non-degenerate.

Finally, note that $\mcalP_1$ is known to be a Hamiltonian structure (see Sect.\,\ref{subsection P1 ham}).
By a standard result on Hamiltonian structures of hydrodynamic type,
the inverse of the corresponding contravariant metric $g_1$, i.e. $\eta$,  is flat.
\end{proof}

An alternative proof of the flatness of the metric $\eta$ can be obtained by explicitly constructing
a family of flat coordinates $\bsz=(z^1,\dots, z^n; z^{n+1})$.
If we regard the superpotential \eqref{superpotential lmd p}
as a formal Laurent series at $p=\infty$, i.e.
$\lmd(p) = p^n + O(p^{n-1})\in C^\infty(\mcalU)\fls{\frac 1p}$,
then there exists a unique formal Laurent series
$\fai(q) = q + O(1) \in C^\infty(\mcalU)\fls{\frac 1q}$ such that
\begin{equation} \label{fai q 1}
  \lmd(\fai(q)) = q^n,
\end{equation}
where $q$ is another formal variable independent of $p$.
In other words, $\fai(q)$ is the inverse series of $\lmd^{\frac 1n}(p) = p + O(1)\in C^\infty(\mcalU)\fls{\frac 1p}$.
Let us denote
\begin{equation} \label{fai q 2}
  \log\frac{q}{\fai(q)}
=
  \frac{1}{\sqrt{n}}
  \left(
    \frac{z^n}{q} + \frac{z^{n-1}}{q^2} + \cdots + \frac{z^2}{q^{n-1}}
  \right)
 +\frac{z^1}{n q^n} + O(\frac{1}{q^{n+1}}) \in C^\infty(\mcalU)\fls{\frac 1q},
\end{equation}
where $z^1,\dots,z^n\in C^\infty(\mcalU)$ are certain functions.

\begin{lem}
Let $z^1,\dots,z^n$ be the functions defined above,
$f(p)$ be the polynomial introduced in \eqref{polynomial fp},
and additionally introduce
\begin{equation} \label{z n+1 log W}
  z^{n+1} := \log W.
\end{equation}
Then the following hold true:
\begin{enumerate}
  \item $\bsz:=(z^1,\dots,z^n; z^{n+1})$ forms a family of local coordinates of $\mcalU$.
  Moreover, $\bsv=(U_1,\dots, U_n; W)$ can be represented by polynomials of $z^1,\dots, z^n$ and $\rme^{z^{n+1}}$,
  and conversely, $z^1,\dots,z^n$ can be represented by polynomials of $\bsv$.
  \item The following formulae for $\bsz$ hold true:
  \begin{align}
    z^1 &=\, -\Res_{p=\infty}\left(
      \lmd(p)\frac{\td p}{p}
    \right) = f(W), \label{flat res formula 1}\\
    z^\afa &=\, -\frac{\sqrt{n}}{n+1-\afa}
    \Res_{p=\infty}
    \left(
      \lmd^{\frac{n+1-\afa}{n}}(p)\frac{\td p}{p}
    \right),\quad
    2\leq \afa\leq n.  \label{flat res formula 2}
  \end{align}
\item The following identities hold true:
\begin{equation} \label{p f p zi}
  \pfrac{f}{z^i}(W) = \begin{cases}
    1 & i=1, \\
    0 & 2\leq i\leq n, \\
    -Wf'(W) & i=n+1.
  \end{cases}
\end{equation}
\end{enumerate}
\end{lem}

\begin{proof}
Let us rewrite \eqref{fai q 1}--\eqref{fai q 2} as follows
\begin{equation} \label{260408-1947-1}
  \fai^{n+1}(q) + \sum_{k=1}^{n}U_k\fai^{n+1-k}(q) - q^n(\fai(q)-W) = 0,
\end{equation}
\begin{equation} \label{260408-1947-2}
\fai(q) = q\exp\left(
  -\frac{1}{\sqrt{n}}
  \sum_{k=1}^{n-1}
  \frac{z^{n+1-k}}{q^k}
  -\frac{z^1}{nq^n}
\right)
\left(
  1 + O(\frac{1}{q^{n+1}})
\right),
\end{equation}
and then substitute \eqref{260408-1947-2} into \eqref{260408-1947-1}.
Notice that the $O(\frac{1}{q^{n+1}})$ in the right-hand side of \eqref{260408-1947-2} does not affect
the coefficients of $q^k$ in the left-hand side of \eqref{260408-1947-1} for $1\leq k\leq n$.
If we impose the following grading $\overline{\deg}$ generated by
\begin{equation}\label{overline deg}
\overline{\deg}\, q = \overline{\deg}\, W = 1,\quad
\overline{\deg}\, U_k = \overline{\deg}\, z^{n+1-k} = k,\quad 1\leq k\leq n,
\end{equation}
then the right-hand side of \eqref{260408-1947-2}
and the left-hand side of \eqref{260408-1947-1} are homogeneous of degree $1$ and $n+1$, respectively.
Comparing the coefficients of $q^k$ of both sides of \eqref{260408-1947-1} for $1\leq k\leq n$,  we obtain
\begin{equation}\label{260408-2015}
\begin{split}
  U_1 &=\, c_1z^n + f_1(W), \\
  U_2 &=\, c_2z^{n-1} + f_2(U_1; z^n; W), \\
  U_3 &=\, c_3z^{n-2} + f_3(U_1, U_2; z^n, z^{n-1}; W), \\
  &\cdots\\
  U_n &=\, c_nz^1 + f_n(U_1,\dots, U_{n-1}; z^n,\dots, z^2; W),
\end{split}
\end{equation}
where $\{c_k\}_{k=1}^n$ are certain non-zero constants, and $\{f_k\}_{k=1}^n$ are certain polynomials such that
\begin{equation}\label{deg fk}
  \overline{\deg}\, f_k = k,\quad 1\leq k\leq n.
\end{equation}
Therefore the relations between $\bsz$ and $\bsv$ are recursively determined by \eqref{260408-2015} and \eqref{z n+1 log W},
and hence assertion (1) of this lemma follows.
The validity of \eqref{flat res formula 1}--\eqref{flat res formula 2} is a consequence of
the \textit{Lagrange inversion formula} of the following form
\begin{align*}
&\,
  \Res_{q=\infty}\left(
    q^{k-1}\log\frac{q}{\fai(q)}\td q
  \right)
 =
  \frac 1k\Res_{p=\infty}
  \left(
    \lmd^{\frac kn}(p)\frac{\td p}{p}
  \right)
\end{align*}
for each $k\in\bbZ_+$. Finally, differentiating both sides of \eqref{fai q 1} and \eqref{fai q 2}, we obtain
\begin{equation}
\lmd'(\fai(q))\fai'(q) = nq^{n-1},  \label{dif lmd 1}
\end{equation}
\begin{equation}
\pfrac{\lmd}{z^i}(\fai(q)) + \lmd'(\fai(q))\pfrac{\fai}{z^i}(q) = 0,\quad 1\leq i\leq n+1,  \label{dif lmd 2}
\end{equation}
\begin{equation}
  \frac{1}{\fai(q)}\pfrac{\fai}{z^i}(q) = \frac{d_i}{\sqrt{n}}\frac{1}{q^{n+1-i}} + O(\frac{1}{q^{n+1}}),\quad \label{dif lmd 3}
  \text{where}\, \,
  d_i:=\begin{cases}
    -1& 2\leq i\leq n, \\
    -\frac{1}{\sqrt{n}} & i=1, \\
    0 & i=n+1.
  \end{cases}
\end{equation}
From $\lmd(p)=\frac{p}{p-W}f(p)$ and $W=\rme^{z^{n+1}}$, we have
\begin{align}
  \pfrac{\lmd}{z^i}(p) &=\, \frac{p}{p-W}\pfrac{f}{z^i}(p),\quad 1\leq i\leq n, \label{dif f 1}\\
  \pfrac{\lmd}{z^{n+1}}(p) &=\, \frac{p}{p-W}\pfrac{f}{z^{n+1}}(p) + \frac{pW}{(p-W)^2}f(p).  \label{dif f 2}
\end{align}
Then for each $1\leq i\leq n$, from \eqref{dif lmd 1}--\eqref{dif f 2} we obtain
\begin{align*}
&\,\pfrac{f}{z^i}(W) = \Res_{p=W}
\left(\pfrac{f}{z^i}(p)\, \frac{\td p}{p-W}\right)
=
  \Res_{p=W}
  \left(
    \pfrac{\lmd}{z^i}(p)\, \frac{\td p}{p}
  \right)\\
=&\,
  -\Res_{p=\infty}
  \left(
    \pfrac{\lmd}{z^i}(p)\, \frac{\td p}{p}
  \right)
=
  -\Res_{q=\infty}\left(
    \pfrac{\lmd}{z^i}(\fai(q))\frac{\fai'(q)}{\fai(q)}\td q
  \right)\\
=&\,
  \sqrt{n}\, d_i\Res_{q=\infty}
  \left[
    \left(
      \frac{1}{q^{n+1-i}} + O(\frac{1}{q^{n+1}})
    \right) q^{n-1}\td q
  \right] = \delta_{i,1}.
\end{align*}
Similarly, we also have
\begin{align*}
&\,
  \pfrac{f}{z^{n+1}}(W)
=
  \Res_{p=W}
  \left(
    \pfrac{\lmd}{z^{n+1}}(p)\, \frac{\td p}{p}
  \right)
 -\Res_{p=W}\left(
   \frac{Wf(p)}{(p-W)^2}\td p
 \right)
=
  -Wf'(W).
\end{align*}
Hence the lemma is proved.
\end{proof}

\begin{thm}
  The metric $\eta$ introduced by \eqref{LG eta} has the local expression
  \begin{equation}\label{eta local z}
     \eta = \sum_{i=1}^{n+1}\td z^i\otimes\td z^{n+2-i}
  \end{equation}
  with respect to the coordinates $\bsz$ in the above lemma.
  In particular, $\eta$ is flat, and $\bsz$ is a family of its flat coordinates.
\end{thm}

\begin{proof}
    Using the formulae \eqref{dif lmd 1}--\eqref{dif lmd 3}, we have
    \begin{align}
      &\,\Res_{p=\infty}
      \left(
        \frac{\pfrac{\lmd}{z^i}(p)\, \pfrac{\lmd}{z^j}(p)}{\lmd'(p)}
        \frac{\td p}{p^2}
      \right)
    =
     \Res_{q=\infty}
     \left(
        \frac{\pfrac{\lmd}{z^i}(\fai(q))\, \pfrac{\lmd}{z^j}(\fai(q))}{\lmd'(\fai(q))}
        \frac{\fai'(q)\td q}{(\fai(q))^2}
      \right)\notag\\
  =&\,
    d_id_j\Res_{q=\infty}
    \left(
      \left(\frac{1}{q^{n+1-i}}+O(\frac{1}{q^{n+1}})\right)
      \left(\frac{1}{q^{n+1-j}}+O(\frac{1}{q^{n+1}})\right)
      q^{n-1}\td q
    \right) \notag \\
  =&\,
    \begin{cases}
     -1 & \text{if $i+j=n+2$ and $i,j\leq n$,}\\
     0 & \text{other cases},
    \end{cases}
    \quad \text{for all}\,\,1\leq i,j\leq n+1.   \label{260409-1146-1}
    \end{align}
  By using \eqref{lmd' p}, \eqref{tilde lmd W},
  \eqref{dif f 1}--\eqref{dif f 2} and \eqref{p f p zi}, we have
  \begin{align}
    \Res_{p=W}
    \left(
      \frac{\pfrac{\lmd}{z^i}(p)\, \pfrac{\lmd}{z^j}(p)}{\lmd'(p)}\frac{\td p}{p^2}
    \right)
  &=\,
    \Res_{p=W}\left(
    \pfrac{f}{z^i}(p)\, \pfrac{f}{z^j}(p)\frac{\td p}{\tilde\lmd(p)}
    \right) = 0,\quad 1\leq i,j\leq n,
  \\
    \Res_{p=W}
    \left(
      \frac{\pfrac{\lmd}{z^i}(p)\, \pfrac{\lmd}{z^{n+1}}(p)}{\lmd'(p)}\frac{\td p}{p^2}
    \right)
  &=\,\frac{Wf(W)\pfrac{f}{z^i}(W)}{\tilde\lmd(W)} = -\delta_{i,1},\quad 1\leq i\leq n,
  \end{align}
  \begin{align}
  &\,\Res_{p=W}
    \left(
      \frac{\pfrac{\lmd}{z^{n+1}}(p)\, \pfrac{\lmd}{z^{n+1}}(p)}{\lmd'(p)}\frac{\td p}{p^2}
    \right)
\notag\\
=&\,
  2\Res_{p=W}
  \left(
    \frac{Wf(p)\,\pfrac{f}{z^{n+1}}(p)}{\tilde\lmd(p)}
    \frac{\td p}{p-W}
  \right)
 +\Res_{p=W}
  \left(
    \frac{W^2(f(p))^2}{\tilde\lmd(p)}
    \frac{\td p}{(p-W)^2}
  \right)   \notag
\\
=&\,
  \frac{2Wf(W)}{\tilde\lmd(W)}\pfrac{f}{z^{n+1}}(W)
 +W^2\left.\frac{\td}{\td p}\right|_{p=W}
 \left(\frac{(f(p))^2}{\tilde\lmd(p)}\right) = 0.\label{260409-1146-4}
  \end{align}
Therefore, the theorem follows from \eqref{LG eta 2}
and \eqref{260409-1146-1}--\eqref{260409-1146-4}.
\end{proof}

\subsection{The unit vector field and the Euler vector field}
\label{subsection:unit and Euler}

Let us construct the generalized Frobenius manifold structure.
We define the multiplication on the tangent space of $\mcalU$ by the relation
\begin{equation}\label{Frob mult}
  \eta(\p'\cdot\p'', \p''') = c(\p', \p'',\p''')
\end{equation}
for all vector fields $\p',\p'',\p'''$ on $\mcalU$, where $\eta$ and $c$ are introduced in \eqref{LG eta} and \eqref{LG c} respectively.
It is clear that $\cdot$ is commutative and compatible with $\eta$, i.e.
\begin{equation}
  \p'\cdot\p'' = \p''\cdot\p',\quad
  \eta(\p'\cdot\p'', \p''') = \eta(\p',\p''\cdot\p''').
\end{equation}

Recall the $\{p_i\}_{1\leq i\leq n+1}$ in \eqref{lmd tilde p},
which are all the distinct critical points of $\lmd(p)$. 
Let $u^i := \lmd(p_i)$ be the corresponding critical values. 
Then $\bsu := (u^1,\dots,u^{n+1})$ forms a family of local coordinates on $\mcalU$,
which are referred to as \textit{canonical coordinates}. Note that
\begin{align}
  \delta_{ij} &=\, \pfrac{u^i}{u^j} = \pp{u^j}\left(\lmd(p_i)\right)
  =
    \pfrac{\lmd}{u^j}(p_i),\quad 1\leq i,j\leq n+1,
\end{align}
because $\lmd'(p_i) = 0$. From \eqref{superpotential lmd p} and \eqref{polynomial fp}, we obtain
\begin{equation}\label{260414-2006}
  \pfrac{\lmd}{u^i}(p) = \frac{p}{(p-W)^2}
  \left(
    (p-W)\pfrac{f}{u^i}(p) + \pfrac{W}{u^i}f(p)
  \right).
\end{equation}
By using the \textit{Lagrange interpolation formula} and \eqref{lmd' p}--\eqref{lmd tilde p} we obtain
\begin{align}
  \pfrac{\lmd}{u^i}(p) &=\, \frac{p}{p_i}
  \left(
    \frac{p_i-W}{p-W}
  \right)^2
  \prod_{j\neq i}
  \frac{p-p_j}{p_i-p_j}
=
    \frac{p}{p_i(p-p_i)}
    \frac{\lmd'(p)}{\lmd''(p_i)}.  \label{260414-1733-2}
\end{align}
Then from \eqref{LG eta}--\eqref{LG c} and \eqref{260414-1733-2}, it can be verified directly that
\begin{align}
  \eta\left(
    \p_{u^i}, \p_{u^j}
  \right) &=\, \frac{\delta_{ij}}{p_i^2\lmd''(p_i)}, \quad
  c\left(
    \p_{u^i}, \p_{u^j}, \p_{u^k}
  \right)
=
  \frac{\delta_{ij}\delta_{ik}}{p_i^2\lmd''(p_i)}.  \label{canonical eta and c}
\end{align}
Therefore, the multiplication $\cdot$ satisfies
$
  \p_{u^i}\cdot \p_{u^j} = \delta_{ij}\p_{u^i}
$.
In particular, $\cdot$ is associative, i.e. $(\p'\cdot\p'')\cdot\p'''=\p'\cdot(\p''\cdot\p''')$. 
Introduce the following vector fields
\begin{equation}\label{e and E}
  e:=\sum_{i=1}^{n+1}\pp{u^i},\qquad
  E:=\sum_{i=1}^{n+1} u^i\pp{u^i}.
\end{equation}
It is clear that $e$ is the unity of the multiplication $\cdot$,
i.e. $e\cdot\p' = \p'$ for all vector field $\p'$.

\begin{prop} Let $e$, $E$ be the vector fields introduced above.
\begin{enumerate}
  \item The vector field $e$ has the following local expression
  \begin{equation}\label{e v-coord}
    e = \grad_\eta\log\frac{W}{U_n} = \sum_{\afa=1}^{n}(n+2-\afa)\frac{WU_{\afa-1}}{U_n}\pp{U_\afa} + \pp{U_n} - \frac{W}{U_n}\pp{W},
  \end{equation}
  here $\grad_\eta$ is the gradient with respect to the metric $\eta$.
  \item The vector field $E$ has the following local expression
  \begin{align}
    E&=\,\sum_{k=1}^{n}\frac{k}{n}U_k\pp{U_k} + \frac 1n W\pp{W} 
     =\,\sum_{k=1}^{n}\frac{n+1-k}{n}z^k\pp{z^k} + \frac 1n\pp{z^{n+1}},  \label{E z-coord}
  \end{align}
  where $\bsz=(z^1,\dots,z^{n+1})$ is the basis of flat coordinates introduced in \eqref{fai q 2}--\eqref{z n+1 log W}.
\end{enumerate}
\end{prop}

\begin{proof}
Note that \eqref{260414-2006}--\eqref{260414-1733-2} imply
\begin{equation}\label{260414-2007}
  \pfrac{W}{u^i}f(p) + (p-W)\pfrac{f}{u^i}(p) =
  \frac{n}{p_i\lmd''(p_i)}\prod_{j\neq i}(p-p_j).
\end{equation}
  Applying $p\mapsto 0$ to both sides of \eqref{lmd tilde p} and \eqref{260414-2007}, we obtain
  \[
    U_n\pfrac{W}{u^i}-W\pfrac{U_n}{u^i} = \frac{WU_n}{p_i^2\lmd''(p_i)},\quad 1\leq i\leq n+1.
  \]
Therefore $\pp{u^i}\log\frac{W}{U_n} = \frac{1}{p_i^2\lmd''(p_i)}$, and \eqref{canonical eta and c} implies
$e=\grad_\eta\log\frac{W}{U_n}$.
Moreover, the gradient of $\log\frac{W}{U_n}$ can also be computed in $\bsv$-coordinates
by using Proposition \ref{prop: eta and P1} and \eqref{g1 afa beta}.
Hence \eqref{e v-coord} holds true.
Multiplying $u^i=\lmd(p_i)=\frac{p_if(p_i)}{p_i-W}$ to both sides of \eqref{260414-2007},
and then summing it over $i$ from $1$ to $n+1$, we obtain
\[
 (\mcalL_EW)f(p) + (p-W)\mcalL_Ef(p)
=
  \sum_{i=1}^{n+1}
    (p_i-W)f(p_i)
    \prod_{j\neq i}\frac{p-p_j}{p_i-p_j},
\]
where $\mcalL_E$ is the Lie derivative along $E$.
Applying $p\mapsto p_i$ for $1\leq i\leq n+1$ to both sides of the above equation,
it is clear that the following polynomial
\begin{align}\label{gp-1}
  g(p)&:=\, (\mcalL_EW)f(p) + (p-W)\mcalL_Ef(p) - (p-W)f(p) \notag\\
  &\phantom:=\,
    -p^{n+1} + \sum_{k=0}^{n}
    \Big(
      (\mcalL_E-1)U_{n+1-k}-W(\mcalL_E-1)U_{n-k} + U_{n-k}\mcalL_EW
    \Big)p^k
\end{align}
has $n+1$ distinct roots $p_1,\dots, p_{n+1}$. Together with \eqref{lmd tilde p}, we obtain
\begin{align}\label{gp-2}
  &\, g(p) = - \prod_{i=1}^{n+1}(p-p_i)
  =
    -p^{n+1} + \frac1n\sum_{k=0}^{n}
    \Big(
      (k+1)WU_{n-k} - (k-1)U_{n+1-k}
    \Big)p^k.
\end{align}
The coefficients of $p^k$, $0\leq k\leq n$, in right-hand sides of \eqref{gp-1}--\eqref{gp-2}
yield a linear system with respect to unknown functions $\{(\mcalL_E-1)U_k\}_{1\leq k\leq n}$ and $\mcalL_EW$.
And then we obtain
\[
  \mcalL_EW = \frac 1n W,\quad
  \mcalL_EU_k = \frac{k}{n}U_k,\quad 1\leq k\leq n
\]
by solving this system. 
Thus we obtain the expression of the vector field $E$ in the local coordinates $\bsv$.
Finally, \eqref{z n+1 log W} and the degree argument \eqref{overline deg}--\eqref{deg fk} imply
\[
\mcalL_E z^{n+1} = \frac 1n,\quad
\mcalL_E z^{\afa} = \frac{n+1-\afa}{n}z^{\afa},\quad 1\leq \afa\leq n.
\]
Therefore \eqref{E z-coord} holds true, and the proposition is proved.
\end{proof}

\begin{thm}\label{thm: n1 GFM}
 The quadruple $\mcalM_{(n,1)}:=(\eta,c,e,E)$
 induced by \eqref{LG eta}--\eqref{LG c} and \eqref{e and E}
 is a generalized Frobenius manifold structure of charge $d=1$.
\end{thm}

\begin{proof}
  We have already known that $\eta$ is flat,
  and $(\eta,\cdot,e)$ forms a Frobenius algebra structure on tangent spaces,
  where $\cdot$ is the multiplication induced by \eqref{Frob mult}.
  Since the unity $e=\grad_\eta\log\frac{W}{U_n}$ is a gradient field,
  by Lemma 2.12 in \cite{JiangLR GFM 2}, the $(0,4)$-tensor $\nabla c$ is $4$-symmetric,
  where $\nabla$ is the Levi-Civita connection with respect to $\eta$.

  From the local expression \eqref{E z-coord} of $E$ with respect to the flat coordinates $\bsz$,
  it is clear that $\nabla\nabla E = 0$ and $\nabla E$ is diagonalizable.
  Under the canonical coordinates $\bsu$,
  it is easy to verify that
  \[
    \mcalL_E(\p'\cdot\p'') = (\mcalL_E\p')\cdot\p'' + \p'\cdot(\mcalL_E\p'') + \p'\cdot\p''
  \]
  for all vector fields $\p',\p''$ by using \eqref{e and E}.
  Finally, by the degree argument \eqref{overline deg}--\eqref{deg fk} and \eqref{LG eta}, \eqref{E z-coord},
  it is clear that
  \[
    \mcalL_E\Big(
      \eta(\p_{v^i}, \p_{v^j})
    \Big)
   =
    \left(
      1-\frac{\overline{\deg}\,v^i + \overline{\deg}\,v^j}{n}
    \right) \eta(\p_{v^i},\p_{v^j}),\quad 1\leq i,j\leq n+1,
  \]
  and therefore $\mcalL_E\eta = \eta$. Hence the charge $d=1$ and the theorem is proved.
\end{proof}

In particular, the coefficient matrix $\mu$ of the operator $\left(\frac{2-d}{2}\boldsymbol{1}-\nabla E\right)$
with respect to the flat coordinates $\bsz$ has the form
\begin{equation}
  \mu=\diag(\mu_1,\cdots,\mu_{n+1}), \quad
  \mu_k:=\frac{k-1}{n}-\frac12,\quad 1\leq k\leq n+1,
\end{equation}
which belongs to the monodromy data of the generalized Frobenius manifold $\mcalM_{(n,1)}$.
We remark that the entries $\mu_2,\dots,\mu_{n+1}$ defined above coincide with that defined in \eqref{P2 mu}.

\subsection{Relation to the dispersionless limit of the $(n,1)$-type RR2T}
\label{subsection:PH}

In this subsection,
we are to provide the \textit{Principal Hierarchy} of the generalized Frobenius manifold $\mcalM_{(n,1)}$
introduced in Theorem \ref{thm: n1 GFM},
and show that the dispersionless limits of the flows of the $(n,1)$-type RR2T \eqref{positive-Lax}--\eqref{negative-Lax}
belong to this Principal Hierarchy.

Let us briefly recall the notion of
the Principal Hierarchy of a generalized Frobenius manifold in the sense of \cite{GFM}.
In general, let $\mcalM(\eta,c,e,E)$ be an $m$-dimensional generalized Frobenius manifold,
fix a family of flat coordinates $(z^1,\ldots,z^m)$.
As it was introduced in \cite{GFM},
the \textit{Principal Hierarchy} of $\mcalM$ is a family of evolutionary PDEs of the form
\begin{equation}\label{PH}
  \frac{\partial z^\alpha}{\partial t^{i,k}}
=
  \sum_{\beta=1}^{m}
  \eta^{\alpha\beta} \partial_x \frac{\partial \theta_{i,k+1}}{\partial z^\beta},
  \quad (i,k) \in \mathcal{I},\quad 1\leq \afa\leq m
\end{equation}
with respective to the time variables $\mathbf{t} = \{t^{i,k}\}_{(i,k) \in \mathcal{I}}$,
where the index set
\begin{align*}
  \mathcal{I} = \left( \{1,2,\ldots,m\} \times \mathbb{Z}_{\geq 0} \right) \cup \left( \{0\} \times \mathbb{Z} \right),
\end{align*}
$(\eta^{\afa\beta}):=(\eta_{\afa\beta})^{-1}$, $\eta_{\afa\beta}:=\eta(\p_{z^\afa}, \p_{z^\beta})$,
and $\theta_{i,k}$ are certain functions on $\mcalM$.
We remark that the definition of the flows $\pp{t^{\afa,k}}$ in \eqref{PH} for $1\leq \afa\leq m$, $k\geq 0$
is analogous to that for Dubrovin--Frobenius manifolds with flat unity \cite{Dubrovin1996, Normal forms},
and the flows $\{\pp{t^{0,k}}\}_{k\in\bbZ}$,
referred to as the additional flows, arise from the non-flatness of the unit vector field.
In particular, the $\pp{t^{0,0}}$-flow is given by the translation along the spatial variable $x$,
i.e. $\pp{t^{0,0}}=\p_x$, and therefore we identify the time variable $t^{0,0}$ with $x$ for convenience.
The flows of the Principal Hierarchy \eqref{PH} can be rewritten as the following Hamiltonian system of hydrodynamic type
\begin{equation}
\frac{\partial z^\alpha}{\partial t^{i,k}}
=
  \sum_{\beta=1}^{m}
(\mathcal{P}_1^{[0]})^{\alpha\beta} \frac{\delta H_{i,k}^{[0]}}{\delta z^\beta}, \quad (i,k) \in \mathcal{I},
\end{equation}
where the Hamiltonian operator $\mathcal{P}_1^{[0]}$ has the form
$
(\mathcal{P}_1^{[0]})^{\alpha\beta} = \eta^{\alpha\beta} \partial_x
$
with respect to the flat coordinates,
and the Hamiltonians
\begin{align*}
 H_{i,k}^{[0]} = \int \theta_{i,k+1} \td x, \quad (i,k) \in \mathcal{I}.
\end{align*}
In other words, the Hamiltonian operator $\mathcal{P}_1^{[0]}$ arises from the flat metric $\eta$,
and the functions $\theta_{i,k}$ serve as Hamiltonian densities.
In fact, the Principal Hierarchy \eqref{PH} admits a bihamiltonian structure
$(\mathcal{P}_1^{[0]},\mathcal{P}_2^{[0]})$,
and satisfies certain bihamiltonian recursion relations,
where the second operator $\mathcal{P}_2^{[0]}$ arises from the \textit{intersection form}
of $M$.
More details about the Principal Hierarchy of generalized Frobenius manifolds can be found in \cite{GFM}.

Now let $m:=n+1$, consider the $(n+1)$-dimensional generalized Frobenius manifold $\mcalM_{(n,1)}$
introduced in Theorem \ref{thm: n1 GFM},
and fix its flat coordinates $\bsz=(z^1,\dots,z^{n+1})$ defined by \eqref{fai q 2}--\eqref{flat res formula 2}.
Then the entries of the coefficient matrix of the flat metric $\eta$ with respect to the coordinates $\bsz$
has the form $\eta_{\afa\beta}=\delta_{\afa+\beta,n+2}$ because of \eqref{eta local z}.

To construct the Principal hierarchy of $\mcalM_{(n,1)}$,
it suffices to give the explicit expression of the densities $\{\theta_{i,k}\}_{(i,k)\in\mcalI}$.
By standard methods (for instance, those analogous to \cite{Kodama LG, MaSL LG} or Corollary 2.24 of \cite{rational-reduction}),
it can be verified that a subset of the Hamiltonian densities
$\theta_{i,k}$ corresponding to $\mcalM_{(n,1)}$ can be chosen as
\begin{align}
  \theta_{i,k} &=\,
-c_{i,k}\Res_{p=\infty}
  \left(
    \lmd^{k+\frac{i-1}{n}}(p)
    \frac{\td p}{p}
  \right),\quad \text{for}\,\, 2\leq i\leq n+1,\, k\geq 0,
\\
  \theta_{0,-k} &=\, (-1)^k(k-1)!
  \Res_{p=0}
  \left(
    \lmd^{-k}(p)\frac{\td p}{p}
  \right),\quad\text{for}\,\,k\geq 1, \label{theta 0-k}
\end{align}
while the explicit expressions for the remaining $\{\theta_{1,k}, \theta_{0,k}\}_{k\geq 0}$ will be given later.
Here the constants $c_{i,k}$ are introduced in \eqref{const coef cik},
and $\lmd^{\frac 1n}(p)\in C^\infty(\mcalU)\fls{\frac 1p}$ is regarded as a formal Laurent series.
It is clear that the above $\theta_{i,k}$, $\theta_{0,-k}$
coincide with the dispersionless limits
of the densities of the Hamiltonians $H_{i,k-1}$, $H_{0,-k-1}$ introduced in \eqref{Ham Hip-1}--\eqref{Ham Hip-2}.
Proposition \ref{prop: eta and P1} implies that the first Hamiltonian structure $\mcalP^{[0]}_1$
of the Principal Hierarchy of $\mcalM_{(n,1)}$
coincides with the dispersionless limit of $\mcalP_1$ \eqref{HamP1-eqn}--\eqref{250812-1309-2}
of the $(n,1)$-type RR2T.
Hence we have the following result:

\begin{prop}
  The dispersionless limits of the flows \eqref{positive-Lax}--\eqref{negative-Lax}
  of the $(n,1)$-type RR2T coincide with the $\pp{t^{i,k}}$- and $\pp{t^{0,-k-1}}$-flows
  of the Principal Hierarchy of $\mcalM_{(n,1)}$, respectively, for $2\leq i\leq n+1$ and $k\geq 0$.
\end{prop}

There are two sequences of remaining flows
$\pp{t^{1,k}}$ and $\pp{t^{0,k}}$ for $k \geq 0$
in the Principal Hierarchy of $\mcalM_{(n,1)}$,
which shall be regarded as the \textit{extended logarithmic flows} of the dispersionless $(n,1)$-type RR2T.
To provide the explicit expression of the corresponding densities $\theta_{1,k}$ and $\theta_{0,k}$,
we need to expand the logarithm of the superpotential $\lmd(p)$ at $p=W$ in two different ways.
Now we regard $\lmd(p)$ as a polynomial of $(p-W)$ and $\frac{1}{p-W}$ as follows
\begin{equation}
  \lmd(p) = \frac{p f(p)}{p-W}
=
  (p-W)^n + \sum_{\ell=-1}^{n-1} \lmd_\ell (p-W)^\ell,
\end{equation}
where $f(p)$ is introduced in \eqref{polynomial fp}, and $\lmd_\ell$ are certain functions.
From \eqref{z n+1 log W}--\eqref{flat res formula 1} we obtain
\begin{equation}
  \lmd_{-1} = Wf(W) = z^1\rme^{z^{n+1}}.
\end{equation}
\begin{itemize}
  \item One way to expand the logarithm of $\lmd(p)$
  is analogous to the construction of the logarithmic flows of the extended bi-graded Toda hierarchy
  \cite{extend bi Toda, extended Toda},
  or the extended constrained KP hierarchy, see Sect.\,5 of \cite{central inv of cKP}.
  Precisely speaking, let us expand the logarithm of $(p-W)\lmd(p)$ and
  $\frac{\lmd(p)}{(p-W)^{n}}$
  as formal power series of $(p-W)$ and $\frac{1}{p-W}$ respectively, i.e.
  introduce the following formal series
  \begin{align*}
    \log_{+}\big((p-W)\lmd(p)\big)
  &:=\,
    \left(z^{n+1} + \log z^1\right) +
    \log\left(\frac{(p-W)\lmd(p)}{\lmd_{-1}}\right) \in C^\infty(\mcalU)\fps{p-W},
  \\
    \log_-\left(\frac{\lmd(p)}{(p-W)^{n}}\right)
  &\in\,  C^\infty(\mcalU)\fps{\frac{1}{p-W}}.
  \end{align*}
  Then we introduce the following bi-infinite formal series
  \begin{align}
    \log(\lmd(p)) &:=\, \frac{n}{n+1} \log_{+}\big((p-W)\lmd(p)\big) +
    \frac{1}{n+1} \log_-\left(\frac{\lmd(p)}{(p-W)^{n}}\right) .
  \end{align}
  \item Another way to expand the logarithm of $\lmd(p)$ arises from the factorization
  \[
    \lmd(p) = p^n \cdot \frac{\lmd(p)}{p^n}.
  \]
  Introduce the following formal series
  \begin{align*}
    \log_+(p) &:= \log W + \log\left(1+\frac{p-W}{W}\right)
    =
      z^{n+1} + \sum_{k=1}^{\infty}(-1)^{k-1}\frac{(p-W)^k}{kW^k},
  \end{align*}
  and $\log_-\left(\frac{\lmd(p)}{p^{n}}\right)\in C^\infty(\mcalU)\fps{\frac{1}{p-W}}$.
  Then introduce the bi-infinite formal series
  \begin{equation}\label{tilde log}
    \widetilde{\log}(\lmd(p)):=
     n \log_+(p) + \log_-\left(\frac{\lmd(p)}{p^n}\right).
  \end{equation}
\end{itemize}

It can be verified that the densities $\theta_{1,k}$, $\theta_{0,k}$
of 
$\mcalM_{(n,1)}$  can be chosen as
\begin{align}
  \theta_{1,0} &=\, z^{n+1},\qquad  \theta_{0,0} = \log\frac{W}{U_n},
\\
  \theta_{1,k} &=\,
\frac{n+1}{n k!}
\Res_{p=W}
\left(
  \lmd^k(p)\left(\log(\lmd(p)) - H_k\right)\frac{\td p}{p}
\right),\quad k\geq 1,
\\
  \theta_{0,k} &=\,
\frac{1}{n k!}
\Res_{p=W}
\left(
  \lmd^k(p)\left(\widetilde\log(\lmd(p)) - H_k\right)\frac{\td p}{p}
\right), \quad k\geq 1,  \label{theta 0k}
\end{align}
where $H_k:=\sum_{i=1}^{k}\frac 1i$ are harmonic numbers,
which generalizes the results of the special cases for $n=1$
(Proposition 2.1 of \cite{extend AL})  and $n=2$ (Proposition 3.2 of \cite{2-1 RR2T}).

\begin{ex}
  The densities $\theta_{0,\pm 1}$ defined by \eqref{theta 0-k} and \eqref{theta 0k} have the forms
  \begin{equation} \label{theta 01}
    \theta_{0,-1} = \frac12\eta(e,e) = -\frac{U_n + W U_{n-1}}{U_n^2},\qquad
    \theta_{0,1} = \frac12\sum_{i=1}^{n+1}z^i z^{n+2-i},
  \end{equation}
where $\eta$ and $e$ are the Frobenius metric and unit vector field, respectively.
\end{ex}
Note that equation \eqref{theta 01} agrees with the general result for generalized Frobenius manifolds;
see (3.15) of \cite{GFM}.

\begin{proof}
  From \eqref{e v-coord} we arrive at
  \begin{align*}
    \frac12\eta(e,e) &=\, \frac12\eta\left(e \, , \, \grad_\eta\log\frac{W}{U_n}\right)
    =
      \frac{1}{2} \mcalL_e \left(\log\frac{W}{U_n}\right) = - \frac{U_n + W U_{n-1}}{U_n^2}.
  \end{align*}
  Then, by a direct computation from \eqref{theta 0-k} with $k=1$, one can verify the validity of
  the first equation of \eqref{theta 01}. 
  To deal with $\theta_{0,1}$, let us expand the logarithm of $\frac{\lmd(p)}{p^n}$
  as formal power series in $\frac{1}{p}$ and in $\frac{1}{p-W}$, respectively, as follows:
  \[
    \widehat{\log}\left(\frac{\lmd(p)}{p^n}\right)
    :=\sum_{k =1}^{\infty}\frac{a_k^{(0)}}{p^k},
  \quad
    \log_-\left(\frac{\lmd(p)}{p^n}\right)
    := \sum_{k=1}^{\infty}\frac{a_k^{(1)}}{(p-W)^k},
  \]
  where $a_k^{(0)}, a_k^{(1)}$ are certain functions on $\mcalU$.
  It is clear that
  \begin{equation}\label{250425-1915}
    \left(\sum_{k=1}^{n}\frac{a_k^{(0)}}{p^k}\right)
   - \left(\sum_{k=1}^{n}\frac{a_k^{(1)}}{(p-W)^k}\right) = O(\frac{1}{p^{n+1}}),\quad p\to\infty.
  \end{equation}
  By using \eqref{fai q 1}--\eqref{fai q 2} and \eqref{polynomial fp} we obtain
  \begin{align}
  &\,
    \frac12\sum_{i=2}^{n} z^i z^{n+2-i}
  =
    -\frac n2 \Res_{q=\infty}
    \left(
      \left(\log\frac{q}{\fai(q)}\right)^2 q^{n-1}\td q
    \right) \notag \\
  =&\,
    \Res_{q=\infty}\left(
      q^n\left(\log\frac{q}{\fai (q)}\right)
      \left(\frac 1q - \frac{\fai'(q)}{\fai(q)}\right)\td q
    \right)
  =
  -\frac{z^1}{n}
  -\Res_{q=\infty}
  \left(
    q^n\left(\log\frac{q}{\fai(q)}\right)
    \frac{\td\fai(q)}{\fai(q)}
  \right)  \notag\\
  =&\,
    -\frac{z^1}{n}
    -\frac 1n\Res_{p=\infty}
    \left(
      \lmd(p) \, \widehat{\log}\left(\frac{\lmd(p)}{p^n}\right)\frac{\td p}{p}
    \right)
  =
    -\frac{z^1}{n}
    -\frac 1n\Res_{p=\infty}
    \left(
      \frac{f(p)}{p-W}
      \sum_{k=1}^{n}
        \frac{a_k^{(0)}}{p^k}
        \td p \notag
    \right).
  \end{align}
Then from \eqref{tilde log}, \eqref{theta 0k}, \eqref{250425-1915} and the above equation,
we arrive at
\begin{align*}
  \theta_{0,1}
&=\,
  \frac 1n\Res_{p=W}\left(
    \lmd(p)
    \left(
      \widetilde{\log}(\lmd(p)) - 1
    \right)
    \frac{\td p}{p}
  \right)
\\
&=\,
  \frac 1n \Res_{p=W}
  \left(
    \frac{f(p)}{p-W}
    \left(
      nz^{n+1} + \log_-\left(\frac{\lmd(p)}{p^n}\right) - 1
    \right) \td p
  \right)
\\
&=\,
  \left(z^{n+1}-\frac 1n\right)z^1
 -\frac 1n\Res_{p=\infty}
  \left(
    \frac{f(p)}{p-W}
    \sum_{k=1}^{n}
      \frac{a_k^{(0)}}{p^k}
     \td p
  \right)
=
\frac 12
  \sum_{i=1}^{n+1}
    z^i z^{n+2-i}.
\end{align*}
Hence \eqref{theta 01} holds true.
\end{proof}

\begin{rmk}\label{rmk: extended n1 RR2T}
  We conjecture that the $(n,1)$-type RR2T \eqref{positive-Lax}--\eqref{negative-Lax} admits
  two sequences of extended flows
  $\{\pp{t^{1,k}}, \pp{t^{0,k}}\}_{k\geq 0}$
  such that the dispersionless limits of them coincide with that in the Principal Hierarchy of $\mcalM_{(n,1)}$.
  The derivation of explicit expressions for these extended flows
  is likely to depend on appropriate definitions of the logarithm of the Lax operator $L$ \eqref{Lax L}.
  However, we do not yet know how to define such $\log L$.
\end{rmk}

\subsection{Examples}

  Recall that all the generalized Frobenius manifolds $\mcalM_{(n,1)}$
  constructed in Sect.\,\ref{subsection:flat metric}--\ref{subsection:unit and Euler}
  are of charge $d=1$.
  The flat metrics $\eta$ and the Euler vector fields $E$
  with respect to the flat coordinates $\bsz$ \eqref{z n+1 log W}--\eqref{flat res formula 2}
  are given by \eqref{eta local z} and \eqref{E z-coord}, respectively.

  Let us provide explicit expressions of the potential $F$ and the unit vector field $e$
  of $\mcalM_{(n,1)}$ with respect to the flat coordinates $\bsz$ for $1\leq n\leq 4$.

\begin{ex}(\cite{Brini 2012,LiuQuZhang,extend AL}).
The flat coordinates $\bsz=(z^1,z^2)$ of $\mcalM_{(1,1)}$ read
\[
  z^1 = U_1 + W,\quad z^2 = \log W.
\]
In this coordinates, the potential $F_{(1,1)}$ of $\mcalM_{(1,1)}$ has the form
\begin{equation}
  F_{(1,1)} = \frac12(z^1)^2 \left(z^2 + \log z^1 \right) + z^1\rme^{z^2},
\end{equation}
and the unit vector field $
  e = \frac{z^1\p_{z^1} - \p_{z^2}}{z^1 - \rme^{z^2}}$.
\end{ex}

\begin{ex}(\cite{2-1 RR2T}). The flat coordinates $\bsz=(z^1,z^2,z^3)$ of $\mcalM_{(2,1)}$ are given by
\begin{align*}
\left\{
	\begin{aligned}
		z^1 &=  U_2 + U_1 W + W^2, \\
		z^2 &=  \frac{1}{\sqrt{2}}\left(U_1 + W\right),  \quad
		z^3 = \log W.
	\end{aligned}
\right.
\end{align*}
In this coordinates, the potential $F_{(2,1)}$ of $\mcalM_{(2,1)}$ has the form
\begin{equation}\label{F(21)}
  F_{(2,1)}=\frac12(z^1)^2 \left(z^3 + \log z^1\right)
   +z^1\left(\frac12 (z^2)^2
   +\sqrt{2}\,z^2 \rme^{z^3}
   +\frac12 \rme^{2z^3}\right)
   -\frac1{24} (z^2)^4,
\end{equation}
and the unit vector field $e$ has the form
\[
  e
=
  \frac{
    z^1\p_{z^1}
   +\sqrt{2}\,\rme^{z^3}\p_{z^2}
   -\p_{z^3}
  }{z^1-\sqrt{2}\,z^2\rme^{z^3}}.
\]
\end{ex}

\begin{ex} The flat coordinates $\bsz=(z^1,z^2,z^3,z^4)$ of $\mcalM_{(3,1)}$ are given by
\begin{align*}
\left\{
	\begin{aligned}
		z^1 &=  U_3 + U_2 W + U_1 W^2 + W^3, \\
		z^2 &=  \frac{1}{\sqrt{3}}\left(U_2 - \frac16 U_1^2 + \frac 23 U_1 W + \frac 56 W^2\right),  \\
        z^3 &=  \frac{1}{\sqrt{3}}\left(U_1 + W\right),\quad
		z^4 =  \log W.
	\end{aligned}
\right.
\end{align*}
In this coordinates, the potential $F_{(3,1)}$ of $\mcalM_{(3,1)}$ has the form
\begin{align}
  F_{(3,1)}&=\,\frac12(z^1)^2 \left(z^4 + \log z^1\right) \notag\\
   &\quad
    +z^1\left(
       z^2 z^3 + \sqrt{3}\, z^2 \rme^{z^4}
      +\frac12(z^3)^2 \rme^{z^4}
      +\frac{\sqrt{3}}{2}\, z^3 \rme^{2z^4}
      +\frac 13\, \rme^{3z^4}    \right) \notag\\
   &\quad
     +\frac{(z^2)^3}{2\sqrt{3}}
     -\frac {(z^2)^2 (z^3)^2}4
     +\frac{z^2 (z^3)^4 }{24\sqrt{3}}
     -\frac{(z^3)^6}{720}, \label{F(31)}
\end{align}
and the unit vector field $e$ has the form
\[
  e
=
  \frac{
    z^1\p_{z^1} + z^3\rme^{z^4}\p_{z^2} + \sqrt{3}\,\rme^{z^4}\p_{z^3} - \p_{z^4}
  }{z^1-\left(\sqrt{3}\,z^2 + \frac12 (z^3)^2\right)\rme^{z^4}}.
\]
\end{ex}

\begin{ex}
The flat coordinates $\bsz=(z^1,\dots,z^5)$ of $\mcalM_{(4,1)}$ are given by
\begin{align*}
\left\{
	\begin{aligned}
		z^1 &=  U_4 + U_3 W + U_2 W^2 + U_1 W^3 + W^4, \\
		z^2 &=  \frac 12 U_3 - \frac18 U_2 U_1 + \frac 38 U_2 W + \frac 5{192} U_1^3 - \frac{3}{64}U_1^2 W
               +\frac{21}{64}U_1 W^2 + \frac{77}{192}W^3,  \\
        z^3 &=  \frac 12 U_2 -\frac 18 U_1^2 + \frac 14 U_1W +\frac 38 W^2, \\
        z^4 &=  \frac{1}{2}\left(U_1 + W\right),\quad
		z^5 =  \log W.
	\end{aligned}
\right.
\end{align*}
In this coordinates, the potential $F_{(4,1)}$ of $\mcalM_{(4,1)}$ has the form
\begin{align}
  F_{(4,1)}&=\, \frac12(z^1)^2 \left(z^5 + \log z^1\right) \notag\\
   &\quad
     + z^1 \left(
         z^2 z^4 + 2z^2\rme^{z^5} + \frac{(z^3)^2}{2} + z^3 z^4 \rme^{z^5} + z^3\rme^{2z^5} \right.\notag\\
   &\qquad\left.
        +\frac{(z^4)^3\rme^{z^5}}{12} + \frac{(z^4)^2\rme^{2z^5}}{2}
        +\frac{2z^4\rme^{3z^5}}{3} + \frac{\rme^{4z^5}}{4}
       \right) \notag\\
   &\quad
     + (z^2)^2z^3 - \frac{(z^2)^2(z^4)^2}{4}
     - \frac{z^2 (z^3)^2 z^4}{2} + \frac{z^2 z^3 (z^4)^3}{12} - \frac{z^2 (z^4)^5}{480} \notag\\
   &\qquad
     - \frac{(z^3)^4}{12} + \frac{(z^3)^3(z^4)^2}{12}
     - \frac{(z^3)^2 (z^4)^4}{48}
     + \frac{z^3 (z^4)^6}{576}
     - \frac{(z^4)^8}{16128},  \label{F(41)}
\end{align}
and the unit vector field $e$ has the form
\[
  e =
    \frac{
      z^1\p_{z^1}
     +\left(z^3 + \frac14(z^4)^2\right)\rme^{z^5}\p_{z^2}
     +z^4\rme^{z^5}\p_{z^3}
     +2\rme^{z^5}\p_{z^4} - \p_{z^5}
    }
    {z^1 - \left(2z^2 + z^3 z^4 + \frac{1}{12}(z^4)^3\right)\rme^{z^5}}.
\]
\end{ex}

At the end of this section,
we provide some preliminary evidence for a relation between the above $\mcalM_{(n,1)}$
and the generalized Frobenius manifold $\mcalM_{\mathsf{A}_{n-1}}$
associated with the $q$-deformed Gelfand--Dickey hierarchy \cite{Zhonglun Cao}.
Here, $\mcalM_{\mathsf{A}_{n-1}}$ is the generalized Frobenius manifold
defined on the orbit space of the affine Weyl group of type $\mathsf{A}_{n-1}$ in the framework of \cite{JiangLR GFM 2},
which can be constructed from the Landau--Ginzburg superpotential
\[
  \lmd_{\mathrm{qGD}}(p) = p^n + U_1 p^{n-1} + \cdots + U_{n-1} p
\]
corresponding to the Lax operator \eqref{Lax qGD};
see Sect.\,2.6 of \cite{JiangLR GFM 2} for more details.

For $n=2,3,4$,
after the change of variables
\[
  z^\afa \mapsto \frac{1}{\sqrt{n}} t^{n+1-\afa},\quad 2\leq \afa\leq n
\]
in the potentials \eqref{F(21)}--\eqref{F(41)},
one finds that the limits
$\lim\limits_{z^1\to 0} F_{(n,1)}$ coincide with the potentials of $\mcalM_{\mathsf{A}_{n-1}}$ appearing in
Examples 2.1--2.3 of \cite{JiangLR GFM 2}.
Based on this observation,
we expect that there are some relations between
$\mcalM_{(n,1)}$ and $\mcalM_{\mathsf{A}_{n-1}}$ for general $n\ge 2$,
which will be investigated in future work.

\section{Conclusion}

In this paper,
we derived a pair of bihamiltonian operators $(\mcalP_1, \mcalP_2)$ for the $(n,1)$-type RR2T,
and proved, by a straightforward computation using the super-variable formalism,
that the pair forms a bihamiltonian structure.
We believe that the method developed in this work based on the super-variable formalism
can be extended to a wide range of differential-difference integrable systems
and may provide a systematic way to investigate their Hamiltonian structures.

On the other hand, from a geometric perspective,
we constructed an $(n+1)$-dimensional generalized Frobenius manifold $\mcalM_{(n,1)}$
with non-flat unity \cite{GFM}
associated with the $(n,1)$-type RR2T,
and proved that the dispersionless flows of the $(n,1)$-type RR2T belong to the Principal Hierarchy of this generalized Frobenius manifold.
Based on the study of the special case $n=2$ in \cite{2-1 RR2T},
we expect that the $(n,1)$-type RR2T is related to the $(n,1)$-type bi-graded Toda hierarchy \cite{extend bi Toda}
and the constrained KP hierarchy \cite{central inv of cKP}
via linear reciprocal transformations,
while the corresponding (generalized) Frobenius manifolds are related via generalized Legendre transformations \cite{LiuQuZhang}.
Moreover, based on several computed examples,
we conjecture that the central invariants \cite{Central inv}
of the bihamiltonian structure of the $(n,1)$-type RR2T are all equal to $\frac 1{24}$,
which provides supporting evidence for the conjecture that
the flows of $(n,1)$-type RR2T
coincide, up to a certain Miura-type transformation,
with those of the topological deformation \cite{GFM}
of the Principal Hierarchy of $\mcalM_{(n,1)}$.
These conjectures will be investigated in future work.

Several related problems also deserve further study.
For example, it would be interesting to investigate whether the generalized Frobenius manifold $\mcalM_{(n,1)}$
arises from an orbit-space construction similar to that in \cite{JiangLR GFM 2},
and to extend the methods developed in this paper to study the Hamiltonian structures of
Takasaki's generalized Ablowitz--Ladik hierarchy \cite{Takasaki}.

\vskip 0.3cm
\noindent \textbf{Acknowledgements.}
This work was supported by National Natural Science Foundation of China (Grant No. 12501323).
The authors would like to thank Zhonglun Cao, Si-Qi Liu and Youjin Zhang for very useful discussions on this work.
The second author was supported by the  “Jingying” Project of Shandong University of Science and Technology.

\section*{Appendix: Proof of Lemma \ref{CZL main lemma}}
\addcontentsline{toc}{section}{Appendix: Proof of Lemma \ref{CZL main lemma}}
\renewcommand{\thesection}{A}

In this Appendix, we shall verify \eqref{CZL main identity} by straightforward calculation.
In comparison with equation \eqref{def of A gamma}--\eqref{def of B gamma}, we additionally introduce
\begin{align*}
  A_0 &:=\, \sum_{\afa=1}^{n-1}
    \left(
      \sum_{s=0}^{0} - \sum_{s=\afa}^{\afa}
    \right) U_s\left(U_{\afa-s}\theta_\afa\right)^{[\afa-s]}
  =
    \sum_{\afa=1}^{n-1}\left(\Lmd^\afa-1\right)\left(U_\afa\theta_\afa\right),
\\
  B_0 &:=\,
    \left(
      \sum_{\beta,\beta'=0}^{n-1} - \sum_{\beta,\beta'=1}^{-1}
    \right)\theta_\beta
    \left(
      U_{\beta+\beta'}\theta_{\beta'}
    \right)^{[\beta']}
=
  \sum_{\beta,\beta'=1}^{n-1}
    \theta_\beta
    \left(
      U_{\beta+\beta'}\theta_{\beta'}
    \right)^{[\beta']},
\end{align*}
then \eqref{CZL main identity} is equivalent to
\begin{equation}\label{CZL main identity 2}
  \int
    \sum_{\gamma=0}^{n-1} A_\gamma B_\gamma \td x = 0.
\end{equation}
It should be emphasized that we always use the notational conventions
\[
  U_0:=1,\quad \theta_0:= 0, \quad
  U_\gamma:=0\ \ \text{if $\gamma<0$ or $\gamma>n$. }
\]

Notice that for all $0\leq\gamma\leq n-1$, $A_\gamma$ and $B_\gamma$ can be rewritten as
\begin{align}
  A_\gamma
&=\,
  \underbrace{
  \sum_{s=0}^{\gamma}U_s
  \sum_{\afa=s+1}^{n-1}
    \left(
      U_{\gamma+\afa-s}\theta_\afa
    \right)^{[\afa-s]}}_{:=\,A_\gamma^{(0)}}
 \,-\,
   \underbrace{
   \sum_{s=\gamma+1}^{n}U_s
   \sum_{\afa=s-\gamma}^{\min\{s,n-1\}}
    \left(
      U_{\gamma+\afa-s}\theta_\afa
    \right)^{[\afa-s]}}_{:=\,A_\gamma^{(1)}},
\\
  B_\gamma
&=\,
  \underbrace{
  \sum_{\beta=\gamma}^{n-1}\theta_\beta
    \sum_{\beta'=\gamma}^{n-1}
      \left(U_{\beta+\beta'-\gamma}\theta_{\beta'}\right)^{[\beta'-\gamma]}
  }_{:=\,B_\gamma^{(0)}}
\,-\,
  \underbrace{
  \sum_{\beta=1}^{\gamma-1}\theta_\beta
    \sum_{\beta'=\gamma-\beta}^{\gamma-1}
      \left(U_{\beta+\beta'-\gamma}\theta_{\beta'}\right)^{[\beta'-\gamma]}
  }_{:=\,B_\gamma^{(1)}},
\end{align}
and then we have
\begin{align}
&\,
  \int\sum_{\gamma=0}^{n-1}
    A_\gamma^{(0)}
    B_\gamma^{(0)} \td x \notag\\
=&\,
  -\int
    \sum_{\beta=1}^{n-1}
    \sum_{s=0}^{\beta}
      U_s\theta_\beta
    \sum_{\gamma=s}^{\beta}
    \sum_{\afa=s+1}^{n-1}
    \sum_{\beta'=\gamma}^{n-1}
      \left(U_{\gamma+\afa-s}\theta_\afa\right)^{[\afa-s]}
      \left(U_{\beta+\beta'-\gamma}\theta_{\beta'}\right)^{[\beta'-\gamma]}
    \td x \notag\\
=&\,
  -\int
    \sum_{\beta=1}^{n-1}
    \sum_{s=0}^{\beta-1}
      U_s\theta_\beta
    \sum_{\gamma=s}^{\beta}
    \sum_{\afa=s+1}^{n-1}
    \sum_{\beta'=\gamma}^{n-1}
      \left(U_{\gamma+\afa-s}\theta_\afa\right)^{[\afa-s]}
      \left(U_{\beta+\beta'-\gamma}\theta_{\beta'}\right)^{[\beta'-\gamma]}
    \td x \notag\\
&\quad
  -\int
    \underbrace{
    \sum_{\beta=1}^{n-1}
    \sum_{\afa=\beta+1}^{n-1}
    \sum_{\beta'=\beta}^{n-1}
      \left(U_\beta\theta_\beta\right)^{[\beta]}
      \left(U_\afa\theta_\afa\right)^{[\afa]}
      \left(U_{\beta'}\theta_{\beta'}\right)^{[\beta']}
    }_{=0}
    \td x,
\label{A0 B0}\\
&\,
  \int \sum_{\gamma=0}^{n-1}
    A_\gamma^{(1)} B_\gamma^{(1)} \td x \notag\\
=&\,
  -\int
    \sum_{\beta=1}^{n-1}
    \sum_{s=\beta+2}^{n}
      U_s\theta_\beta
    \sum_{\gamma=\beta+1}^{s-1}
    \sum_{\afa=s-\gamma}^{\min\{s,n-1\}}
    \sum_{\beta'=\gamma-\beta}^{\gamma-1}
      \left(U_{\gamma+\afa-s}\theta_\afa\right)^{[\afa-s]}
      \left(U_{\beta+\beta'-\gamma}\theta_{\beta'}\right)^{[\beta'-\gamma]}
    \td x \notag\\
=&\,
  -\int
    \sum_{\beta=1}^{n-1}
    \sum_{s=\beta+2}^{n}
      U_s\theta_\beta
    \sum_{\gamma=\beta+1}^{s-1}
    \sum_{\afa=s-\gamma}^{s-1}
    \sum_{\beta'=\gamma-\beta}^{\gamma-1}
      \left(U_{\gamma+\afa-s}\theta_\afa\right)^{[\afa-s]}
      \left(U_{\beta+\beta'-\gamma}\theta_{\beta'}\right)^{[\beta'-\gamma]}
    \td x \notag\\
&\quad
  -\int
    \sum_{\beta=1}^{n-1}
    \sum_{s=\beta+2}^{n-1}
      U_s\theta_\beta
    \sum_{\gamma=\beta+1}^{s-1}
    \sum_{\beta'=\gamma-\beta}^{\gamma-1}
      U_\gamma\theta_s \cdot
      \left(U_{\beta+\beta'-\gamma}\theta_{\beta'}\right)^{[\beta'-\gamma]}
    \td x.
\label{A1 B1}
\end{align}

By using the following changes of summation indices
\begin{equation}\label{CZL 5-changes}
\begin{cases}
  \tilde\beta := \beta', \\
  \tilde s := \beta+\beta'-\gamma, \\
  \tilde\gamma := \beta'+s-\gamma, \\
  \tilde\afa := \beta, \\
  \tilde\beta' := \afa,
\end{cases}
\qquad
\begin{cases}
  \hat\beta := \afa, \\
  \hat s := \gamma+\afa-s, \\
  \hat\gamma := \afa+\beta-s, \\
  \hat\afa := \beta', \\
  \hat\beta' := \beta,
\end{cases}
\end{equation}
one can verify by careful calculations that
\begin{align}
&\,
  \int \sum_{\gamma=0}^{n-1}
    A_\gamma^{(0)} B_\gamma^{(1)} \td x \notag\\
=&\,
  -\int
    \sum_{\gamma=0}^{n-1}
    \sum_{\beta=1}^{\gamma-1}
    \sum_{s=0}^{\gamma}
    \sum_{\afa=s+1}^{n-1}
    \sum_{\beta'=\gamma-\beta}^{\gamma-1}
      U_s\theta_\beta \cdot
      \left(U_{\gamma+\afa-s}\theta_\afa\right)^{[\afa-s]}
      \left(U_{\beta+\beta'-\gamma}\theta_{\beta'}\right)^{[\beta'-\gamma]} \td x \notag\\
=&\,
  -\int
    \sum_{\tilde\beta=1}^{n-1}
    \sum_{\tilde s=0}^{\tilde\beta-1}
      U_{\tilde s}\theta_{\tilde\beta}
    \sum_{\tilde\gamma=-\infty}^{\tilde\beta}
    \sum_{\tilde\afa=\tilde s+1}^{n-1}
    \sum_{\tilde\beta' = \tilde\gamma+\tilde\afa-\tilde s+1}^{n-1}
      \left(U_{\tilde\gamma+\tilde\afa-\tilde s}\theta_{\tilde \afa}\right)^{[\tilde\afa- \tilde s]}
      \left(U_{\tilde\beta+\tilde\beta'-\tilde\gamma}\theta_{\tilde\beta'}\right)^{[\tilde\beta'-\tilde\gamma]} \td x,
\label{A0 B1}\\
&\,
  \int \sum_{\gamma=0}^{n-1}
    A_\gamma^{(1)} B_\gamma^{(0)} \td x \notag\\
=&\,
  -\int
    \sum_{\gamma=0}^{n-1}
    \sum_{\beta=\gamma}^{n-1}
    \sum_{s=\gamma+1}^{n}
    \sum_{\afa=s-\gamma}^{\min\{s,n-1\}}
    \sum_{\beta'=\gamma}^{n-1}
      U_s\theta_\beta \cdot
      \left(U_{\gamma+\afa-s}\theta_\afa\right)^{[\afa-s]}
      \left(U_{\beta+\beta'-\gamma}\theta_{\beta'}\right)^{[\beta'-\gamma]} \td x \notag\\
=&\,
  -\int
    \sum_{\hat\beta=1}^{n-1}
    \sum_{\hat s=0}^{\hat\beta-1}
      U_{\hat s}\theta_{\hat\beta}
    \sum_{\hat\gamma=\hat s}^{n-1}
    \sum_{\hat\afa=\hat s + 1}^{n-1}
    \sum_{\hat\beta' = \hat\gamma}^{\min\{\hat\gamma+\hat\afa-\hat s, n-1\}}
      \left(U_{\hat\gamma+\hat\afa-\hat s}\theta_{\hat \afa}\right)^{[\hat\afa- \hat s]}
      \left(U_{\hat\beta+\hat\beta'-\hat\gamma}\theta_{\hat\beta'}\right)^{[\hat\beta'-\hat\gamma]} \td x. 
\label{A1 B0}
\end{align}
Then following from \eqref{A0 B0}, \eqref{A0 B1}--\eqref{A1 B0}, and using \eqref{CZL 5-changes} again, we obtain
\begin{align}
&\,
  \int\sum_{\gamma=0}^{n-1}
    \left(
      A_\gamma^{(0)} B_\gamma^{(0)}
     -A_\gamma^{(0)} B_\gamma^{(1)}
     -A_\gamma^{(1)} B_\gamma^{(0)}
    \right)\td x \notag\\
=&\,
  \int
    \sum_{\beta=1}^{n-1}
    \sum_{s=0}^{\beta-1}
      U_{s}\theta_{\beta}
    \sum_{\gamma=-\infty}^{s-1}
    \sum_{\afa=s+1}^{n-1}
    \sum_{\beta' = \gamma+\afa-s+1}^{n-1}
      \left(U_{\gamma+\afa- s}\theta_{\afa}\right)^{[\afa-s]}
      \left(U_{\beta+\beta'-\gamma}\theta_{\beta'}\right)^{[\beta'-\gamma]} \td x \notag \\
&\quad
  +\int
    \sum_{\beta=1}^{n-1}
    \sum_{s=0}^{\beta-1}
      U_{s}\theta_{\beta}
    \sum_{\gamma=\beta+1}^{n-1}
    \sum_{\afa=s+1}^{n-1}
    \sum_{\beta' = \gamma}^{\min\{\gamma+\afa-s,n-1\}}
      \left(U_{\gamma+\afa- s}\theta_{\afa}\right)^{[\afa-s]}
      \left(U_{\beta+\beta'-\gamma}\theta_{\beta'}\right)^{[\beta'-\gamma]} \td x \notag \\
=&\,
    \int
    \sum_{\beta=1}^{n-1}
    \sum_{s=0}^{\beta-1}
    \sum_{\gamma=-\infty}^{s-1}
    \sum_{\afa=s+1}^{n-1}
    \sum_{\beta' = \gamma+\afa-s+1}^{n-1}
      \left(U_{\beta+\beta'-\gamma}\theta_{\beta'}\right)
      \left(U_{s}\theta_{\beta}\right)^{[\gamma-\beta']}
      \left(U_{\gamma+\afa- s}\theta_{\afa}\right)^{[\gamma+\afa-s-\beta']}
    \td x \notag \\
&\quad
  +\int
    \sum_{\beta=1}^{n-1}
    \sum_{s=0}^{\beta-1}
    \sum_{\gamma=\beta+1}^{n-1}
    \sum_{\afa=s+1}^{n-1}
    \sum_{\beta' = \gamma}^{\min\{\gamma+\afa-s,n-1\}}
      \left(U_{\gamma+\afa- s}\theta_{\afa}\right)
      \left(U_{\beta+\beta'-\gamma}\theta_{\beta'}\right)^{[\beta'-\gamma-\afa+s]}
      \left(U_{s}\theta_{\beta}\right)^{[s-\afa]}
    \td x \notag \\
=&\,
  \int
    \sum_{\tilde\beta=1}^{n-1}
    \sum_{\tilde s = \tilde\beta + 2}^{n}
      U_{\tilde s}\theta_{\tilde\beta}
    \sum_{\tilde\gamma=\tilde\beta+1}^{\tilde s - 1}
    \sum_{\tilde\afa=\tilde s - \tilde\gamma}^{\tilde s - 1}
    \sum_{\tilde\beta' = \tilde\gamma+\tilde\afa-\tilde s + 1}^{\tilde\gamma - 1}
      \left(U_{\tilde\gamma+\tilde\afa-\tilde s}\theta_{\tilde \afa}\right)^{[\tilde\afa- \tilde s]}
      \left(U_{\tilde\beta+\tilde\beta'-\tilde\gamma}\theta_{\tilde\beta'}\right)^{[\tilde\beta'-\tilde\gamma]} \td x \notag \\
&\quad
   +\int
    \sum_{\hat\beta=1}^{n-1}
    \sum_{\hat s= \hat\beta + 2}^{n}
      U_{\hat s}\theta_{\hat\beta}
    \sum_{\hat\gamma=\hat \beta + 1}^{\hat s -1}
    \sum_{\hat\afa=\hat s - \hat\gamma}^{\hat s}
    \sum_{\hat\beta' = \hat\gamma - \hat\beta}^{\min\{\hat\gamma+\hat\afa-\hat s, \hat\gamma-1\}}
      \left(U_{\hat\gamma+\hat\afa-\hat s}\theta_{\hat \afa}\right)^{[\hat\afa- \hat s]}
      \left(U_{\hat\beta+\hat\beta'-\hat\gamma}\theta_{\hat\beta'}\right)^{[\hat\beta'-\hat\gamma]} \td x \notag \\
=&\,
   \int
    \sum_{\beta=1}^{n-1}
    \sum_{s=\beta+2}^{n}
      U_s\theta_\beta
    \sum_{\gamma=\beta+1}^{s-1}
    \sum_{\afa=s-\gamma}^{s-1}
    \sum_{\beta'=\gamma-\beta}^{\gamma-1}
      \left(U_{\gamma+\afa-s}\theta_\afa\right)^{[\afa-s]}
      \left(U_{\beta+\beta'-\gamma}\theta_{\beta'}\right)^{[\beta'-\gamma]}
    \td x \notag\\
&\quad
  +\int
    \sum_{\beta=1}^{n-1}
    \sum_{s=\beta+2}^{n-1}
      U_s\theta_\beta
    \sum_{\gamma=\beta+1}^{s-1}
    \sum_{\beta'=\gamma-\beta}^{\gamma-1}
      U_\gamma\theta_s \cdot
      \left(U_{\beta+\beta'-\gamma}\theta_{\beta'}\right)^{[\beta'-\gamma]}
    \td x. \label{260118-1615}
\end{align}
Therefore, \eqref{A1 B1} and \eqref{260118-1615} imply
\begin{align*}
&\,
  \int\sum_{\gamma=0}^{n-1}
    A_\gamma B_\gamma \td x
=
  \int
    \sum_{\gamma=0}^{n-1}
      \left(
        A_\gamma^{(0)}B_\gamma^{(0)}
       -A_\gamma^{(0)}B_\gamma^{(1)}
       -A_\gamma^{(1)}B_\gamma^{(0)}
       +A_\gamma^{(1)}B_\gamma^{(1)}
      \right)
    \td x = 0,
\end{align*}
hence \eqref{CZL main identity 2} holds true,
and Lemma \ref{CZL main lemma} has been proved.

\end{document}